\documentclass[11pt]{article}
\usepackage[pagebackref,colorlinks=true,citecolor=blue]{hyperref}
\usepackage{amsmath} % AMS Math Package
\usepackage{amsthm} % Theorem Formatting
\usepackage{thmtools}
\usepackage{amssymb}	% Math symbols such as \mathbb
\usepackage{graphicx} % Allows for eps images
\usepackage{multicol} % Allows for multiple columns
\usepackage{multirow}
\usepackage{color}
\usepackage[dvips,letterpaper,margin=1in,bottom=1in]{geometry}
\usepackage[capitalize]{cleveref}



\usepackage[utf8]{inputenc}
\usepackage[english]{babel}
\usepackage{mathtools}

\usepackage{interval}
\intervalconfig{soft open fences}

\newtheorem{theorem}{Theorem}[section]

\newtheorem{lemma}[theorem]{Lemma}
\newtheorem{corollary}[theorem]{Corollary}
\newtheorem{proposition}[theorem]{Proposition}

\newtheorem{definition}[theorem]{Definition}

\newcommand{\braket}[2]{\left< #1 \vphantom{#2} \middle| #2 \vphantom{#1} \right>} % for Dirac brackets
\newcommand{\ketbra}[2]{\ensuremath{\ket{#1}\!\bra{#2}}}

\DeclarePairedDelimiter\rbra{\lparen}{\rparen}
\DeclarePairedDelimiter\sbra{\lbrack}{\rbrack}
\DeclarePairedDelimiter\cbra{\{}{\}}
\DeclarePairedDelimiter\abs{\lvert}{\rvert}
\DeclarePairedDelimiter\Abs{\lVert}{\rVert}
\DeclarePairedDelimiter\ceil{\lceil}{\rceil}

\DeclarePairedDelimiter\ket{\lvert}{\rangle}
\DeclarePairedDelimiter\bra{\langle}{\rvert}
\DeclarePairedDelimiter\ave{\langle}{\rangle}

\newcommand{\tr} {\operatorname{tr}}
\newcommand{\poly} {\operatorname{poly}}
\newcommand{\diag} {\operatorname{diag}}

\newcommand{\rank} {\operatorname{rank}}

\newcommand{\spec} {\operatorname{spec}}

\newcommand{\ran} {\operatorname{ran}}

\crefname{ineq}{inequality}{inequalities}
\creflabelformat{ineq}{#2{\upshape(#1)}#3}
\usepackage{enumitem}

\newcommand{\ignore}[1]{}

\DeclareMathOperator*{\Prob}{\mathbf{Pr}}

\newcommand{\R}{\mathbb R}
\newcommand{\C}{\mathbb C}
\newcommand{\N}{\mathbb N}

\newcommand{\BQP}{\mathsf{BQP}}
\newcommand{\QSZK} {\mathsf{QSZK}}

\newcommand{\eps}{\varepsilon}

\newcommand{\calD}{\mathcal{D}}

\newcommand{\calH}{\mathcal{H}}
\newcommand{\calI}{\mathcal{I}}

\newcommand{\calL}{\mathcal{L}}

\newcommand{\calO}{\mathcal{O}}

\DeclarePairedDelimiterX\diverg[2]{(}{)}{#1 \,\|\, #2}
\newcommand{\dtv}[2]{\mathrm{d}_{\mathrm{TV}}(#1,#2)}
\newcommand{\Fid}[2]{\mathrm{F}(#1,#2)}
\newcommand{\dH}[3]{\mathrm{d}_{\mathrm{H}}^{#3}\rbra*{#1,#2}}
\newcommand{\dAa}[4]{\mathrm{A}_{#1}^{#4}\rbra*{#2,#3}}
\newcommand{\dA}[3]{\mathrm{A}^{#3}\rbra*{#1,#2}}

\newcommand{\dtr}[3]{\mathrm{d}_{\mathrm{tr}}^{#3}\rbra*{#1,#2}}
\newcommand{\dRena}[4]{\mathrm{D}_{\textnormal{R\'{e}n},#1}^{#4}\diverg{#2}{#3}}

\newcommand{\dTsa}[4]{\mathrm{D}_{\textnormal{Tsa},#1}^{#4}\diverg{#2}{#3}}
\newcommand{\df}[3]{\mathrm{D}_{f}^{#3}\diverg{#1}{#2}}

\newcommand{\QSD}{\textsc{QSD}}
\newcommand{\PureQSD}{\textsc{PureQSD}}

\newcommand{\TsaQSD}{\textsc{TsallisQSD}}
\newcommand{\TsaPureQSD}{\textsc{TsallisPureQSD}}
\newcommand{\HellQSD}{\textsc{HellingerQSD}}
\newcommand{\HellPureQSD}{\textsc{HellingerPureQSD}}
\newcommand{\TsaLowQSD}{\textsc{TsallisLowRankQSD}}
\newcommand{\HellLowQSD}{\textsc{HellingerLowRankQSD}}

\usepackage{algorithm}
\usepackage{algpseudocode}
\usepackage{tabularx}
\usepackage{booktabs}
\usepackage{threeparttable}

\usepackage{adjustbox}

\usepackage{footnotehyper} 
\makesavenoteenv{table}  

\newcommand{\footremember}[2]{%
    \footnote{#2}
    \newcounter{#1}
    \setcounter{#1}{\value{footnote}}%
}

\usepackage{tikz}
\begin{document}

\title{On Estimating the Quantum Tsallis Relative Entropy}

\author{Jinge Bao\footremember{1}{School of Informatics, University of Edinburgh. \href{mailto:jingebao1011@gmail.com}{\nolinkurl{jingebao1011@gmail.com}}} \and
Minbo Gao\footremember{2}{Institute of Software, CAS. \href{mailto:gaomb@ios.ac.cn}{\nolinkurl{gaomb@ios.ac.cn}} or
\href{mailto:gmb17@tsinghua.org.cn}{\nolinkurl{gmb17@tsinghua.org.cn}}.} \and
Qisheng Wang\footremember{3}{School of Computer Science, Shanghai Jiao Tong University. \href{mailto:QishengWang1994@gmail.com}{\nolinkurl{QishengWang1994@gmail.com}}.}}

\date{}

\maketitle

\begin{abstract}
    The relative entropy between quantum states quantifies their distinguishability. 
    The estimation of certain relative entropies has been investigated in the literature, e.g., the von Neumann relative entropy and sandwiched R{\'e}nyi relative entropy. 
    In this paper, we present a comprehensive study of the estimation of the quantum Tsallis relative entropy. 
    We show that for any constant $\alpha \in (0, 1)$, the $\alpha$-Tsallis relative entropy between two quantum states of rank $r$ can be estimated with sample complexity $\operatorname{poly}(r)$, which can be made more efficient if we know their state-preparation circuits. 
    As an application, we obtain an approach to tolerant quantum state certification with respect to the quantum Hellinger distance with sample complexity $\widetilde{O}(r^{3.5})$, which \textit{exponentially} outperforms the folklore approach based on quantum state tomography when $r$ is polynomial in the number of qubits. 
    In addition, we show that the quantum state distinguishability problems with respect to the quantum $\alpha$-Tsallis relative entropy and quantum Hellinger distance are $\mathsf{QSZK}$-complete in a certain regime, and they are $\mathsf{BQP}$-complete in the low-rank case. 
\end{abstract}

\newpage

\tableofcontents

\newpage

\section{Introduction}
Measuring the distinguishability between quantum states is a fundamental problem in quantum information theory, with applications in, e.g., quantum state discrimination \cite{Che00,BC09,BK15} and quantum property testing \cite{MdW16}.
Distinguishability measures of quantum states include relative entropies (cf.\ \cite{Weh78,OP93,Ved02}, e.g., the von Neumann relative entropy \cite{Ume62}, the Petz-R\'enyi relative entropy \cite{Pet86,Ren61}, the sandwiched R\'enyi relative entropy \cite{WWY14,MDS+13}), the Bures distance \cite{Bur69} and Uhlmann fidelity \cite{Uhl76,Joz94}, the trace distance \cite{Rus94}, and the Hilbert-Schmidt distance \cite{Oza00}. 
There have been approaches to estimating these distinguishability measures in the literature. 
The Hilbert-Schmidt distance, also known as the quantum $\ell_2$ distance, can be directly estimated by the SWAP test \cite{BCWdW01}.
Efficient quantum algorithms for estimating the fidelity (and Bures distance) and the trace distance (also known as the quantum $\ell_1$ distance) were recently developed in \cite{WZC+23,WGL+24,GP22,WZ23,Wan24,LWWZ25,WZ24a,FW25,UNWT25} for the low-rank and pure cases. 
Quantum algorithms for estimating the quantum $\ell_\alpha$ distance for $\alpha > 1$ were developed in \cite{WGL+24,LW25b}. 
The estimation of the von Neumann relative entropy was demonstrated in \cite{Hay25} based on the Schur transform \cite{BCH06}. 
The estimation of the sandwiched R{\'e}nyi relative entropy was considered in \cite{WGL+24,WZL24,LWWZ25} and the estimation of the Petz-R{\'e}nyi relative entropy was considered in \cite{LF25}. 

In this paper, we consider the estimation of the quantum $\alpha$-Tsallis relative entropy \cite{Abe03}:
\[
    \dTsa{\alpha}{\rho}{\sigma}{} \coloneqq \frac{1}{1-\alpha}\rbra*{1-\tr\rbra*{\rho^\alpha\sigma^{1-\alpha}}}, \quad 0 < \alpha < 1.
\]
The quantum $\alpha$-Tsallis relative entropy is a generalization of the quantum $\alpha$-Tsallis entropy \cite{Tsa88,Rag95}.
The latter converges to the von Neumann entropy when $\alpha \to 1$ while the former converges to the von Neumann relative entropy when $\alpha \to 1^-$ \cite{Abe03b}:
\[
    \lim_{\alpha \to 1^-} \dTsa{\alpha}{\rho}{\sigma}{} = \mathrm{D}\diverg{\rho}{\sigma} \coloneqq \tr\rbra[\big]{\rho\rbra*{\log\rbra{\rho}-\log\rbra{\sigma}}}.
\]
As a measure of distinguishability between quantum states, the quantum Tsallis relative entropy is also related to the quantum Petz-R\'enyi relative entropy \cite{Pet86} and the quantum Chernoff bound \cite{ACM+07,ANSV08,Fan25}. 
In particular, the quantum $1/2$-Tsallis relative entropy is essentially the squared Hellinger distance (up to a constant factor) \cite{LZ04}:
\[
    \dH{\rho}{\sigma}{2}{} \coloneqq \frac{1}{2} \tr\rbra*{\rbra*{\sqrt{\rho}-\sqrt{\sigma}}^2} = \frac{1}{2} \dTsa{1/2}{\rho}{\sigma}{} = 1 - \mathrm{A}\rbra{\rho, \sigma},
\]
where $\mathrm{A}\rbra{\rho, \sigma} \coloneqq \tr\rbra{\sqrt{\rho}\sqrt{\sigma}}$ is known as the affinity. 
For general $\alpha$, the quantum $\alpha$-Tsallis relative entropy is known to be related to variational representations \cite{SH20}, and it can be used to quantify the coherence \cite{Ras16} and imaginarity \cite{Xu24} of quantum states. 
As a representative application, quantum Tsallis relative entropies have recently been adopted as an objective function of quantum Boltzmann machines in \cite{Wil25}. 
For more properties of the quantum $\alpha$-Tsallis relative entropy, see, e.g., \cite{FYK04}. 

The main contribution of this paper is that we provide a computational complexity picture of the estimation of the quantum Tsallis relative entropy. 
A comparison with the results for other quantum distinguishability measures is presented in \cref{tab:cmp}. 
In sharp contrast to previous work, this is, to our knowledge, the first comprehensive study of the estimation of \textit{a family of quantum relative entropies}.\footnote{Estimators for the quantum Jensen-Shannon divergence and the quantum Jensen-(Shannon-)Tsallis divergence \cite{BH09} are implied by the estimators for the von Neumann entropy \cite{AISW20,BMW16,GL20,WGL+24,WZ24b} and the quantum $\alpha$-Tsallis entropy \cite{LW25a}.  However, these types of divergences are not relative entropies.}
Specifically, our results on the estimation of the quantum Tsallis relative entropy range over the quantum query complexity, the quantum sample complexity, and the hardness in terms of computational complexity classes. 
As an application, we obtain an approach to tolerant quantum state certification with respect to the quantum Hellinger distance, which \textit{exponentially} outperforms the folklore approach based on quantum state tomography \cite{HHJ+17,OW16} in the low-rank case. To our knowledge, this is the \textit{first} efficient quantum tester for tolerant quantum state certification with respect to the quantum Hellinger distance.

\begin{table}[!htp]
\centering
\caption{The computational complexity of the estimation of quantum distinguishability measures.}
\label{tab:cmp}
\adjustbox{max width=\textwidth}{
\begin{tabular}{cccccc}
\toprule
    & \begin{tabular}{c}
         Quantum $\ell_\alpha$ Distance \\
         for $\alpha > 1$ \\
         (Hilbert-Schmidt \\
         Distance for $\alpha = 2$) 
    \end{tabular} & Trace Distance & \begin{tabular}{c}
         Uhlmann Fidelity \\
         (Bures Distance)
    \end{tabular} & \begin{tabular}{c}
         Von Neumann \\
         Relative Entropy
    \end{tabular} & \begin{tabular}{c}
         Quantum $\alpha$-Tsallis \\
         Relative Entropy \\
         for $0 < \alpha < 1$ \\
         (Hellinger Distance \\
         for $\alpha = 1/2$)
    \end{tabular} \\ \midrule
\begin{tabular}{c}
     Query \\
     Complexity
\end{tabular}  & \multirow{3}{*}{\begin{tabular}{c} $O\rbra{1}$ \\ \cite{BCWdW01,LW25b} \end{tabular}} & \begin{tabular}{c} $O\rbra{r}$ \\ \cite{WZ23} \end{tabular} & \begin{tabular}{c} $O\rbra{r}$ \\ \cite{UNWT25} \end{tabular} & / & \begin{tabular}{c} $\widetilde{O}\rbra{r^{1.5}}$ \\ \cref{thm:estimator-main} \end{tabular} \\ 
\begin{tabular}{c}
     Sample \\
     Complexity
\end{tabular}   &  & \begin{tabular}{c} $\widetilde{O}\rbra{r^2}$ \\ \cite{WZ23} \end{tabular}  & \begin{tabular}{c} $\widetilde{O}\rbra{r^{5.5}}$ \\ \cite{GP22} \end{tabular} & \begin{tabular}{c} $O\rbra{d^2}$\footnote{It is assumed that 
all eigenvalues of $\sigma$ are no less than
$\exp\rbra{-O\rbra{d}}$
when estimating $\mathrm{D}\diverg{\rho}{\sigma}$, where $\rho$ and $\sigma$ are $d$-dimensional.} \\ \cite{Hay25} \end{tabular} & \begin{tabular}{c} $\widetilde{O}\rbra{r^{3.5}}$ \\ \cref{thm:estimator-main} \end{tabular} \\ \midrule
\begin{tabular}{c}
     Hardness \\
     (Low-Rank)
\end{tabular} & \multirow{3}{*}{\begin{tabular}{c} $\BQP$-hard \\ \cite{RASW23} \end{tabular}} & \begin{tabular}{c} $\BQP$-hard \\ \cite{WZ23} \end{tabular} & \begin{tabular}{c} $\BQP$-hard \\ \cite{RASW23} \end{tabular} & \begin{tabular}{c} $\BQP$-hard \\ \cite{LW25a} \end{tabular} & \begin{tabular}{c} $\BQP$-hard \\ \cref{thm:completeness-main} \end{tabular} \\ 
\begin{tabular}{c}
     Hardness \\
     (General)
\end{tabular}  &    & \begin{tabular}{c} $\QSZK$-hard \\ \cite{Wat02} \end{tabular} & \begin{tabular}{c} $\QSZK$-hard \\ \cite{Wat02} \end{tabular} & \begin{tabular}{c} $\QSZK$-hard \\ \cite{BASTS10} \end{tabular} & \begin{tabular}{c} $\QSZK$-hard \\ \cref{thm:completeness-main} \end{tabular} \\ \bottomrule
\end{tabular}
}
\end{table}

\subsection{Main results}

Our first result is an efficient quantum algorithm for estimating the quantum Tsallis relative entropy. 

\begin{theorem}[Estimators for quantum Tsallis relative entropy]\label{thm:estimator-main}
    For any constant $\alpha \in \interval[open]{0}{1}$, given two unknown quantum states $\rho$ and $\sigma$:
    \begin{itemize}
        \item If $r = \max\cbra{\rank\rbra{\rho}, \rank\rbra{\sigma}}$ is known, then we can estimate $\dTsa{\alpha}{\rho}{\sigma}{}$ to within additive error $\eps$ using 
            \begin{itemize}
               \item  $O\rbra*{\frac{r^{2+3\alpha}}{\varepsilon^{\frac{2}{\alpha}+\frac{3}{1-\alpha} }}\log^{2}\rbra*{\frac{r}{\varepsilon}}}$ samples of $\rho$ and $\sigma$  for $0 < \alpha <1/2$, \\
               $O\rbra*{\frac{r^{3.5}}{\varepsilon^{10}}\log^{4}\rbra*{\frac{r}{\varepsilon}}}$ samples of $\rho$ and $\sigma$  for $\alpha = 1/2$, \\
                 and $O\rbra*{\frac{r^{5-3\alpha}}{\varepsilon^{\frac{3}{\alpha}+\frac{2}{1-\alpha} }}\log^{2}\rbra*{\frac{r}{\varepsilon}}}$ samples of $\rho$ and $\sigma$  for  $1/2 < \alpha < 1$
                (see \cref{thm:Tsallis-relative-sample});
                 \item  $O\rbra*{\frac{r^{1+\alpha}}{\eps^{\frac{1}{\alpha} + \frac{1}{1-\alpha}}}}$ queries to the state-preparation circuits of $\rho$ and $\sigma$ for $0 < \alpha < 1/2$, \\
                 $O\rbra*{\frac{r^{1.5}}{\eps^{4}}\log\rbra*{\frac{r}{\eps}}}$, queries to the state-preparation circuits of $\rho$ and $\sigma$  for $\alpha = 1/2$, \\
                 and $O\rbra*{\frac{r^{2-\alpha}}{\eps^{\frac{1}{\alpha} + \frac{1}{1-\alpha}}}}$ queries to the state-preparation circuits of $\rho$ and $\sigma$  for $1/2 < \alpha < 1$(see~\cref{thm:Tsallis-relative-query}).
            \end{itemize} 
        \item If $r = \rank\rbra{\rho}$ is known, then we can estimate $\dTsa{\alpha}{\rho}{\sigma}{}$ to within additive error $\eps$ using 
            \begin{itemize}
                \item $O\rbra*{\frac{r^{\frac{5}{\alpha}-3}}{\varepsilon^{\frac{5}{\alpha}}}\log^{4}\rbra*{\frac{r}{\varepsilon}} + \frac{r^{ \frac{3}{\alpha}-1}}{\varepsilon^{ \frac{3}{\alpha} + \frac{2}{1-\alpha}}}\log^{2}\rbra*{\frac{r}{\varepsilon}}}$ samples of $\rho$ and $\sigma$ (see \cref{thm:sample-complexity-of-general-affinity-rank-rho});
                \item $O\rbra*{\frac{r^{\frac{2}{\alpha}-1}}{\varepsilon^{\frac{2}{\alpha}}}\log \rbra*{\frac{r}{\varepsilon}} + \frac{r^{\frac{1}{\alpha}}}{\varepsilon^{\frac{1}{\alpha}+ \frac{1}{1-\alpha}}}}$ queries to the state-preparation circuits of $\rho$ and $\sigma$ (see \cref{thm:query-complexity-of-general-affinity-rank-rho}).
            \end{itemize} 
        \item If $r = \rank\rbra{\sigma}$ is known, then we can estimate $\dTsa{\alpha}{\rho}{\sigma}{}$ to within additive error $\eps$ using 
            \begin{itemize}
                \item $O\rbra*{\frac{r^{5-3\alpha}}{\varepsilon^{\frac{5}{\alpha}}}\log^{4}\rbra*{\frac{r}{\varepsilon}} + \frac{r^{\frac{2}{1-\alpha}+ 3(1-\alpha)}}{\varepsilon^{\frac{3}{\alpha} +\frac{2}{1-\alpha}  }}\log^{2}\rbra*{\frac{r}{\varepsilon}}}$ samples of $\rho$ and $\sigma$ (see \cref{thm:sample-complexity-of-general-affinity-rank-sigma});
                \item $O\rbra*{\frac{r^{2-\alpha}}{\varepsilon^{\frac{2}{\alpha}}}\log \rbra*{\frac{r}{\varepsilon}} + \frac{r^{(1-\alpha)+\frac{1}{1-\alpha}}}{\varepsilon^{\frac{1}{\alpha}+\frac{1}{1-\alpha}}}}$ queries to the state-preparation circuits of $\rho$ and $\sigma$ (see \cref{thm:query-complexity-of-general-affinity-rank-sigma}).
            \end{itemize} 

        \item If $r = \min\cbra{\rank\rbra{\rho}, \rank\rbra{\sigma}}$ is known, then we can estimate $\dTsa{\alpha}{\rho}{\sigma}{}$ to within additive error $\eps$ using 
            \begin{itemize}
                  \item $O\rbra*{\frac{r^{\frac{2}{\alpha}+\frac{5}{1-\alpha}-3}}{\varepsilon^{\frac{2}{\alpha}+\frac{5}{1-\alpha}}}\log^{2}\rbra*{\frac{r}{\varepsilon}} }$ samples of $\rho$ and $\sigma$ for  $0< \alpha \le \frac{1}{2}$, \\
                  and $O\rbra*{\frac{r^{\frac{5}{\alpha}+\frac{2}{1-\alpha}-3}}{\varepsilon^{\frac{5}{\alpha}+\frac{2}{1-\alpha}}}\log^{2}\rbra*{\frac{r}{\varepsilon}} }$ samples of $\rho$ and $\sigma$ for $\frac{1}{2}< \alpha < 1$  (see \cref{thm:sample-complexity-of-general-affinity-min-rank});
                   \item $O\rbra*{\frac{r^{\frac{1}{\alpha}+\frac{2}{1-\alpha}-1}}{\varepsilon^{\frac{1}{\alpha}+\frac{2}{1-\alpha}}}}$  queries to the state-preparation circuits of $\rho$ and $\sigma$ for $0 < \alpha \leq \frac{1}{2}$,\\
                   and $O\rbra*{\frac{r^{\frac{2}{\alpha}+\frac{1}{1-\alpha}-1}}{\varepsilon^{\frac{2}{\alpha}+\frac{1}{1-\alpha}}}}$ queries to the state-preparation circuits of $\rho$ and $\sigma$ for $\frac{1}{2} < \alpha < 1$ (see \cref{thm:query-complexity-of-general-affinity-min-rank}).
            \end{itemize} 
    \end{itemize}
\end{theorem}

For simplicity, \cref{thm:estimator-main} actually gives an upper bound on the sample complexity and query complexity for estimating the quantum $\alpha$-Tsallis relative entropy for any constant $\alpha \in \rbra{0, 1}$. 
The specific sample and query complexities depend on $\alpha$. 
See \cref{sec:ub} for the details.
For completeness, we also provide lower bounds of $\Omega\rbra{r}$ and $\Omega\rbra{r^{1/3}}$ respectively on the sample complexity and query complexity in \cref{sec:lb}, meaning that there is only room for a polynomial improvement over our upper bounds. 
In particular, we summarize the quantum complexities of estimating the Hellinger distance (i.e., $\alpha=1/2$) in Table \ref{tab:hellinger-distance-complexity}. 
%For more details, please refer to the corresponding theorems. 

\begin{table}[!htb]
\centering
\caption{The query/sample complexity of estimating quantum Hellinger distance.}
\label{tab:hellinger-distance-complexity}
    \begin{tabular}{ccc}
    \toprule
    Condition & Query complexity & Sample complexity\\
    \midrule
    $r = \max\cbra{\rank\rbra{\rho}, \rank\rbra{\sigma}}$     & $\widetilde{O}\rbra{r^{1.5}}$  & $\widetilde{O}\rbra{r^{3.5}}$    \\
    % $r = \rank\rbra{\rho}$     & $\widetilde{O}\rbra{r^{3}}$   &  $\widetilde{O}\rbra{r^{7}}$    \\
    $r = \rank\rbra{\rho}$ or $\rank\rbra{\sigma}$  & $\widetilde{O}\rbra{r^{2.5}}$   & $\widetilde{O}\rbra{r^{5.5}}$     \\
    $r = \min\cbra{\rank\rbra{\rho}, \rank\rbra{\sigma}}$     & $O\rbra{r^5}$     & $\widetilde{O}\rbra{r^{11}}$    \\
    \bottomrule
    \end{tabular}
\end{table}

As an application, this gives a quantum tester for the tolerant closeness testing between quantum states with respect to the quantum Hellinger distance (see \cref{thm:tolerant-state-certification-Hellinger-distance}).
For comparison, the tolerant quantum state certification with respect to the trace distance was considered in \cite{BOW19}. 

Our second result is the completeness of the quantum state distinguishability problem with respect to the quantum Tsallis relative entropy, denoted as $\TsaQSD_\alpha\sbra{a, b}$, which is to determine whether $\dTsa{\alpha}{\rho}{\sigma}{} \geq a$ or $\dTsa{\alpha}{\rho}{\sigma}{} \leq b$, where $\rho$ and $\sigma$ are two unknown $n$-qubit quantum states. 

\begin{theorem}[Completeness of $\TsaQSD_\alpha$, informal version of \cref{thm:hardness-combined}] \label{thm:completeness-main}
    For $\alpha \in \rbra{0, 1}$, $\TsaQSD_\alpha\sbra{a, b}$ is $\QSZK$-complete for $0 < b < 2\alpha\rbra{1-\alpha}^4 a^4 < 32\rbra{1-\alpha}^4 \alpha^5$ and it is $\BQP$-complete for $0 < b < a < \frac{1}{1-\alpha}$ in the low-rank case where the quantum states are of rank $r = \poly\rbra{n}$. 
\end{theorem}

In the special case when $\alpha = 1/2$, \cref{thm:completeness-main} further implies the completeness of the quantum state distinguishability problem with respect to the quantum Hellinger distance, denoted as $\HellQSD\sbra{a, b}$, which is to determine whether $\dH{\rho}{\sigma}{} \geq a$ or $\dH{\rho}{\sigma}{} \leq b$. 

\begin{corollary}[Completeness of $\HellQSD$]\label{corollary:completeness-hellinger-main}
    $\HellQSD\sbra{a, b}$ is $\QSZK$-complete for $0 < \sqrt{2} b < a^4 < 1/4$ and it is $\BQP$-complete for $0 < b < a < 1$ in the low-rank case where the quantum states are of rank $r=\poly\rbra{n}$. 
\end{corollary}

\cref{thm:completeness-main} (and \cref{corollary:completeness-hellinger-main}) gives a family of $\QSZK$-complete problems, which are the quantum state distinguishability problem with respect to a family of distinguishability measures $\dTsa{\alpha}{\rho}{\sigma}{}$ for any constant $\alpha \in \interval[open]{0}{1}$. 
In comparison, previous $\QSZK$-complete problems (in certain regimes) include the quantum state distinguishability problem with respect to trace distance (and fidelity) \cite{Wat02,Wat09}, the von Neumann entropy difference \cite{BASTS10}, the separability testing \cite{HMW14}, the productness testing \cite{GHMW15}, and the $G$-symmetry testing \cite{RLW25}.
In \cite{LW25b}, it was shown that the quantum $\ell_\alpha$ distance is $\QSZK$-complete for $\alpha$ inverse polynomially close to $1$. 

\subsection{Techniques}

For the upper bounds on the query and sample complexities, the key step is to estimate the value of $\tr\rbra{\rho^{\alpha}\sigma^{1-\alpha}}$. 
This can be done by the Hadamard test \cite{AJL09} while using the identity $\tr\rbra{\rho^{\alpha}\sigma^{1-\alpha}} = \tr\rbra{\rho \cdot \rho^{\alpha-1}\sigma^{1-\alpha}}$. 
To this end, we implement a unitary block-encoding of $\rho^{\alpha-1}\sigma^{1-\alpha}$ by quantum singular value transformation \cite{GSLW19} with the approximation polynomials of negative power functions \cite{Gil19} and positive power functions \cite{LW25a}. 
Specifically, let $p_1\rbra{x}$ and $p_2\rbra{x}$ be the polynomials given by \cref{lem:poly_approx_neg,lem:poly_approx_pos}, respectively, such that $\abs{p_1\rbra{x}} \leq 1$ and $\abs{p_2\rbra{x}} \leq 1$ for $x \in \interval{-1}{1}$ and 
\begin{align*}
    \abs*{p_1\rbra{x} - \frac{\delta_1^{1-\alpha}}{2} x^{\alpha-1}} & \leq \eps_1 \textup{ for } x \in \interval{-1}{-\delta_1} \cup \interval{\delta_1}{1}, \\
    \abs*{p_2(x)-\frac{1}{2}x^{1-\alpha}} & \leq \eps_2  \textup{ for } x \in \interval{-1}{1},
\end{align*}
where $\delta_1, \varepsilon_1, \varepsilon_2 \in \interval[open]{0}{1}$ are parameters to be determined that control the errors.
Then, a unitary block-encoding of $p_1\rbra{\rho}p_2\rbra{\sigma}$ can be implemented using the block-encoding techniques in \cite{LC19,GSLW19}. 
Then, it can be shown that an estimate of $\tr\rbra*{\rho p_1\rbra*{\rho} p_2\rbra*{\sigma}}$ can be obtained using the block-encoding version of the Hadamard test \cite{GP22}, which, in particular, can be used as an estimate of (scaled)
$\tr \rbra{\rho^{\alpha}\sigma^{1-\alpha}}$ with the precision given as follows (see \cref{prop:query-error-total}):
\[
\abs*{\tr\rbra[\big]{\rho p_1\rbra*{\rho} p_2\rbra*{\sigma}} 
- \frac{\delta_1^{1-\alpha}}{4} \tr \rbra[\big]{\rho^{\alpha}\sigma^{1-\alpha}} }\leq \rbra*{r\eps_2+ \frac{r^{\alpha}}{2}} \rbra*{\frac{3}{2}\delta_1+ \eps_1} + \frac{\delta_1^{1-\alpha}}{2} r^{1-\alpha} \eps_2,
\]
where $r = \max\cbra{\rank\rbra{\rho}, \rank\rbra{\sigma}}$. 
Choosing the values of these parameters appropriately, we can then estimate $\tr \rbra{\rho^{\alpha}\sigma^{1-\alpha}}$ with quantum query complexity $\widetilde{O}\rbra{r^{\min\cbra{1+\alpha, 2-\alpha}}} = \widetilde{O}\rbra{r^{1.5}}$. 
For other conditions $r = \rank\rbra{\rho}$, $r = \rank\rbra{\sigma}$, and $r = \min\cbra{\rank\rbra{\rho}, \rank\rbra{\sigma}}$, the error can be bounded similarly (but differently, see \cref{prop:query-error-total-rank-rho,prop:query-error-total-rank-sigma,prop:query-error-total-min-rank} for the details). 
To obtain the sample complexity, we adopt the algorithmic tool called samplizer \cite{WZ25,WZ24b} that enables us to simulate the aforementioned query-based approach by samples of quantum states $\rho$ and $\sigma$, which is a convenient use of the density matrix exponentiation \cite{LMR14,KLL+17,GKP+24} to simulate quantum query algorithms. With further analysis, we obtain a sample complexity of $\widetilde{O}\rbra{r^{3.5}}$. 

For the $\QSZK$-completeness of $\TsaQSD_\alpha\sbra{a, b}$, we reduce it to the quantum state distinguishability problem with respect to the trace distance \cite{Wat02,Wat09}. 
To this end, we adopt the inequalities between the trace distance and the quantum Tsallis relative entropy, which can use the trace distance as both upper \cite{ACM+07,ANSV08} and lower \cite{Ras13} bounds on the quantum Tsallis relative entropy. 
For the $\BQP$-completeness of the low-rank version of $\TsaQSD_\alpha\sbra{a, b}$, we reduce it to the estimation of the closeness
between pure quantum states \cite{RASW23,WZ23}. 

\subsection{Related work}

The entropy of a quantum state can be viewed as a special case of the quantum relative entropy, obtained when the reference state is the maximally mixed state.
The estimation of von Neumann entropy was studied in \cite{AISW20,GL20,CLW20,GH20,GHS21,WGL+24,LGLW23,WZ24b}.
The estimation of quantum R\'enyi entropy was given in \cite{AISW20,SH21,WGL+24,WZL24}. 
The estimation of quantum Tsallis entropy was given in \cite{EAO+02,Bru04,BCE+05,vEB12,JST17,SCC19,YS21,QKW24,ZL24,CWYZ25,SLLJ25,LW25a,CW25,ZWZY25,Wan25}.

Quantum state certification has been investigated in \cite{BOW19} for the trace distance, the fidelity, the Hilbert-Schmidt distance, and the quantum $\chi^2$ distance and in \cite{GL20} for the quantum $\ell_3$ distance. 
An instance-optimal approach to quantum state certification with respect to the trace distance was presented in \cite{OW25}. 

Tolerant property testing is a refinement of standard property testing, first introduced by Parnas, Ron, and Rubinfeld~\cite{PRR06}. While standard testers distinguish between the objects that own the property and those that are far from having it, tolerant testers distinguish between the objects that are close to having the property and those that are far. The tolerant testing model has been studied in distribution testing~\cite{GL20}, stabilizer states testing~\cite{AD25,ABD24,BvDH25,MT25,IL24,CGYZ25}, Hamiltonian testing~\cite{Car24,CW23,BCO24,Esc24,ADEG25,KL25,GJW+25}, junta unitaries~\cite{CLL24,BLY+25}, and junta states~\cite{BEG24}.

\subsection{Discussion}

In this paper, we provide a comprehensive picture of the estimation of the quantum $\alpha$-Tsallis relative entropy from the point of view of different complexities: quantum query complexity, quantum sample complexity, and quantum computational complexity.
As an application, we show that the tolerant quantum state certification with respect to the quantum Hellinger distance can be solved using our algorithms. 
To conclude this section, we raise several questions for future work. 

\begin{itemize}
    \item A possible direction is to investigate applications in quantum Boltzmann machines with the quantum Tsallis relative entropy as the objective function \cite{Wil25}. 
    In particular, as mentioned in \cite[Section 6.4]{Wil25}, estimating the gradient of the quantum Tsallis relative entropy can be an important task. 
    \item Another possible application is to estimate the imaginarity \cite{Xu24} of a quantum state $\rho$:
    \[
    \mathrm{M}_{\alpha}\rbra{\rho} \coloneqq \rbra{1-\alpha} \dTsa{\alpha}{\rho}{\rho^*}{}, \quad \alpha \in \rbra{0, 1},
    \]
    where $\rho^*$ is the (complex) conjugate of $\rho$. 
    The imaginarity $\mathrm{M}_{\alpha}\rbra{\rho}$ can be estimated using our algorithm in \cref{thm:estimator-main} (with minor modifications), where a challenge is to implement a unitary block-encoding of $\rho^*$. This may be done by using the protocols in \cite{MSM19,EHM+23}. 

    \item Our sample and query complexities for estimating the quantum Tsallis relative entropy are not tight yet. 
    A meaningful future direction is to close the gap between their upper and lower bounds.

    \item For the quantum state distinguishability problem $\QSD\sbra{a, b}$ with respect to the trace distance, it is known to be $\QSZK$-complete when $0 < b < a^2 < 1$ \cite{Wat02,Wat09}.
    In comparison, \cref{corollary:completeness-hellinger-main} shows that this problem with respect to the quantum Hellinger distance requires $0 < \sqrt{2}b < a^4/4 < 1$ to be $\QSZK$-complete. 
    A question is: can we loosen the condition for the problem to be $\QSZK$-complete?
    Improvements in this line of research can be found in \cite{Liu25}, for example.

    \item In addition to the quantities considered in this paper, a problem that we can consider is the estimation of other generalizations of the quantum Hellinger distance \cite{BGJ19,PV20} and other quantum divergences such as the one with $p$-power means \cite{LL21}. 
\end{itemize}

\section{Preliminaries}

This section introduces the quantum computational model, basic quantum algorithmic toolkit, efficient polynomial approximation of power functions, and several matrix inequalities.

\subsection{Notations}

\noindent \textbf{Mathematical notations.}
We use $\log\rbra{\cdot}$ to denote the natural logarithm with base $\mathrm{e}$. We denote by $\C$ and $\R$ the sets of complex numbers and real numbers. We use $\C^{n \times n}$ to denote the set of $n$ by $n$ complex matrices. 
We denote by $\R\sbra{x}$ the set of polynomials with real coefficients.
For a complex number $z$, we use $\Re\rbra{z}$ to denote its real part.
A Hilbert space is a complete inner product space.
For a finite-dimensional Hilbert space $\calH$, let $\calL\rbra{\calH}$ be the space of linear operators, and $\calL_+\rbra{\calH}$ be the set of positive semi-definite operators on it.
For $A,B\in \calL\rbra{\calH}$, we denote by $A^\dagger$ the Hermitian conjugate of $A$ and $\ave{A, B} = \tr\rbra{A^\dag B}$.
The rank, kernel, range, and spectrum (the multi-set of the eigenvalues) of a linear operator $A\in \calL\rbra{\calH}$ are denoted as $\rank\rbra{A}$, $\ker\rbra{A}$, $\ran\rbra{A}$, and $\spec\rbra{A}$ respectively.
For a linear operator $A\in \calL\rbra{\calH}$, there is a unique positive square root of the positive semi-definite operator $A^{\dag}A$ which we denote as $\abs{A} \in \calL_+\rbra{\calH}$.
For $p\in \interval[open right]{1}{\infty}$, the Schatten $p$-norm of a linear operator $A$ is defined as 
\[
    \Abs*{A}_p \coloneqq \rbra*{\tr \rbra*{\abs*{A}^p } }^{\!1/p} = \rbra*{\tr \rbra*{ (A^\dagger A)^{p/2} } }^{\!1/p}.
\]
The limit when $p$ goes to $\infty$ is the operator norm, which we denote as $\Abs{A}_\infty$ or simply $\Abs{A}$.
A function $f:\R \mapsto \R$ can be extended to matrix function for an $n \times n $ Hermitian operator $A$ with spectral decomposition $A = U\Sigma U^\dag$ as $f\rbra{A}\coloneqq Uf\rbra{\Sigma}U^\dag$, where $\Sigma =\diag\rbra{\lambda_1,\ldots,\lambda_n}$, and $f\rbra{\Sigma} \coloneqq \diag\rbra{f\rbra{\lambda_1},\ldots,f\rbra{\lambda_n}}$. 

\noindent \textbf{Notions in quantum computing.}
The state space of a quantum system is described by a (complex) Hilbert space. In this paper, we only consider finite-dimensional Hilbert spaces. A (pure) state of a quantum system corresponds to a unit vector in a Hilbert space $\calH$. We employ the Dirac notation of ket (e.g., $\ket{\phi}$) to denote column vectors as pure states, and bra (e.g., $\bra{\phi}$) to denote row vectors.
For an $n$-dimensional Hilbert space $\mathcal{H}$, we use $\cbra{\ket{j}}_{j=0}^{n-1}$ to denote a set of orthonormal basis of it.
Generally, the state of a quantum system described by $\calH$ is represented by a density operator on $\calH$, which is a positive semi-definite operator with trace $1$. We usually use $\rho, \sigma$ to denote density operators. The set of all density operators on $\mathcal{H}$ is denoted as $\mathcal{D}(\mathcal{H})=\{\rho \in \mathcal{L}_+\rbra{\calH}: \tr\rbra{\rho}=1\}$.

The evolution of a quantum system is modeled by a unitary operator $U$
satisfying $UU^{\dagger} = U^{\dagger}U = I$. For a pure state $\ket{\phi}$, the state after
evolution $U$ is $U\ket{\phi}$. For a state $\rho$, the state after evolution $U$ is
$U\rho U^{\dagger}$.

Measurement is also a basic operation for a quantum system. A projective measurement $\mathcal{M}$ is described by a set of projectors $\cbra{P_i}$ satisfying $P_i^2 = P_i$, $P_i = P_i^{\dagger}$, and $\sum_i P_i = I$. For a pure state $\ket{\phi}$, after performing the measurement $\mathcal{M}$, the measurement result $i$ will take place with probability $p_i= \bra{\phi}P_i\ket{\phi}$, and the state after observing the measurement result is $P_i\ket{\phi}/\sqrt{p_i}$.
For a mixed state $\rho$, after performing the measurement $\mathcal{M}$, the measurement result $i$ will take place with probability $p_i= \tr(P_i\rho)$, and the state after observing the measurement result is $P_i\rho P_i/p_i$.

For two quantum systems described by Hilbert spaces $\mathcal{H}_1$ and $\mathcal{H}_2$, the composite system is described by the tensor product $\mathcal{H}_1\otimes \mathcal{H}_2$. 
For a pure state $\ket{\phi}_A$ in system $A$ and a pure state $\ket{\psi}_B$ in system $B$, if the two systems do not interfere with each other, the state of the joint system is described by the state $\ket{\phi}_A\ket{\psi}_B \coloneqq \ket{\phi}_A \otimes \ket{\psi}_B$. We will sometimes abuse the subscription of quantum state $n_\rho$ to indicate the subsystem as well as the number of qubits in the system, if it does not cause any confusion.

Recall that a linear map $\mathcal{E}: \mathcal{D}(\mathcal{H}_1) \to \mathcal{D}(\mathcal{H}_2)$ is called completely positive if $(\mathcal{E}\otimes \mathcal{I})(\rho)$ is positive for any Hilbert space $\mathcal{H}$ and $\rho\in \mathcal{D}(\mathcal{H}_1 \otimes \mathcal{H})$, 
where $\mathcal{I}(\sigma) = \sigma$ for any $\sigma\in \mathcal{D}(\mathcal{H})$ is the identity channel on $\mathcal{H}$, and is trace-preserving if $\tr(\mathcal{E}(\rho)) = \tr(\rho)$ for any $\rho\in \mathcal{D}(\mathcal{H}_1) $. 
General quantum operations on a quantum system are called quantum channels, and are described by completely positive and trace-preserving (linear) maps from density operators to density operators. We usually use $\mathcal{E}$ to denote such a quantum channel. For two quantum channels $\mathcal{E}$ and $\mathcal{F}$ over a $d$-dimensional Hilbert space, their diamond norm is defined as
\[
    \Abs{\mathcal{E} - \mathcal{F}}_{\diamond} \coloneqq \max_{\rho} \Abs{(\mathcal{E} \otimes \mathcal{I}_d) (\rho) - (\mathcal{F} \otimes \mathcal{I}_d) (\rho)}_1,
\]
where $\calI_d$ is the identity channel, where $\calI_d\rbra{\rho} = \rho$ for any $d$-dimensional density operator $\rho$, and the maximization is over all density operators on a $d^2$-dimensional Hilbert space.

For more details about quantum computation and information, we refer readers to \cite{NC10}.

\subsection{Useful matrix inequalities}

In this part, we recall some matrix inequalities that will be used in later sections.
We begin with a generalization of the famous H\"older inequality into the matrix case.
\begin{theorem}[Matrix H\"older inequality, {\cite[Theorem 2]{Bau11}}]\label{fact:holder}
    For any $p,q \in \interval{1}{\infty}$ such that $\frac{1}{p}+\frac{1}{q} = 1$ and matrices $A,B \in \C^{n \times n}$, $\abs{\tr\rbra{A^\dag B}} \leq \Abs{A}_p \Abs{B}_q$.
    In particular, $\abs{\tr\rbra{A^\dag B}} \leq \Abs{A} \Abs{B}_1$. 
\end{theorem}

We recall the following inequality of the relation between different Schatten norms.
\begin{lemma}[{\cite[Equation (1.169)]{Wat18}}]\label{lem:norm_contractivity}
    For a non-zero matrix $A\in \mathbb{C}^{n\times n}$ with rank $r=\rank\rbra{A}$ and $1 \leq p \leq q \leq \infty$, we have $\Abs{A}_p \leq r^{\frac{1}{p}-\frac{1}{q}}\Abs{A}_q$.
\end{lemma}

The following inequality provides an upper bound on the quantum Chernoff bound~\cite{ACM+07, ANSV08}.

\begin{theorem}[{\cite[Theorem 1]{ACM+07}} and {\cite[Theorem 2]{ANSV08}}]\label{fact:affinity_lb}
    Let $A, B\in \mathbb{C}^{n\times n}$ be positive semi-definite matrices, then for any $0 \leq s \leq 1$,
    \[
        \tr\rbra*{A^sB^{1-s}} \geq \frac{1}{2}\tr\rbra*{A+B-\abs*{A-B}}.
    \]
\end{theorem}

We also prove a simple lemma for the trace differences of bounded rank matrices.
\begin{lemma}\label{lem:tr-dif-bound}
Let $A,B \in \mathbb{C}^{n\times n}$ be matrices such that 
$\max\{\rank(A), \rank(B)\}\le r$.
Then
\[
\bigl|\tr(A-B)\bigr|
\;\le\;
2r\,\|A-B\|.
\]
\end{lemma}
\begin{proof}
Set $M:=A-B$. By the sub-additivity of rank,
\[
\operatorname{rank}(M)\le \operatorname{rank}(A)+\operatorname{rank}(B)\le 2r.
\]
We have
\[
\abs*{\tr(M)}\le \Abs*{M}_2 = \sum_{j=1}^{2r} \sigma_j(M),
\]
where we write the singular values of $M$ as
$\sigma_1(M)\ge \sigma_2(M)\ge \cdots \ge \sigma_{2r}(M) \ge 0$.
In addition, note that $\sigma_1(M) = \Abs*{M} =\Abs*{A-B}$.
Combining the above inequalities yields
\[
\abs*{\tr(M)} \le 2r \sigma_1(M)
\le 2r\Abs*{A-B},
\]
as claimed.
\end{proof}

\subsection{Quantum entropies}

To measure the statistical uncertainty with the description of a quantum system, the von Neumann entropy is used as a quantum counterpart of the classical Shannon entropy \cite{Sha48a,Sha48b}.

\begin{definition}[Von Neumann entropy, \cite{Neu27}]
    The von Neumann entropy of a density operator $\rho \in \calD\rbra{\calH}$ is defined as
    \[
        \mathrm{S}\rbra*{\rho} = -\tr\rbra*{\rho\log\rho}.
    \]
\end{definition}

Another useful quantum entropy is the quantum Tsallis entropy~\cite{Tsa88,Rag95}.

\begin{definition}[Quantum Tsallis entropy,~\cite{Tsa88}]
    The Tsallis entropy of a density operator $\rho \in \calD\rbra{\calH}$ is defined as
    \[
        \mathrm{S}_q\rbra*{\rho}=\frac{1-\tr\rbra*{\rho^q}}{q-1}.
    \]
\end{definition}

Note that the Tsallis entropy reduces to the von Neumann entropy when taking the limit $q \to 1$.

\subsection{Closeness measures of quantum states}

We recall some common measures between quantum states, such as trace distance and Uhlmann fidelity.

\begin{definition}[Trace distance,~\cite{Rus94}]
    The trace distance between two density operators $\rho, \sigma \in \calD\rbra{\calH}$ is defined as 
    \[
        \dtr{\rho}{\sigma}{}{} = \frac{1}{2}\Abs{\rho-\sigma}_1 = \frac{1}{2}\tr\rbra*{\abs{\rho-\sigma}} = \frac{1}{2}\tr\rbra*{\rbra*{\rbra*{\rho-\sigma}^\dagger\rbra*{\rho-\sigma}}^{1/2}}.
    \]
\end{definition}

\begin{definition}[Uhlmann fidelity,~\cite{Uhl76,Joz94}]
    The Uhlmann fidelity between two density operators $\rho, \sigma \in \calD\rbra{\calH}$ is defined as 
    \[
        \Fid{\rho}{\sigma} = \tr\rbra*{\abs*{\sqrt{\rho}\sqrt{\sigma}}} = \tr\rbra*{\sqrt{\sqrt{\sigma}\rho\sqrt{\sigma}}}.
    \]
\end{definition}

Quantum affinity is used to measure the similarity between quantum states.
In this work, we consider the following parameterized generalization of quantum affinity.

\begin{definition}[Quantum affinity]\label{def:alpha_affinity}
    For $\alpha \in \interval[open]{0}{1}$, the $\alpha$-affinity between density operators $\rho, \sigma\in \calD\rbra{\calH}$ is defined as
    \[
        \dAa{\alpha}{\rho}{\sigma}{}{} = \tr \rbra*{ \rho^\alpha \sigma^{1 - \alpha}}.
    \]
\end{definition}

The case of $\alpha=1/2$ coincides with standard symmetric definition of quantum affinity $\dA{\rho}{\sigma}{} = \dAa{1/2}{\rho}{\sigma}{}$ (see~\cite{LZ04}).
Moreover, we have $0 \leq \dAa{\alpha}{\rho}{\sigma}{}{} \leq 1$ for all $\alpha, \rho$ and $\sigma$, and it equals $1$ if and only if $\rho = \sigma$.

\begin{definition}[Quantum Petz-R\'{e}nyi relative entropy,~\cite{Pet86,Ren61}]
    For $\alpha \in \interval[open]{0}{1}\cup \interval[open]{1}{+\infty}$, and $\rho,\sigma \in \mathcal{D}\rbra{\mathcal{H}}$, the $\alpha$-Petz-R\'{e}nyi relative entropy of $\rho$ with respect to $\sigma$ is defined as 
    \[
        \dRena{\alpha}{\rho}{\sigma}{}=
        \begin{cases}
             \frac{1}{\alpha-1}\log\tr(\rho^{\alpha}\sigma^{1-\alpha}), &\text{ if }  \alpha < 1 \text{ or } \ker(\sigma) \subseteq \ker(\rho);  \\
             +\infty, &\text{ otherwise} .
        \end{cases}
    \]
    Furthermore, we define $0$-, $1$-, and $\infty$-Petz-R\'{e}nyi relative entropies as the limits of 
    $\dRena{\alpha}{\rho}{\sigma}{}$ when $\alpha\to 0^+$, 
    $\alpha\to 1$, and $\alpha \to +\infty$, respectively.

\end{definition}

Note that 
\[
    \lim_{\alpha \to 1} \dRena{\alpha}{\rho}{\sigma}{} = \tr\rbra*{\rho\rbra*{\log{\rho}-\log{\sigma}}},
\]
which means the $1$-Petz-R\'{e}nyi relative entropy corresponds to the well-known von Neumann relative entropy (also known as Umegaki relative entropy~\cite{Ume62}).

\begin{definition}[Quantum Tsallis relative entropy,~\cite{FYK04,Ras13}]
    Let $\alpha \in \interval[open]{0}{1}$ and $\rho, \sigma \in \calD\rbra{\calH}$. The $\alpha$-Tsallis relative entropy of $\rho$ with respect to $\sigma$ is defined as
    \[
        \dTsa{\alpha}{\rho}{\sigma}{} = \frac{1}{1-\alpha}\rbra*{1-\tr\rbra*{\rho^\alpha\sigma^{1-\alpha}}}.
    \]
\end{definition}
Note that $ \dTsa{\alpha}{\rho}{\sigma}{} = \frac{1}{1-\alpha}\rbra*{1-\dAa{\alpha}{\rho}{\sigma}{}}$ by
definition.
The quantum Tsallis relative entropy can be regarded as a one-parameter
extension of the von Neumann relative entropy.

Classically, Csisz\'{a}r $f$-divergences \cite{Csi67,Csi08} are well-known generalizations of the Kullback-Liebler divergence \cite{KL51}.
In this work, we adopt the following definition of quantum Petz $f$-divergences, which can be regarded as a quantum counterpart of Csisz\'{a}r $f$-divergences~\cite{HMPB11}.

For operators $A,B \in \calL_+\rbra{\calH}$, we denote by $\Lambda_A$ and $\Gamma_B$ the left- and right-multiplication operations by $A$ and $B$ respectively, defined as $\Lambda_A : X \mapsto AX$ and $\Gamma_B: X \mapsto XB$ for $X \in \calL\rbra{\calH}$. Note that left- and right-multiplication operations are super-operators and commute with each other. Let $f$ be a continuous function on $\interval[open right]{0}{+\infty}$, we define
\[
    f\rbra*{\Lambda_A\Gamma_{B^{-1}}} = \sum_{a \in \spec\rbra*{A}}\sum_{b \in \spec\rbra*{B}}f\rbra*{ab^{-1}}\Lambda_{P_a}\Gamma_{Q_b},
\]
where $A=\sum_a aP_a$ and $B=\sum_b bQ_b$ are the spectral decompositions of $A$ and $B$, respectively.

Now we are ready to define the quantum Petz $f$-divergence.

\begin{definition}[Quantum Petz $f$-divergence~\cite{Pet85,Pet86,Pet10,HMPB11}]
    Let $A,B \in \calL_+\rbra{\calH}$ with $\ran\rbra{A} \subseteq \ran\rbra{B}$, and $f$ be a continuous function. The quantum Petz $f$-divergence of $A$ with respect to $B$ is
    \[
        \df{A}{B}{} \coloneqq \ave{B^{1/2}, f\rbra{\Lambda_A\Gamma_{B^{-1}}}\rbra{B^{1/2}}}.
    \]
\end{definition}

It is easy to verify that Umegaki relative entropy and quantum Tsallis relative entropy are in the family of quantum Petz $f$-divergences with generator functions $f_{\mathrm{Umegaki}}\rbra{x}=x\log\rbra{x}$ and $f_{\mathrm{Tsallis},\alpha}\rbra{x}=\frac{x^\alpha-x}{1-\alpha}$ respectively.
Similar to the quantum Pinsker inequality for quantum relative entropy (see \cite[Theorem 5.38]{Wat18}), Pinsker-type inequalities for Tsallis relative entropy are also established in~\cite{Gil10,Ras13}.

\begin{lemma}[Adapted from \cite{ACM+07,ANSV08,Ras13}]\label{lem:Tsa_vs_tr}
    For $\alpha \in \interval[open]{0}{1}$ and $\rho, \sigma \in \calD\rbra{\calH}$, 
    \[
        2\alpha\dtr{\rho}{\sigma}{2}+\frac{2}{9}\alpha\rbra{\alpha+1}\rbra{2-\alpha}\dtr{\rho}{\sigma}{4} \leq \dTsa{\alpha}{\rho}{\sigma}{} \leq \frac{\dtr{\rho}{\sigma}{}}{1-\alpha}.
    \]
\end{lemma}

\begin{proof}
    The first quantum Pinsker-type inequality is from~\cite[Equation (41)]{Ras13}. The second inequality can be derived from~\cref{fact:affinity_lb}.
\end{proof}

As a special case of \cref{lem:Tsa_vs_tr} when $\alpha = 1/2$, we have the inequality between the trace distance and the quantum Hellinger distance, stated as follows.

\begin{lemma}[{\cite[Theorem 2]{ACM+07}} and {\cite[Fact 2.25 and Proposition 2.31]{FO24}}]\label{fact:dh_vs_dtr}
For $\rho, \sigma \in \calD\rbra{\calH}$, 
    \[
        \dH{\rho}{\sigma}{2} \leq \dtr{\rho}{\sigma}{} \leq \sqrt{2}\dH{\rho}{\sigma}{}.
    \]
\end{lemma}

\subsection{Quantum computational model}

In this work, we use the standard quantum circuit model as our computational model.

\noindent \textbf{Quantum query complexity.}
A quantum unitary oracle provides access to an unknown unitary operator.
Given quantum unitary oracles $U_1, U_2, \dots, U_k$, 
a quantum query algorithm $\mathcal{A}^{U_1, U_2, \dots, U_k}$ can be described by the following quantum circuit:
\[
    W_{T} V_{T} W_{T-1} V_{T-1} \ldots W_{1} V_{1} W_{0},
\]
where each $V_{i}$ is a query to (controlled-)$U_{j}$ or (controlled-)$U_j^\dag$ for some $j$, 
and $W_{i}$'s are unitary operators implemented by one- and two-qubit quantum gates (which are independent of the oracles).
The query complexity of $\mathcal{A}$ is $T$.
The time complexity of $\mathcal{A}$ is the sum of its query complexity and the number of one- and two-qubit gates implementing $W_0, \ldots, W_T$.

In this work, we consider the following quantum unitary oracle called \emph{purified quantum query access}~\cite{GL20}.

\begin{definition}[Purified quantum query access]
    Let $\rho \in \mathcal{D}(\mathcal{H})$
    be an unknown quantum state.
    An $\rbra{n+m}$-qubit unitary operator $\calO_{\rho}$ is said to be a purified quantum query access oracle
    for $\rho$ if 
     \[
        \ket{\psi}  = \calO_{\rho} \ket{0}_n\ket{0}_{m},
    \]
    where $\ket{\psi}$ is a purification of $\rho$, i.e., $\rho = \tr_{m}\rbra*{\ketbra{\psi}{\psi}}$.
\end{definition}

\noindent \textbf{Quantum sample complexity.}
In addition to quantum query algorithms, 
we also consider quantum algorithms with samples of quantum states as their inputs.
For a quantum algorithm $\mathcal{A}'$ with samples of density operators $\rho_{i}$'s as its input, we assume the algorithm takes the form $ \mathcal{E} (\bigotimes_{i}\rho_{i}^{\otimes k_i})$ with $k_i$ being the number of samples of $\rho_i$,
where $\mathcal{E}$ is a quantum channel implemented by one- and two-qubit gates.
The sample complexity of $\mathcal{A}'$ is the sum of $k_i$'s.
The time complexity of $\mathcal{A}'$ is the number of one- and two-qubit gates implementing the quantum channel $\mathcal{E}$.

\subsection{Quantum algorithmic toolkit}

\subsubsection{Quantum amplitude estimation}

Quantum amplitude estimation is a basic quantum algorithmic subroutine that is a cornerstone of many quantum speedups.

\begin{lemma}[Quantum amplitude estimation~{\cite[Theorem 12]{BHMT02}}]\label{lem:amp_est}
    Suppose that $U$ is a unitary operator such that 
    \begin{equation*}
        U\ket{0}=\sqrt{p}\ket{0}\ket{\phi_0}+\sqrt{1-p}\ket{1}\ket{\phi_1},
    \end{equation*}
    where $\ket{\psi_0}$ and $\ket{\psi_1}$ are normalized pure states. 
    There is a quantum algorithm $\mathsf{AmpEst}\rbra{U,\varepsilon,\delta}$ that outputs an estimate of $p$ to within additive error $\varepsilon$ with success probability at least $1-\delta$, using $O\rbra{\frac{1}{\varepsilon}\log\rbra{\frac{1}{\delta}}}$ queries to $U$.
\end{lemma}

\subsubsection{Block-encoding}

Block-encoding is a common technique used to embed a matrix into a unitary operator and then use it in a quantum circuit. In this work, black-encoding is used to embed density operators. We recall the definition of block-encoding.

\begin{definition}[Block-encoding~{\cite[Definition 24]{GSLW19}}]
    Suppose that $A$ is an $n$-qubit linear operator.
    For real numbers $\alpha$, $\varepsilon > 0$ and a positive integer $a$, an $\rbra{n+a}$-qubit unitary operator $B$ is said to be an $\rbra{\alpha,a,\varepsilon}$-block-encoding of $A$ if
    \[
        \Abs*{\alpha\bra{0}^{\otimes a}B\ket{0}^{\otimes a}-A} \leq \varepsilon.
    \]
\end{definition}

Given purified access to a density operator, we can construct its block encoding, as indicated in the following lemma. 

\begin{lemma}[Block-encoding of a density operator{~\cite[Lemma 25]{GSLW19}}]\label{lem:be_dp}
    Suppose $\rho$ is a density matrix with purified access $U_{\rho}$ which is an $\rbra{n+a}$-qubit operator. Then, there exists an $\rbra{2n+a}$-qubit unitary operator $\widetilde{U}$ which is an $\rbra{1,n+a,0}$-block-encoding of $\rho$, using $O(1)$ queries to $U_{\rho}$.
\end{lemma}

The following theorem shows how to compute the matrix product between two block-encoded matrices.

\begin{lemma}[Product of block-encoded matrices~{\cite[Lemma 53 in the full version]{GSLW19}}]\label{lem:product_BE_matrices}
    Let $U$ be an $\rbra{\alpha,a,\varepsilon}$-block-encoding of an $n$-qubit operator $A$ and  $V$ is an $\rbra{\beta,b,\delta}$-block-encoding of an $n$-qubit operator $B$, then $\widetilde{U} = \mathsf{BEProduct}\rbra{U,V}
    \coloneqq \rbra{I_b\otimes U} \rbra{I_a\otimes V}$ is an $\rbra{\alpha\beta, a+b, \alpha\varepsilon+\beta\delta}$-block-encoding of the $n$-qubit operator $AB$.
\end{lemma}

The Hadamard test \cite{AJL09} can be used to estimate $\tr\rbra{A\rho}$.
We use the version of \cite{GP22}.

\begin{lemma}[Hadamard test for block-encoding~{\cite[Lemma 9]{GP22}}]\label{lem:hadamard_test}
    Suppose $U$ is a $\rbra{1,a,0}$-block-encoding of an $n$-qubit operator $A$. Given an $n$-qubit density state $\rho$, there exists a quantum algorithm $\mathsf{HadamardTest}\rbra{U, \rho}$ that returns $0$ with probability $\frac{1}{2}+\frac{1}{2}\Re\rbra{\tr\rbra{A\rho}}$, using one query to $U$ and $O\rbra{n}$ one- and two-qubit gates.
\end{lemma}

\subsubsection{Quantum singular value transformation}

In this part, we review the quantum singular value transformation (QSVT) proposed in~\cite{GSLW19}, an important quantum algorithm design toolkit.
For a Hermitian matrix $A$, consider its spectral decomposition as $A = \sum_i \lambda_i \ket{\phi_i}\bra{\phi_i}$. 
QSVT is able to implement the matrix polynomial function
$p(A) = \sum_i p(\lambda_i) \ket{\phi_i}\bra{\phi_i}$ 
for some polynomial $p$, given the block-encoding access 
of $A$. This is formally described in the following theorem.

\begin{lemma}[Quantum singular value transformation~{\cite[Theorem 31]{GSLW19}}]\label{lem:qsvt}
    Suppose $A$ is a Hermitian operator with its $\rbra{\alpha,a,\eps}$-block-encoding access $U$ given.
    Let $p \in \R\sbra{x}$ be a polynomial of degree $d$ such that $\abs{p\rbra{x}} \leq 1/2$ for $x \in \interval{-1}{1}$.
    Then, there is a quantum unitary $\widetilde{U} = \mathsf{EigenTrans}\rbra{U, p, \delta}$ being an $\rbra{1,a+2,4d\sqrt{\eps/\alpha}+\delta}$-block-encoding of $p\rbra{A/\alpha}$, which uses $O\rbra{d}$ queries to $U$ and $O\rbra{\rbra{a+1}d}$ one- and two-qubit quantum gates. 
    Moreover, the classical description of $\widetilde{U}$ can be computed on a classical computer in time $\poly\rbra{d,\log\rbra{1/\delta}}$.
\end{lemma}

\subsubsection{Quantum samplizer}

To convert a quantum algorithm with query access to a quantum algorithm with sample access,
we will adopt the algorithmic tool \textit{quantum samplizer} \cite{WZ25,WZ24b}.
The quantum samplizer abstracts the methods used in \cite{GP22,WZ23} for estimating properties of quantum states.  
The key ingredient of the quantum samplizer is the density matrix exponentiation \cite{LMR14,KLL+17,GKP+24}. 
Here, for our purpose, we need a quantum multi-samplizer (for mixed states), generalizing the quantum multi-samplizer for pure states in \cite{WZ24a}.

We first define the quantum multi-samplizer as follows.

\begin{definition} [Quantum multi-samplizer]
    A $k$-samplizer, denoted as $\mathsf{Samplize}_*\ave{*}\sbra{*}$, is a converter from a quantum query algorithm to a quantum sample algorithm such that: for any precision $\delta > 0$, quantum query algorithm $\mathcal{A}^{U_1, U_2, \dots, U_k}$ with query access to the unitary oracles $U_1, U_2, \dots, U_k$, and $n$-qubit quantum states $\rho_1, \rho_2, \dots, \rho_k$, there are unitary operators $U_{\rho_1}, U_{\rho_2}, \dots, U_{\rho_k}$ that are $\rbra{1, m, 0}$-block-encodings of $\rho_1/2, \rho_2/2, \dots, \rho_k/2$ (for some $m \geq 1$), respectively, such that
    \[
    \Abs*{\mathsf{Samplize}_\delta\ave{\mathcal{A}^{U_1, U_2, \dots, U_k}}\sbra{\rho_1, \rho_2, \dots, \rho_k} - \mathcal{A}^{U_{\rho_1}, U_{\rho_2}, \dots, U_{\rho_k}}}_\diamond \leq \delta.
    \]
\end{definition}

Following similar techniques in \cite{WZ24a}, we have the following theorem for implementing a quantum multi-samplizer.

\begin{theorem} \label{thm:multi-samplizer}
    For any $k \geq 1$, there is a $k$-samplizer $\mathsf{Samplize}_*\ave{*}\sbra{*}$ such that for any quantum query algorithm $\mathcal{A}^{U_1, U_2, \dots, U_k}$ that uses $Q_j$ queries to $U_j$ for each $1 \leq j \leq k$ and any $n$-qubit quantum states $\rho_1, \rho_2, \dots, \rho_k$, $\mathsf{Samplize}_\delta\ave{\mathcal{A}^{U_1, U_2, \dots, U_k}}\sbra{\rho_1, \rho_2, \dots, \rho_k}$ uses
    \[
        O\rbra*{\frac{Q_j Q}{\delta} \log^2\rbra*{\frac{Q}{\delta}}}
    \]
    samples of $\rho_j$ for each $1 \leq j \leq k$, where $Q = Q_1 + Q_2 + \dots + Q_k$.
    Moreover, if $\mathcal{A}^{U_1, U_2, \dots, U_k}$ uses $T$ one- and two-qubit gates, then $\mathsf{Samplize}_\delta\ave{\mathcal{A}^{U_1, U_2, \dots, U_k}}\sbra{\rho_1, \rho_2, \dots, \rho_k}$ uses
    \[
        T+O\rbra*{\frac{Q^2n}{\delta}\log^2\rbra*{\frac{Q}{\delta}}}
    \]
    one- and two-qubit gates. 
\end{theorem}

For completeness, the proof of the theorem is provided in \cref{sec:multi-samplizer}.

\subsection{Polynomial approximation}

Two efficient polynomial approximations are used in this paper. The first result is to approximate negative power functions.

\begin{lemma}[Polynomial approximations of negative power functions~{\cite[Corollary 67 in the full version]{GSLW19}}]\label{lem:poly_approx_neg}
    Let $\delta, \varepsilon \in \rbra{0,1/2}$ and $c > 0$. For the function $f\rbra{x} = \frac{\delta^c}{2}x^{-c}$, there exists an odd polynomial $p_{c, \varepsilon, \delta, -} \in \R\sbra{x}$ such that
    \begin{itemize}
        \item $\abs{p_{c, \varepsilon, \delta, -}\rbra{x}} \leq 1$ for $x \in \interval{-1}{1}$, and
        \item $\sbra{p_{c, \varepsilon, \delta, -}\rbra{x}-f\rbra{x}} \leq \varepsilon$, for $x \in \interval{-1}{-\delta} \cup \interval{\delta}{1}$.
    \end{itemize}
    Moreover, the degree of the polynomial $p_{c, \varepsilon, \delta, -}\rbra{x}$ is $O\rbra{\frac{\max\cbra{1,c}}{\delta}\log(\frac{1}{\varepsilon})}$, and the coefficients of the polynomial $p_{c, \varepsilon, \delta, -}\rbra{x}$ can be computed in classical polynomial time.
\end{lemma}

The following theorem describes how to approximate positive power functions by polynomials.

\begin{lemma}[Polynomial approximations of positive constant power functions~{\cite[Lemma 3.1]{LW25a}}]\label{lem:poly_approx_pos}
    Let $\eps \in \rbra{0,1/2}$. Let $r$ be a fixed positive integer and $\alpha$ be a fixed real number in $\rbra{-1,1}$. 
    For the function $f\rbra{x} \coloneqq \frac{1}{2}x^{r-1}\abs{x}^{1+\alpha}$, there exists a polynomial $p_{r, \alpha, \varepsilon,+}\rbra{x} \in \R\sbra{x}$ such that
    \begin{itemize}
        \item $\abs{p_{r, \alpha, \varepsilon,+}\rbra{x}} \leq 1$ for $x \in \interval{-1}{1}$ and
        \item $\abs{p_{r, \alpha, \varepsilon,+}\rbra{x}-f\rbra{x}} \leq \eps$ for $x \in \interval{-1}{1}$.
    \end{itemize}
    Moreover, the degree of the polynomial $p_{r, \alpha, \varepsilon,+}\rbra{x}$ is $O\rbra{(\frac{1}{\eps})^{\frac{1}{r+\alpha}}}$, and the coefficients of the polynomial $p_{r, \alpha, \varepsilon,+}\rbra{x}$ can be computed in classical polynomial time.
\end{lemma}

\subsection{Closeness testing of quantum states}

We first define the problem of testing the states with respect to the trace distance.

\begin{definition}[Quantum state distinguishability problem, $\QSD$, adapted from~\cite{Wat02,Wat09}]\label{def:qsd}
     Let $Q_\rho$ and $Q_\sigma$ be two quantum circuits with $m\rbra{n}$-qubit input and $n$-qubit output, where $m\rbra{n}$ is a polynomial in $n$. Let $\rho$ and $\sigma$ be $n$-qubit quantum states obtained by performing $Q_\rho$ and $Q_\sigma$ on input state $\ket{0}^{\otimes m\rbra{n}}$. Let $a\rbra{n}$ and $b\rbra{n}$ be efficiently computable functions such that $0 \leq a\rbra{n} < b\rbra{n} \leq 1$. The problem $\QSD\sbra{a,b}$ is to decide whether:
    \begin{itemize}
        \item \textup{(Yes)} $\dtr{\rho}{\sigma}{} \geq a\rbra{n}$, or
        \item \textup{(No)} $\dtr{\rho}{\sigma}{} \leq b\rbra{n}$.
    \end{itemize}
\end{definition}

Furthermore, we define the restricted version where $\rho$ and $\sigma$ are pure states.

\begin{definition}[Pure quantum state distinguishability problem, $\PureQSD$]
    Let $Q_\phi$ and $Q_\psi$ be two quantum circuits with $m\rbra{n}$-qubit input and $n$-qubit output, where $m\rbra{n}$ is a polynomial in $n$. 
    Let $\ket{\phi}$ and $\ket{\psi}$ be $n$-qubit pure quantum states obtained by performing $Q_\phi$ and $Q_\sigma$ on input state $\ket{0}^{\otimes m\rbra{n}}$. 
    Let $a\rbra{n}$ and $b\rbra{n}$ be efficiently computable functions such that $0 \leq a\rbra{n} < b\rbra{n} \leq 1$. The problem $\PureQSD\sbra{a,b}$ is to decide whether:
    \begin{itemize}
        \item \textup{(Yes)} $\dtr{\ketbra{\phi}{\phi}}{\ketbra{\psi}{\psi}}{} \geq a\rbra{n}$, or
        \item \textup{(No)} $\dtr{\ketbra{\phi}{\phi}}{\ketbra{\psi}{\psi}}{} \leq b\rbra{n}$.
    \end{itemize}    
\end{definition}

The following lemma shows the regime of $a\rbra{n}$ and $b\rbra{n}$ in which $\QSD$ is $\QSZK$-hard. It will be used to prove the $\QSZK$-hardness of estimating the quantum Tsallis relative entropy and the quantum Hellinger distance.

\begin{lemma}[$\QSZK$-containment and hardness of $\QSD\sbra{a,b}$,~{\cite{Wat02,Wat09,BDRV19}}]\label{lem:qszk-hard}
    Let $a\rbra{n}$ and $b\rbra{n}$ be efficiently computable functions such that $0 \leq b\rbra{n} < a\rbra{n} \leq 1$.
    \begin{itemize}
        \item $\QSD\sbra{a,b}$ is in $\QSZK$, when $a\rbra{n}^2-b\rbra{n} \geq 1/O\rbra{\log \rbra{n}}$.
        \item For any constant $\tau \in \rbra{0,1/2}$, $\QSD\sbra{a,b}$ is $\QSZK$-hard, when $a\rbra{n} \leq 1-2^{-n^\tau}$ and $b\rbra{n} \geq 2^{-n^\tau}$.
    \end{itemize}
\end{lemma}

When the given states are pure, the problem $\PureQSD$ is $\BQP$-hard \cite{RASW23,WZ23}. 
We recall the version in \cite{LW25a}. 

\begin{lemma}[$\BQP$-hardness of $\PureQSD\sbra{a,b}$,~{\cite[Lemma 2.17]{LW25a}}]\label{lem:BQP-hard}
    Let $a\rbra{n}$ and $b\rbra{n}$ be efficiently computable functions such that $0 \leq b\rbra{n} < a\rbra{n} \leq 1$ and $a\rbra{n}-b\rbra{n} \geq 1/\poly\rbra{n}$.
    Then, $\PureQSD\sbra*{a, b}$ is $\BQP$-hard when $a\rbra{n} \leq 1-2^{-n-1}$ and $b\rbra{n} \geq 2^{-n-1}$.
\end{lemma}

\section{Upper Bounds} \label{sec:ub}

In this section, we show query and sample complexity upper bounds for estimating quantum Tsallis relative entropy.

\subsection{Query complexity upper bound} \label{sec:query_up}

Our result about the query complexity upper bound for estimating quantum Tsallis relative entropy is as follows.

\begin{theorem}[Query upper bound for estimating quantum Tsallis relative entropy] 
\label{thm:Tsallis-relative-query}
    Let $\alpha\in \interval[open]{0}{1}$ be a constant.
    There is a quantum algorithm that, for any $\eps \in \interval[open]{0}{1}$, given purified quantum query access oracles $\mathcal{O}_{\rho}$ and $\mathcal{O}_{\sigma}$ respectively for quantum states $\rho, \sigma \in \calD\rbra{\calH}$ of rank at most $r$, with probability at least $2/3$, estimates $\dTsa{\alpha}{\rho}{\sigma}{}$ to within additive error $\eps$, using   
    \[
        \begin{cases}
            O\rbra*{\frac{r^{1+\alpha}}{\eps^{\frac{1}{\alpha} + \frac{1}{1-\alpha}}}}, & \text{ if } \alpha \in \interval[open]{0}{1/2}, \\
            O\rbra*{\frac{r^{1.5}}{\eps^{4}}\log\rbra*{\frac{r}{\eps}}}, & \text{ if } \alpha = 1/2, \\
            O\rbra*{\frac{r^{2-\alpha}}{\eps^{\frac{1}{\alpha} + \frac{1}{1-\alpha}}}}, & \text{ if } \alpha \in \interval[open]{1/2}{1},
        \end{cases}
    \]
    queries to $\mathcal{O}_{\rho}$ and $\mathcal{O}_{\sigma}$.
\end{theorem}

The key step in estimating the quantum Tsallis relative entropy is estimating the quantum affinity.
At a high level, to estimate the quantum affinity $\dAa{\alpha}{\rho}{\sigma}{}$, it suffices to obtain a good estimate of $\tr\rbra*{\rho^{\alpha}\sigma^{1-\alpha}} = \tr(\rho \cdot \rho^{\alpha-1}\sigma^{1-\alpha} )$. 
To accomplish this, we need to implement a block-encoding $U_{\mathsf{ProdBE}}$ of $\rho^{\alpha-1}\sigma^{1-\alpha}$ by QSVT \cite{GSLW19}, and the desired value can be estimated via the Hadamard test \cite{AJL09,GP22} and quantum amplitude estimation \cite{BHMT02}. 

We first describe the algorithm as follows and formally state it in \cref{algo:affinity-query}. Suppose $\calO_\rho$ and $\calO_\sigma$ are $\rbra{n+n_\rho}$- and $\rbra{n+n_\sigma}$-qubit purified query access oracles for $n$-qubit quantum states $\rho$ and $\sigma$ respectively.

\begin{algorithm}[!htbp]
    \caption{$\mathsf{AffinityEstQ}_{\alpha}\rbra{\calO_\rho, \calO_\sigma, \varepsilon_1, \varepsilon_2, \varepsilon_H, \delta_1, \delta_1', \delta_2'}$}
    \label{algo:affinity-query}
    \begin{algorithmic}[1]
        \Require $\rbra{n+n_\rho}$- and $\rbra{n+n_\sigma}$-qubit purified quantum query access oracles $\mathcal{O}_{\rho}$ and $\mathcal{O}_{\sigma}$ respectively for $n$-qubit quantum states $\rho$ and $\sigma$; parameters $\varepsilon_1, \varepsilon_2, \varepsilon_H, \delta_1, \delta_1', \delta_2' \in \interval[open]{0}{1}$.

        \Ensure An estimate of $\dAa{\alpha}{\rho}{\sigma}{}$. 
        
        \State Let $p_1\coloneqq p_{1-\alpha, \varepsilon_1, \delta_1, -}$ be the polynomial specified in \cref{lem:poly_approx_neg}, and $p_2\coloneqq p_{0, 1-\alpha, \varepsilon_2 , +}$ be the polynomial specified in \cref{lem:poly_approx_pos}. 

        \State Let $U_{A}$ be a unitary operator that is a $\rbra{1, n+n_\rho, 0}$-block-encoding of $\rho$ and $U_B$ be a unitary operator that is a $\rbra{1, n+n_\sigma, 0}$-block-encoding of $\sigma$ obtained by applying~\cref{lem:be_dp} to $\mathcal{O}_{\rho}$ and $\mathcal{O}_{\sigma}$, respectively.
        
        \State Let $U_{p_1(A)} \gets \mathsf{EigenTrans} (U_A, p_1/2, \delta_1')$ and $U_{p_2(B)} \gets \mathsf{EigenTrans} (U_B, p_2/2, \delta_2')$ by \cref{lem:qsvt}.

        \State Let $U_{p_1(A)p_2(B)} \gets \mathsf{BEProduct}(U_{p_1(A)}, U_{p_2(B)})$ by \cref{lem:product_BE_matrices}.

        \State Let $U_{\textup{HT}}$ denote the unitary part (i.e., without the final measurement in computational basis) of the quantum circuit $\mathsf{HadamardTest}(U_{p_1(A)p_2(B)}, \tr_{n_\rho}(\mathcal{O}_{\rho}\ket{0}\bra{0}\mathcal{O}_{\rho}^{\dagger}))$ by \cref{lem:hadamard_test}.
        
        \State $X\gets \mathsf{AmpEst}(U_{\textup{HT}}, \varepsilon_H, 3/4)$ by \cref{lem:amp_est}.
        
        \State \Return $16{\delta_1^{\alpha-1}}\rbra{2X-1}$.
    \end{algorithmic}
\end{algorithm}

\textbf{Step 1: Construct the block-encoding of $\rho$ and $\sigma$.}
By~\cref{lem:be_dp}, we can construct the $U_A$ and $U_B$ by using $O\rbra{1}$ queries to $\calO_\rho$ and $\calO_\sigma$, respectively, such that $U_A$ and $U_B$ are $\rbra{1, n+n_\rho, 0}$ and $\rbra{1, n+n_\sigma, 0}$-block-encodings of $\rho$ and $\sigma$, respectively.

\textbf{Step 2: Construct the block-encoding of $p_1\rbra{\rho}$ where $p_1\rbra{x} \approx \frac{\delta_1^{1-\alpha}}{2}x^{\alpha-1}$.}
Let $\eps_1, \delta_1, \delta_1' \in \interval[open]{0}{1/2}$ be parameters to be determined. By~\cref{lem:poly_approx_neg}, there exists a polynomial $p_1 \in \R\sbra{x}$ of degree $d_1 = O\rbra{\frac{1}{\delta_1}\log\rbra{\frac{1}{\eps_1}}}$ such that 
\[
    \abs*{p_1\rbra{x} - \frac{\delta_1^{1-\alpha}}{2} x^{\alpha-1}} \leq \eps_1 \textup{ for } x \in \interval{-1}{-\delta_1} \cup \interval{\delta_1}{1},
\]
and 
\[
    \abs*{p_1\rbra{x}} \leq 1 \textup{ for } x \in \interval{-1}{1}.
\]
By~\cref{lem:qsvt}, with $p \coloneqq \frac{1}{2}p_1$, $\alpha \coloneqq 1$, $a \coloneqq n+n_\rho$ and $\eps \coloneqq 0$, we can implement a quantum circuit $U_{p_1\rbra{A}}$ that is a $\rbra{1,n+n_\rho+2,\delta'_1}$-block-encoding of $\frac{1}{2}p_1(\rho)$, by using $O\rbra{d_1}=O\rbra{\frac{1}{\delta_1}\log\rbra{\frac{1}{\eps_1}}}$ queries to $U_A$ and the circuit description of $U_{p_1\rbra{A}}$ can be computed in classical time $\poly\rbra{d_1,\log\rbra{\frac{1}{\delta'_1}}}$.

\textbf{Step 3: Construct the block-encoding of $p_2\rbra{\sigma}$ where $p_2\rbra{x} \approx \frac{1}{2}x^{1-\alpha}$.}
Let $\eps_2,\delta'_2 \in \interval[open]{0}{1/2}$ be parameters to be specified later. By~\cref{lem:poly_approx_pos}, there exists an polynomial $p_2$ of degree $d_2 = O\rbra{\rbra{\frac{1}{\eps_2}}^{\frac{1}{1-\alpha}}}$ such that 
\[
    \abs*{p_2(x)-\frac{1}{2}x^{1-\alpha}} \leq \eps_2  \textup{ for } x \in \interval{-1}{1},
\]
and 
\[
    \abs*{p_2(x)} \leq 1 \textup{ for } x \in \interval{-1}{1}.
\]
By~\cref{lem:poly_approx_pos}, with $p \coloneqq \frac{1}{2}p_2$, $\alpha \coloneqq 1$, $a \coloneqq n+n_\sigma$ and $\eps \coloneqq 0$, we can implement $U_{p_2\rbra{B}}$ that is a $\rbra{1,n+n_\sigma+2,\delta'_2}$-block-encoding of $\frac{1}{2}p_2\rbra{\sigma}$, by using $O\rbra{d_2} = O\rbra{\rbra{\frac{1}{\eps_2}}^{\frac{1}{1-\alpha}}}$ queries to $U_B$ and the circuit description of $U_{p_2\rbra{B}}$ can be computed in classical time $\poly\rbra{d_2,\log\rbra{\frac{1}{\delta'_2}}}$.
    
\textbf{Step 4: Construct the block-encoding of $p_1\rbra{\rho}p_2\rbra{\sigma}$.} 
By~\cref{lem:product_BE_matrices}, we can implement a quantum circuit  $U_{p_1\rbra{A}p_2\rbra{B}}$ that is a $\rbra{1,2n+n_\rho+n_\sigma+4,\delta'_1+\delta'_2}$-block-encoding of $\frac{1}{4}p_1\rbra{\rho}p_2\rbra{\sigma}$.

\textbf{Step 5: Estimate $\tr\rbra{p_1\rbra{\rho}p_2\rbra{\sigma}\rho}$.}
By~\cref{lem:hadamard_test}, we can implement a quantum circuit using one query to $U_{p_1\rbra{A}p_2\rbra{B}}$ and a sample of $\rho$ (prepared by one query to $\calO_\rho$) that outputs $x \in \cbra{0,1}$ such that
\[
    \Prob\sbra*{x = 0} = \frac{1 + \Re\rbra*{\tr\rbra{\bra{0}_{2n+n_\rho+n_\sigma+4} U_{p_1\rbra{A}p_2\rbra{B}} \ket{0}_{2n+n_\rho+n_\sigma+4} \rho}}}{2}.
\]
Let $X$ be the estimate of $\Prob\sbra{x = 0}$ within additive error $\eps_H$ by \cref{lem:amp_est}, using $O\rbra{\frac{1}{\eps_H}}$ queries to $U_{p_1\rbra{A}p_2\rbra{B}}$ and $\calO_\rho$.
Specifically, it holds that
\begin{equation*}
    \Prob \sbra[\big]{ \abs*{X-\Prob\sbra*{x = 0}} \leq \eps_H } \geq \frac{3}{4}.
    \label{eq:hadamard_test}
\end{equation*}

\textbf{Step 6: Return $16\delta_1^{\alpha-1}\rbra{2X-1}$ as an estimate of $\dAa{\alpha}{\rho}{\sigma}{}$.}

\vspace{10pt}

We now analyze the error and determine all the parameters in the algorithm as follows.

\begin{proposition}\label{prop:query-error-single}
   Let $\alpha \in \interval[open]{0}{1}$ be a constant. 
   For any density operator $\rho \in \calD\rbra{\calH}$, positive real numbers $\varepsilon_1, \delta_1 \in \interval[open]{0}{1}$, we have
    \[
        \Abs*{\rho p_1\rbra*{\rho} - \frac{\delta_1^{1-\alpha}}{2}\rho^{\alpha} } \leq \frac{3}{2}\delta_1+ \eps_1,
    \]
    where $p_1\coloneqq p_{1-\alpha, \varepsilon_1, \delta_1, -}$ is the polynomial specified in \cref{lem:poly_approx_neg}.
\end{proposition}

\begin{proof}
    Let $\lambda_1, \lambda_2, \dots, \lambda_k$ denote the non-zero eigenvalues of $\rho$.
    For any $j\in [k]$, if $\lambda_j \geq \delta_1$, by our choice of $p_1$, we have
    \[
        \abs*{p_1\rbra*{\lambda_j}-\frac{\delta_1^{1-\alpha}}{2}\lambda_j^{\alpha-1}} \leq \eps_1.
    \]
    Note that $0 \leq \lambda_j \leq 1$, we conclude
    \[
        \abs*{\lambda_j p_1\rbra*{\lambda_j}-\frac{\delta_1^{1-\alpha}}{2}\lambda_j^{\alpha}} \leq \eps_1.
    \]
    Then consider the case when $0 \le \lambda_j \leq \delta_1$, 
    In this case, we have
    \[
        \abs*{p_1\rbra*{\lambda_j}-\frac{\delta_1^{1-\alpha}}{2}\lambda_j^{\alpha-1}} \leq  \abs*{p_1\rbra*{\lambda_j}} + \abs*{\frac{\delta_1^{1-\alpha}}{2}\lambda_j^{\alpha-1}} \leq \frac{3}{2},
    \]
    and multiplying both sides of the inequality by $\lambda_j$ gives the $\frac{3}{2}\delta_1$ upper bound.
    Combining both cases, we obtain the upper bound 
    $\frac{3}{2}\delta_1 + \eps_1$ as we desired.
\end{proof}

\begin{proposition}\label{prop:query-error-first}
    Let $\alpha \in \interval[open]{0}{1}$ be a constant.
    For any density operators $\rho, \sigma \in \calD\rbra{\calH}$, positive real numbers $\eps_1, \delta_1, \eps_2 \in \interval[open]{0}{1}$, we have
    \[
        \abs*{\tr\rbra*{\rho p_1\rbra*{\rho} p_2\rbra*{\sigma}} 
        - \tr\rbra*{\rho  \frac{\delta_1^{1-\alpha}}{2}\rho^{\alpha-1} p_2\rbra*{\sigma}}} 
        \leq \rbra*{r\eps_2+ \frac{r^{\alpha}}{2}} \rbra*{\frac{3}{2}\delta_1+ \eps_1},
    \]
    where $r=\max\{\rank(\rho), \rank(\sigma)\}$, $p_1\coloneqq p_{1-\alpha, \varepsilon_1, \delta_1, -}$ is the polynomial specified in \cref{lem:poly_approx_neg}, and $p_2\coloneqq p_{0, 1-\alpha, \varepsilon_2 , +}$ is the polynomial specified in \cref{lem:poly_approx_pos}.
\end{proposition}

\begin{proof}
    By our choice of $p_2$, we know
    \[
        \Abs*{p_2\rbra*{\sigma}-\frac{1}{2}\sigma^{1-\alpha}}_1 \leq r\eps_2.
    \]
    Let $\lambda_1, \lambda_2, \dots, \lambda_j$ denote the non-zero eigenvalues of $\sigma$ with $j\le r$. 
    We have $\sum_i \lambda_i = 1$.
    By the power mean inequality, for $1-\alpha\le 1$, we have
    \[
        \rbra*{\frac{\sum_i \lambda_i^{1-\alpha}}{j}}^{\frac{1}{1-\alpha}} \leq \frac{\sum_i \lambda_i}{j} = \frac{1}{j},
    \]
    which gives $\sum_i \lambda_i^{1-\alpha} \le j^{\alpha} \le r^{\alpha}$.
    This gives
    \[
        \Abs*{\frac{1}{2}\sigma^{1-\alpha}}_1 \leq  \frac{r^{\alpha}}{2}. 
    \]
    Combining the above, by the triangle inequality, we can get 
    \[
        \Abs*{p_2\rbra*{\sigma}}_1 \le  \Abs*{p_2\rbra*{\sigma}-\frac{1}{2}\sigma^{1-\alpha}}_1 + \Abs*{\frac{1}{2}\sigma^{1-\alpha}}_1 \leq r\eps_2+ \frac{r^{\alpha}}{2}.
    \]

    Now we have
    \[
    \begin{aligned}
    \abs*{\tr\rbra*{\rho p_1\rbra*{\rho} p_2\rbra*{\sigma}} - \tr\rbra*{\rho \frac{\delta_1^{1-\alpha}}{2}\rho^{\alpha-1}  p_2\rbra*{\sigma}} }
    & = \tr\rbra*{p_2\rbra*{\sigma} \rbra*{\rho p_1\rbra*{\rho}-
    \rho \frac{\delta_1^{1-\alpha}}{2}\rho^{\alpha-1} }}\\
    & \leq \Abs*{p_2\rbra*{\sigma}}_1  \Abs*{\rho p_1\rbra*{\rho} - \frac{\delta_1^{1-\alpha}}{2}\rho^{\alpha} } \\
    & \leq \rbra*{r\eps_2+ \frac{r^{\alpha}}{2}} \rbra*{\frac{3}{2}\delta_1+ \eps_1},
    \end{aligned}
    \]
    where the second line is obtained by the matrix H\"older inequality (\cref{fact:holder}), and the fourth line is obtained by applying \cref{prop:query-error-single}.
\end{proof}

\begin{proposition}\label{prop:query-error-second}
    Let $\alpha\in \interval[open]{0}{1}$ be a constant. 
    For any density operators $\rho, \sigma\in \mathcal{D}(\mathcal{H})$, positive real numbers $\varepsilon_1, \delta_1, \eps_2 \in \interval[open]{0}{1}$, we have
    \begin{equation*}
        \abs*{  \tr\rbra*{ \frac{\delta_1^{1-\alpha}}{2}\rho^{\alpha}p_2\rbra*{\sigma}}- \tr \rbra*{\frac{\delta_1^{1-\alpha}}{4}\rho^{\alpha}\sigma^{1-\alpha}} } \le \frac{\delta_1^{1-\alpha}}{2} r^{1-\alpha} \varepsilon_2,
    \end{equation*}
    where $r=\max\{\rank(\rho), \rank(\sigma)\}$, $p_1\coloneqq p_{1-\alpha, \varepsilon_1, \delta_1, -}$ is the polynomial specified in \cref{lem:poly_approx_neg}, and $p_2\coloneqq p_{0, 1-\alpha, \varepsilon_2 , +}$ is the polynomial specified in \cref{lem:poly_approx_pos}.
\end{proposition}

\begin{proof}
    This follows a similar reasoning to that in \cref{prop:query-error-first}.
    First, by the power mean inequality, we have
    \[
        \Abs*{\rho^{\alpha}}_1 \le r^{1-\alpha}.
    \]
    By our choice of $p_2$, we have
    \[
        \Abs*{p_2\rbra*{\sigma} -  \frac{1}{2}\sigma^{1-\alpha}} \leq \eps_2.
    \]
    Therefore, we deduce
    \begin{equation*}
    \begin{aligned}
        \abs*{  \tr\rbra*{ \frac{\delta_1^{1-\alpha}}{2}\rho^{\alpha}p_2\rbra*{\sigma}} - \tr \rbra*{\frac{\delta_1^{1-\alpha}}{4}\rho^{\alpha}\sigma^{1-\alpha}} }
        & \leq \frac{\delta_1^{1-\alpha}}{2} \tr \rbra*{\rho^{\alpha} \abs*{p_2\rbra*{\sigma}-\frac{1}{2}\sigma^{1-\alpha}}} \\
        & \leq  \frac{\delta_1^{1-\alpha}}{2} \Abs*{\rho^{\alpha}}_1 \Abs*{p_2\rbra*{\sigma}-\frac{1}{2}\sigma^{1-\alpha}} \\
        & \leq  \frac{\delta_1^{1-\alpha}}{2} r^{1-\alpha} \eps_2,
    \end{aligned}
    \end{equation*}
    where the second line is obtained by matrix H\"older inequality (\cref{fact:holder}).
\end{proof}

\begin{proposition}\label{prop:query-error-total}
    Let $X$, $\varepsilon_H$, $\eps_1$, $\eps_2$, $\delta_1$, $\delta_1'$, $\delta_2'$ be the parameters in \cref{algo:affinity-query}.
    If $\abs{X - \Prob\sbra{x = 0}} \leq \varepsilon_H$, then
    \[
        \abs*{\frac{16}{\delta_1^{1-\alpha}}\rbra*{2X-1} - \dAa{\alpha}{\rho}{\sigma}{}} 
        \leq \frac{16}{\delta_1^{1-\alpha}} \rbra*{2\varepsilon_H + \delta_1' + \delta_2'} + \rbra*{r\eps_2+ \frac{r^{\alpha}}{2}} \rbra*{6\delta_1^{\alpha}+ \frac{4\varepsilon_1}{\delta_1^{1-\alpha}}} + 2 r^{1-\alpha} \varepsilon_2.
    \]
\end{proposition}

\begin{proof}
    Suppose $\abs{X - \Prob\sbra{x = 0}} \leq \varepsilon_H$ and $U_{p_1(A)p_2(B)}$ is a $\rbra{1, 2n+n_\rho+n_\sigma+4, \delta_1'+\delta_2'}$-block-encoding of $\frac{1}{4}p_1(A)p_2(B)$, we have
    \[
        \abs*{ \rbra*{2X-1} - \Re\rbra*{\tr\rbra*{\bra{0}_{2n+n_\rho+n_\sigma+4}U_{p_1\rbra{A}p_2\rbra{B}}\ket{0}_{2n+n_\rho+n_\sigma+4} \rho}} } \leq 2\varepsilon_H ,
    \]
    which gives
    \[
        \abs*{ 4\rbra*{2X-1} - \tr\rbra*{\rho p_1\rbra*{\rho} p_2\rbra*{\sigma}} } \leq 8\varepsilon_H + 4\delta_1' + 4\delta_2'.
    \]
    By \cref{prop:query-error-first,prop:query-error-second}, we have
    \[
    \abs*{\tr\rbra*{\rho p_1\rbra*{\rho} p_2\rbra*{\sigma}} 
    - \tr \rbra*{\frac{\delta_1^{1-\alpha}}{4}\rho^{\alpha}\sigma^{1-\alpha}} }\leq \rbra*{r\eps_2+ \frac{r^{\alpha}}{2}} \rbra*{\frac{3}{2}\delta_1+ \eps_1} + \frac{\delta_1^{1-\alpha}}{2} r^{1-\alpha} \eps_2.
    \]
    Therefore, the result follows from the triangle inequality.
\end{proof}

\begin{theorem}[Query upper bound for estimating quantum affinity]\label{thm:query-complexity-of-general-affinity}
    Let $\alpha \in \interval[open]{0}{1}$ be a constant.
    There is a quantum algorithm $\mathsf{AffinityEstQ}_{\alpha}(\calO_\rho,\calO_\sigma, r, \varepsilon)$ that, for any $\eps \in \interval[open]{0}{1}$, given query access to density operators $\rho, \sigma \in \calD(\calH)$ with rank at most $r$, with probability at least $2/3$, estimating $\dAa{\alpha}{\rho}{\sigma}{}$ to within additive error $\eps$, using 
    \[
        \begin{cases}
            O\rbra*{\frac{r^{1+\alpha}}{\eps^{\frac{1}{\alpha} + \frac{1}{1-\alpha}}}}, & \text{ if } \alpha \in \interval[open]{0}{1/2}, \\
            O\rbra*{\frac{r^{1.5}}{\eps^{4}}\log\rbra*{\frac{r}{\eps}}}, & \text{ if } \alpha = 1/2, \\
            O\rbra*{\frac{r^{2-\alpha}}{\eps^{\frac{1}{\alpha} + \frac{1}{1-\alpha}}}}, & \text{ if } \alpha \in \interval[open]{1/2}{1},
        \end{cases}
    \]
    queries to $\calO_{\rho}$ and $\calO_{\sigma}$.
\end{theorem}

\begin{proof}
    We first omit the first three lines of the algorithm.
    In this case, for any $\alpha\in \interval[open]{0}{1}$,
    setting 
    \[
        \eps_1 =\delta_1 = \frac{\eps^{1/\alpha}}{16^{1/\alpha}r}, \quad \eps_2 = \frac{r^{\alpha-1}}{8}\eps, \quad \eps_H = \frac{\eps^{1/\alpha}}{8 \cdot 16^{1/\alpha}r^{1-\alpha}},
        \quad \delta_1' = \delta_2' = \frac{\eps^{1/\alpha}}{16 \cdot 16^{1/\alpha}r^{1-\alpha}}
    \] 
    in \cref{prop:query-error-total}, we know the error between the algorithm output and the desired affinity can be bounded by $\eps$.

    Now we consider the query complexity of the algorithm. 
    By our choice of parameters, we have
    \[
        d_1 = O\rbra*{\frac{1}{\delta_1}\log\rbra*{\frac{1}{\varepsilon_1}}} = O\rbra*{\frac{r}{\varepsilon^{1/\alpha}}\log \rbra*{\frac{r}{\varepsilon}}},
    \quad 
        d_2 = O\rbra*{\rbra*{\frac{1}{\eps_2}}^{1/(1-\alpha)}} =  O\rbra*{\frac{r}{\varepsilon^{1/(1-\alpha)}}}.
    \]
    
    We then discuss the complexity based on the value of $\alpha$.

    \textbf{Case 1: $\alpha \in \interval[open left]{0}{1/2}$.}
    In this case, we have $d_2 = O(d_1)$.
    Then, the query algorithm uses
    \[
        O\rbra*{\frac{r}{\eps^{1/\alpha}}\log \rbra*{\frac{r}{\eps}}}
    \]
    queries. Since we need to repeat $O(1/\eps_H)$ times, the total queries are 
    \[
        O\rbra*{\frac{r^{2-\alpha}}{\eps^{2/\alpha}}\log\rbra*{\frac{r}{\eps}}}.
    \]

    \textbf{Case 2: $\alpha \in \interval[open]{1/2}{1}$.}
    In this case, we have $d_1 = O(d_2)$.
    Then, the query algorithm uses
    \[
        O\rbra*{\frac{r}{\eps^{1/(1-\alpha)}}}
    \]
    queries. Since we need to repeat $O(1/\eps_H)$ times, the total queries are
    \[
        O\rbra*{\frac{r^{2-\alpha}}{\eps^{1/(1-\alpha)+ 1/\alpha}}}.
    \]
    
    Now, note that $\dAa{\alpha}{\rho}{\sigma}{} = \dAa{1-\alpha}{\sigma}{\rho}{}$
    Therefore, for $\alpha \in \interval[open]{0}{1/2}$, we also have an algorithm with query complexity
    \[
        O\rbra*{\frac{r^{1+\alpha}}{\eps^{1/\alpha+ 1/(1-\alpha)}}}.
    \]
    Similarly, for $\alpha \in \interval[open]{1/2}{1}$, we also have an algorithm with query complexity
    \[
        O\rbra*{\frac{r^{1+\alpha}}{\eps^{2/(1-\alpha)}}\log\rbra*{\frac{r}{\eps}}}.
    \]

    Combining the above discussions, the query complexity of the algorithm is
    \[
    \begin{cases}
        O\rbra*{\dfrac{r^{1+\alpha}}{\eps^{1/\alpha+ 1/(1-\alpha)}}}, & \text{ if } \alpha \in \interval[open]{0}{1/2}, \\
        O\rbra*{\dfrac{r^{1.5}}{\eps^{4}}\log\rbra*{\dfrac{r}{\eps}}}, & \text{ if } \alpha = 1/2, \\
        O\rbra*{\dfrac{r^{2-\alpha}}{\eps^{1/(1-\alpha)+ 1/\alpha}}}, & \text{ if } \alpha \in \interval[open]{1/2}{1}.
    \end{cases}
    \]
    These yield the proof. 
\end{proof}

Our algorithm \cref{algo:affinity-query} can be applied to estimating Tsallis relative entropy and Hellinger distance of quantum states.

\begin{proof}[Proof of \cref{thm:Tsallis-relative-query}]
    We notice that
    $\dTsa{\alpha}{\rho}{\sigma}{} = \frac{1}{1-\alpha} \rbra{1-\dAa{\alpha}{\rho}{\sigma}{}}$.
    Therefore, 
    to obtain an estimate of $\dTsa{\alpha}{\rho}{\sigma}{}$ within additive error
    $\varepsilon$, 
    it suffices to estimate $\dAa{\alpha}{\rho}{\sigma}{}$ within
    $\rbra{1-\alpha}\eps$ error. 
    The claim follows from using the algorithm $\mathsf{AffinityEstQ}_{\alpha}\rbra{\calO_\rho, \calO_\sigma, r, \rbra{1-\alpha}\varepsilon}$
    and applying \cref{thm:query-complexity-of-general-affinity}.
\end{proof}

\subsection{Query complexity with weaker rank conditions}

In this part, we consider the query complexity of estimating the Tsallis entropy when we can bound the rank of 
only one density operator.

\subsubsection{Given the rank of \texorpdfstring{$\rho$}{}}

We first consider the case of estimating $\dAa{\alpha}{\rho}{\sigma}{} = \tr \rbra{\rho^{\alpha}\sigma^{1-\alpha}}$ when the rank of $\rho$ is bounded.

Note that \cref{prop:query-error-single,prop:query-error-second} still holds in our case. For readability, we rewrite \cref{prop:query-error-second} as follows.
\begin{proposition}[Restatement of \cref{prop:query-error-second}]\label{prop:query-error-second-rank-rho}
    Let $\alpha\in \interval[open]{0}{1}$ be a constant. 
    For any density operators $\rho, \sigma\in \mathcal{D}(\mathcal{H})$, positive real numbers $\varepsilon_1, \delta_1, \eps_2 \in \interval[open]{0}{1}$, we have
    \begin{equation*}
        \abs*{  \tr\rbra*{ \frac{\delta_1^{1-\alpha}}{2}\rho^{\alpha}p_2\rbra*{\sigma}}- \tr \rbra*{\frac{\delta_1^{1-\alpha}}{4}\rho^{\alpha}\sigma^{1-\alpha}} } \le \frac{\delta_1^{1-\alpha}}{2} r_{\rho}^{1-\alpha} \varepsilon_2,
    \end{equation*}
    where $r_{\rho}\ge\rank(\rho)$, and $p_2\coloneqq p_{0, 1-\alpha, \varepsilon_2 , +}$ is the polynomial specified in \cref{lem:poly_approx_pos}.
\end{proposition}

\begin{proposition}\label{prop:query-error-first-rank-rho}
    Let $\alpha \in \interval[open]{0}{1}$ be a constant.
    For any density operators $\rho, \sigma \in \calD\rbra{\calH}$, positive real numbers $\eps_1, \delta_1, \eps_2 \in \interval[open]{0}{1}$, we have
    \[
        \abs*{\tr\rbra*{\rho p_1\rbra*{\rho} p_2\rbra*{\sigma}} 
        - \tr\rbra*{\rho  \frac{\delta_1^{1-\alpha}}{2}\rho^{\alpha-1} p_2\rbra*{\sigma}}} 
        \leq \rbra*{\eps_2+ \frac{1}{2}} \rbra*{\frac{3}{2}\delta_1+ \eps_1} r_{\rho},
    \]
    where $r_{\rho}\ge\rank(\rho)$, $p_1\coloneqq p_{1-\alpha, \varepsilon_1, \delta_1, -}$ is the polynomial specified in \cref{lem:poly_approx_neg}, and $p_2\coloneqq p_{0, 1-\alpha, \varepsilon_2 , +}$ is the polynomial specified in \cref{lem:poly_approx_pos}.
\end{proposition}

\begin{proof}
    By our choice of $p_2$, we know
    \[
        \Abs*{p_2\rbra*{\sigma}-\frac{1}{2}\sigma^{1-\alpha}} \leq \eps_2.
    \]
    Note that
    \[
        \Abs*{\frac{1}{2}\sigma^{1-\alpha}} \leq  \frac{1}{2},
    \] 
    By the triangle inequality, we can get 
    \[
        \Abs*{p_2\rbra*{\sigma}} \le  \Abs*{p_2\rbra*{\sigma}-\frac{1}{2}\sigma^{1-\alpha}} + \Abs*{\frac{1}{2}\sigma^{1-\alpha}} \leq \eps_2+ \frac{1}{2}.
    \]

    Now we have
    \[
    \begin{aligned}
    \abs*{\tr\rbra*{\rho p_1\rbra*{\rho} p_2\rbra*{\sigma}} - \tr\rbra*{\rho \frac{\delta_1^{1-\alpha}}{2}\rho^{\alpha-1}  p_2\rbra*{\sigma}} }
    & = \tr\rbra*{p_2\rbra*{\sigma} \rbra*{\rho p_1\rbra*{\rho}-
    \rho \frac{\delta_1^{1-\alpha}}{2}\rho^{\alpha-1} }}\\
    & \leq \Abs*{p_2\rbra*{\sigma}}  \Abs*{\rho p_1\rbra*{\rho} - \frac{\delta_1^{1-\alpha}}{2}\rho^{\alpha} }_1 \\
    & \leq \rbra*{\eps_2+ \frac{1}{2}} \rbra*{\frac{3}{2}\delta_1+ \eps_1}r_{\rho},
    \end{aligned}
    \]
    where the second line is obtained by the matrix H\"older inequality (\cref{fact:holder}), and the fourth line is obtained by applying \cref{prop:query-error-single}.
\end{proof}

\begin{proposition}\label{prop:query-error-total-rank-rho}
    Let $X$, $\varepsilon_H$, $\eps_1$, $\eps_2$, $\delta_1$, $\delta_1'$, $\delta_2'$ be the parameters in \cref{algo:affinity-query}.
    If $\abs{X - \Prob\sbra{x = 0}} \leq \varepsilon_H$, then
    \[
        \abs*{\frac{16}{\delta_1^{1-\alpha}}\rbra*{2X-1} - \dAa{\alpha}{\rho}{\sigma}{}} 
        \leq \frac{16}{\delta_1^{1-\alpha}} \rbra*{2\varepsilon_H + \delta_1' + \delta_2'} +r_{\rho} \rbra*{\eps_2+ \frac{1}{2}} \rbra*{6\delta_1^{\alpha}+ \frac{4\varepsilon_1}{\delta_1^{1-\alpha}}} + 2 r_{\rho}^{1-\alpha} \varepsilon_2,
    \]
     where $r_{\rho}\ge \rank(\rho)$.
\end{proposition}

\begin{proof}
       Suppose $\abs{X - \Prob\sbra{x = 0}} \leq \varepsilon_H$ and $U_{p_1(A)p_2(B)}$ is a $\rbra{1, 2n+n_\rho+n_\sigma+4, \delta_1'+\delta_2'}$-block-encoding of $\frac{1}{4}p_1(A)p_2(B)$,  in the same reasoning as in~\Cref{prop:query-error-total}, we have we have
    \[
        \abs*{ \rbra*{2X-1} - \Re\rbra*{\tr\rbra*{\bra{0}_{2n+n_\rho+n_\sigma+4}U_{p_1\rbra{A}p_2\rbra{B}}\ket{0}_{2n+n_\rho+n_\sigma+4} \rho}} } \leq 2\varepsilon_H ,
    \]
    which gives
    \[
        \abs*{ 4\rbra*{2X-1} - \tr\rbra*{\rho p_1\rbra*{\rho} p_2\rbra*{\sigma}} } \leq 8\varepsilon_H + 4\delta_1' + 4\delta_2'.
    \]

    By \cref{prop:query-error-first-rank-rho,prop:query-error-second-rank-rho}, we have
    \[
      \abs*{\tr\rbra*{\rho p_1\rbra*{\rho} p_2\rbra*{\sigma}} 
    - \tr \rbra*{\frac{\delta_1^{1-\alpha}}{4}\rho^{\alpha}\sigma^{1-\alpha}} }\le \rbra*{r_{\rho}\varepsilon_2+ \frac{1}{2} r_{\rho}} \rbra*{\frac{3}{2}\delta_1+ \varepsilon_1} + \frac{\delta_1^{1-\alpha}}{2} r_{\rho}^{1-\alpha} \varepsilon_2.
    \]
    Therefore, the result follows from the triangle inequality.
\end{proof}

\begin{theorem}[Query upper bound for estimating quantum Tsallis relative entropy with rank bounds on $\rho$] \label{thm:query-complexity-of-general-affinity-rank-rho}
  Let $\alpha \in \interval[open]{0}{1}$ be a constant.
    There is a quantum algorithm that, for any $\eps \in \interval[open]{0}{1}$, given query access to density operators $\rho, \sigma \in \calD(\calH)$ with the rank of $\rho$ upper bounded by $r$, with probability at least $2/3$, estimating $\dTsa{\alpha}{\rho}{\sigma}{}$ to within additive error $\eps$, using 
    \[
        O\rbra*{\frac{r^{\frac{2}{\alpha}-1}}{\varepsilon^{\frac{2}{\alpha}}}\log \rbra*{\frac{r}{\varepsilon}} + \frac{r^{\frac{1}{\alpha}}}{\varepsilon^{\frac{1}{\alpha}+ \frac{1}{1-\alpha}}}}
    \]
    queries to $\calO_{\rho}$ and $\calO_{\sigma}$.
\end{theorem}

\begin{proof}
    We notice that
    $\dTsa{\alpha}{\rho}{\sigma}{} = \frac{1}{1-\alpha} \rbra{1-\dAa{\alpha}{\rho}{\sigma}{}}$.
    Therefore, 
    to obtain an estimate of $\dTsa{\alpha}{\rho}{\sigma}{}$ within additive error
    $\varepsilon$, 
    it suffices to estimate $\dAa{\alpha}{\rho}{\sigma}{}$ to within
    $\rbra{1-\alpha}\eps$ error. 
    We now show how to 
 estimate $\dAa{\alpha}{\rho}{\sigma}{}$ to within
    $\eps$ error.
    
  For any $\alpha\in \interval[open]{0}{1}$, setting 
    \[
        \eps_1 =\delta_1 = \frac{\eps^{1/\alpha}}{40^{1/\alpha}r^{1/\alpha}}, \quad \eps_2 = \frac{\eps}{8r^{1-\alpha}}, \quad \eps_H = \frac{\eps^{1/\alpha}}{8 \cdot 40^{1/\alpha}r^{(1-\alpha)/\alpha}},
        \quad \delta_1' = \delta_2' = \frac{\eps^{1/\alpha}}{16 \cdot 40^{1/\alpha}r^{(1-\alpha)/\alpha}}
    \]
    in \cref{prop:query-error-total-rank-rho}, 
    we have
    \[
     \abs*{\frac{16}{\delta_1^{1-\alpha}}\rbra{1-2X} - \dAa{\alpha}{\rho}{\sigma}{}{}} \le \varepsilon.
    \]

    Now we consider the query complexity of the algorithm. 
    By our choice of parameters, we have
    \[
        d_1 = O\rbra*{\frac{1}{\delta_1}\log\rbra*{\frac{1}{\varepsilon_1}}} = O\rbra*{\frac{r^{1/\alpha}}{\varepsilon^{1/\alpha}}\log \rbra*{\frac{r}{\varepsilon}}},
    \quad 
        d_2 = O\rbra*{\frac{1}{\eps_2^{1/(1-\alpha)}}} =  O\rbra*{\frac{r}{\varepsilon^{1/(1-\alpha)}}}.
    \]
Since we need to repeat $O(1/\eps_H)$ times, the total number of queries is
\[
O\rbra*{\frac{d_1+d_2}{\varepsilon_H}} = O\rbra*{\frac{r^{2/\alpha-1}}{\varepsilon^{2/\alpha}}\log \rbra*{\frac{r}{\varepsilon}} + \frac{r^{1/\alpha}}{\varepsilon^{1/\alpha+ 1/(1-\alpha)}}}. 
\]
\end{proof}

\subsubsection{Given the rank of \texorpdfstring{$\sigma$}{}}

We then consider the case when the rank of $\sigma$ is bounded.

Note that \cref{prop:query-error-single,prop:query-error-first} still holds in our case. For readability, we rewrite \cref{prop:query-error-first} as follows.

\begin{proposition}[Restatement of \cref{prop:query-error-first}] \label{prop:query-error-first-rank-sigma}
    Let $\alpha \in \interval[open]{0}{1}$ be a constant.
    For any density operators $\rho, \sigma \in \calD\rbra{\calH}$, positive real numbers $\eps_1, \delta_1, \eps_2 \in \interval[open]{0}{1}$, we have
    \[
        \abs*{\tr\rbra*{\rho p_1\rbra*{\rho} p_2\rbra*{\sigma}} 
        - \tr\rbra*{\rho  \frac{\delta_1^{1-\alpha}}{2}\rho^{\alpha-1} p_2\rbra*{\sigma}}} 
        \leq \rbra*{r_{\sigma}\eps_2+ \frac{r_{\sigma}^{\alpha}}{2}} \rbra*{\frac{3}{2}\delta_1+ \eps_1},
    \]
    where $r_{\sigma}\ge \rank(\sigma)$, $p_1\coloneqq p_{1-\alpha, \varepsilon_1, \delta_1, -}$ is the polynomial specified in \cref{lem:poly_approx_neg}, and $p_2\coloneqq p_{0, 1-\alpha, \varepsilon_2 , +}$ is the polynomial specified in \cref{lem:poly_approx_pos}.
\end{proposition}

If the rank of $\rho$ can not be bounded, we still have the following slightly weaker bound.

\begin{proposition}\label{prop:query-error-second-rank-sigma}
   Let $\alpha\in \interval[open]{0}{1}$ be a constant. 
   For any density operators $\rho, \sigma\in \mathcal{D}(\mathcal{H})$, 
   positive real numbers $\varepsilon_1, \delta_1, \varepsilon_2 \in 
   \interval[open]{0}{1}$,
   we have
    \begin{equation*}
        \abs*{  \tr\rbra*{ \frac{\delta_1^{1-\alpha}}{2}\rho^{\alpha}p_2\rbra*{\sigma}}- \tr \rbra*{\frac{\delta_1^{1-\alpha}}{4}\rho^{\alpha}\sigma^{1-\alpha}} } \leq  \frac{\delta_1^{1-\alpha}}{2} r_{\sigma} \eps_2,
    \end{equation*}
    where $r_{\sigma}\ge  \rank(\sigma)$, $p_1\coloneqq p_{1-\alpha, \varepsilon_1, \delta_1, -}$ is the polynomial specified in \cref{lem:poly_approx_neg}, and $p_2\coloneqq p_{0, 1-\alpha, \varepsilon_2 , +}$ is the polynomial specified in \cref{lem:poly_approx_pos}.
\end{proposition}

\begin{proof}
    First, it is direct to see that
    \[
        \Abs*{\rho^{\alpha}} \le 1.
    \]
    By our choice of $p_2$, we have
    \[
        \Abs*{p_2\rbra*{\sigma} -  \frac{1}{2}\sigma^{1-\alpha}}_1 \le r_{\sigma}\varepsilon_2.
    \]
    Therefore, we deduce
    \[
    \begin{aligned}
       \abs*{  \tr\rbra*{\frac{\delta_1^{1-\alpha}}{2} \rho^{\alpha}p_2\rbra*{\sigma}}- \tr \rbra*{\frac{\delta_1^{1-\alpha}}{4}\rho^{\alpha}\sigma^{1-\alpha}} }
    & \le \frac{\delta_1^{1-\alpha}}{2} \tr \rbra*{\rho^{\alpha} \rbra*{p_2\rbra*{\sigma} -  \frac{1}{2}\sigma^{1-\alpha}}} \\
    &\le  \frac{\delta_1^{1-\alpha}}{2^{\alpha}} \Abs*{\rho^{\alpha}}
    \Abs*{p_2\rbra*{\sigma} -  \frac{1}{2}\sigma^{1-\alpha}}_1 \\
      &\le  \frac{\delta_1^{1-\alpha}}{2} r_{\sigma} \varepsilon_2,
    \end{aligned}
    \]
    where the second line is obtained by the matrix H\"{o}lder inequality (\cref{fact:holder}).
\end{proof}

\begin{proposition}\label{prop:query-error-total-rank-sigma}
    Let $X$, $\varepsilon_H$, $\eps_1$, $\eps_2$, $\delta_1$, $\delta_1'$, $\delta_2'$ be the parameters in \cref{algo:affinity-query}.
    If $\abs{X - \Prob\sbra{x = 0}} \leq \varepsilon_H$, then
    \[
        \abs*{\frac{16}{\delta_1^{1-\alpha}}\rbra*{2X-1} - \dAa{\alpha}{\rho}{\sigma}{}} 
        \leq \frac{16}{\delta_1^{1-\alpha}} \rbra*{2\varepsilon_H + \delta_1' + \delta_2'} + \rbra*{r_{\sigma}\eps_2+ \frac{r_{\sigma}^{\alpha}}{2}} \rbra*{6\delta_1^{\alpha}+ \frac{4\varepsilon_1}{\delta_1^{1-\alpha}}} + 2 r_{\sigma} \varepsilon_2,
    \]
      where $r_{\sigma}\ge  \rank(\sigma)$.
\end{proposition}
\begin{proof}
    Suppose $\abs{X - \Prob\sbra{x = 0}} \leq \varepsilon_H$ and $U_{p_1(A)p_2(B)}$ is a $\rbra{1, 2n+n_\rho+n_\sigma+4, \delta_1'+\delta_2'}$-block-encoding of $\frac{1}{4}p_1(A)p_2(B)$,  in the same reasoning as in~\Cref{prop:query-error-total}, we have we have
    \[
        \abs*{ \rbra*{2X-1} - \Re\rbra*{\tr\rbra*{\bra{0}_{2n+n_\rho+n_\sigma+4}U_{p_1\rbra{A}p_2\rbra{B}}\ket{0}_{2n+n_\rho+n_\sigma+4} \rho}} } \leq 2\varepsilon_H ,
    \]
    which gives
    \[
        \abs*{ 4\rbra*{2X-1} - \tr\rbra*{\rho p_1\rbra*{\rho} p_2\rbra*{\sigma}} } \leq 8\varepsilon_H + 4\delta_1' + 4\delta_2'.
    \]
    By \cref{prop:query-error-first-rank-sigma,prop:query-error-second-rank-sigma}, we have
    \[
    \abs*{\tr\rbra*{\rho p_1\rbra*{\rho} p_2\rbra*{\sigma}} 
    - \tr \rbra*{\frac{\delta_1^{1-\alpha}}{4}\rho^{\alpha}\sigma^{1-\alpha}} }\leq \rbra*{r_{\sigma}\eps_2+ \frac{r_{\sigma}^{\alpha}}{2}} \rbra*{\frac{3}{2}\delta_1+ \eps_1} + \frac{\delta_1^{1-\alpha}}{2} r_{\sigma} \eps_2.
    \]
    Therefore, the result follows from the triangle inequality.
\end{proof}

\begin{theorem}[Query upper bound for estimating quantum Tsallis relative entropy with rank bounds on $\sigma$] \label{thm:query-complexity-of-general-affinity-rank-sigma}
 Let $\alpha \in \interval[open]{0}{1}$ be a constant.
    There is a quantum algorithm that, for any $\eps \in \interval[open]{0}{1}$, given query access to density operators $\rho, \sigma \in \calD(\calH)$ with the rank of $\sigma$ upper bounded by $r$, with probability at least $2/3$, estimating $\dTsa{\alpha}{\rho}{\sigma}{}$ to within additive error $\eps$, using 
    \[
        O\rbra*{\frac{r^{2-\alpha}}{\varepsilon^{\frac{2}{\alpha}}}\log \rbra*{\frac{r}{\varepsilon}} + \frac{r^{(1-\alpha)+\frac{1}{1-\alpha}}}{\varepsilon^{\frac{1}{\alpha}+\frac{1}{1-\alpha}}}}
    \]
    queries to $\calO_{\rho}$ and $\calO_{\sigma}$.
\end{theorem}

\begin{proof}
We notice that
    $\dTsa{\alpha}{\rho}{\sigma}{} = \frac{1}{1-\alpha} \rbra{1-\dAa{\alpha}{\rho}{\sigma}{}}$.
    Therefore, 
    to obtain an estimate of $\dTsa{\alpha}{\rho}{\sigma}{}$ within additive error
    $\varepsilon$, 
    it suffices to estimate $\dAa{\alpha}{\rho}{\sigma}{}$ to within
    $\rbra{1-\alpha}\eps$ error. 
    We now show how to 
 estimate $\dAa{\alpha}{\rho}{\sigma}{}$ to within
    $\eps$ error.

    For any $\alpha\in \interval[open]{0}{1}$,  
  setting 
    \[
        \eps_1 =\delta_1 = \frac{\eps^{1/\alpha}}{40^{1/\alpha}r}, \quad \eps_2 = \frac{\eps}{40r}, \quad \eps_H = \frac{\eps^{1/\alpha}}{8 \cdot 40^{1/\alpha}r^{1-\alpha}},
        \quad \delta_1' = \delta_2' = \frac{\eps^{1/\alpha}}{16 \cdot 40^{1/\alpha}r^{1-\alpha}}
    \]
    in \cref{prop:query-error-total-rank-sigma},
    we have
    \[
     \abs*{\frac{16}{\delta_1^{1-\alpha}}\rbra{1-2X} - \dAa{\alpha}{\rho}{\sigma}{}{}} \le \varepsilon.
    \]

    Now we consider the query complexity of the algorithm. 
    By our choice of parameters, we have
    \[
        d_1 = O\rbra*{\frac{1}{\delta_1}\log\rbra*{\frac{1}{\varepsilon_1}}} = O\rbra*{\frac{r}{\varepsilon^{1/\alpha}}\log \rbra*{\frac{r}{\varepsilon}}},
    \quad 
        d_2 = O\rbra*{\frac{1}{\eps_2^{1/(1-\alpha)}}} =  O\rbra*{\frac{r^{1/(1-\alpha)}}{\varepsilon^{1/(1-\alpha)}}}.
    \]
Since we need to repeat $O(1/\eps_H)$ times, the total number of queries is
\[
O\rbra*{\frac{d_1+d_2}{\varepsilon_H}} = O\rbra*{\frac{r^{2-\alpha}}{\varepsilon^{2/\alpha}}\log \rbra*{\frac{r}{\varepsilon}} + \frac{r^{(1-\alpha)+1/(1-\alpha)}}{\varepsilon^{1/\alpha+1/(1-\alpha)}}}. 
\]
\end{proof}

\subsubsection{Given the smaller rank of \texorpdfstring{$\rho$}{} and \texorpdfstring{$\sigma$}{}}
In the following, we consider the case when the minimum of rank of $\rho$ and rank of $\sigma$ has an upper bound $r$.
Note that~\Cref{prop:query-error-single} still holds in this case.

\begin{proposition}\label{prop:error-operator-norm}
      Let $\alpha \in \interval[open]{0}{1}$ be a constant.
    For any density operators $\rho, \sigma \in \calD\rbra{\calH}$, positive real numbers $\eps_1, \delta_1, \eps_2 \in \interval[open]{0}{1}$, we have
\begin{equation*}
    \Abs*{\rho p_1\rbra*{\rho} p_2\rbra*{\sigma} - \frac{\delta_1^{1-\alpha}}{4}\rho^{\alpha}\sigma^{1-\alpha} } \le \frac{3}{2}\delta_1+ \eps_1 + \eps_2,
\end{equation*}
    where $p_1\coloneqq p_{1-\alpha, \varepsilon_1, \delta_1, -}$ is the polynomial specified in \cref{lem:poly_approx_neg}, and $p_2\coloneqq p_{0, 1-\alpha, \varepsilon_2 , +}$ is the polynomial specified in \cref{lem:poly_approx_pos}.
\end{proposition}

\begin{proof}
    By~\Cref{prop:query-error-single}, we have
        \[
        \Abs*{\rho p_1\rbra*{\rho} - \frac{\delta_1^{1-\alpha}}{2}\rho^{\alpha} } \leq \frac{3}{2}\delta_1+ \eps_1.
    \]
     By our choice of $p_2$, we have
    \[
        \Abs*{p_2\rbra*{\sigma} -  \frac{1}{2}\sigma^{1-\alpha}} \leq \eps_2.
    \]
   Note that 
    $ \Abs*{p_2\rbra*{\sigma}} \leq 1$.
   Therefore, we have
   \[
   \begin{aligned}
         \Abs*{\rho p_1\rbra*{\rho} p_2\rbra*{\sigma} - \frac{\delta_1^{1-\alpha}}{4}\rho^{\alpha}\sigma^{1-\alpha} } & \le 
 \Abs*{\rho p_1\rbra*{\rho} p_2\rbra*{\sigma} - \frac{1}{2}\rho p_1\rbra*{\rho}\sigma^{1-\alpha} } + \Abs*{\frac{1}{2}\rho p_1\rbra*{\rho}\sigma^{1-\alpha} - \frac{\delta_1^{1-\alpha}}{4}\rho^{\alpha}\sigma^{1-\alpha} }\\
 &\le 
 \Abs*{\rho p_1\rbra*{\rho}}\Abs*{ p_2\rbra*{\sigma} - \frac{1}{2}\sigma^{1-\alpha} } + \Abs*{\frac{1}{2}\sigma^{1-\alpha}}\Abs*{\rho p_1\rbra*{\rho} - \frac{\delta_1^{1-\alpha}}{2}\rho^{\alpha}} \\
 &\le 
\frac{3}{2}\delta_1+ \eps_1 + \eps_2.
   \end{aligned}
   \]
\end{proof}

\begin{proposition}\label{prop:query-error-total-min-rank}
    Let $X$, $\varepsilon_H$, $\eps_1$, $\eps_2$, $\delta_1$, $\delta_1'$, $\delta_2'$ be the parameters in \cref{algo:affinity-query}.
    If $\abs{X - \Prob\sbra{x = 0}} \leq \varepsilon_H$, then
    \[
        \abs*{\frac{16}{\delta_1^{1-\alpha}}\rbra*{2X-1} - \dAa{\alpha}{\rho}{\sigma}{}} 
        \leq \frac{2}{\delta_1^{1-\alpha}}\rbra*{16\varepsilon_H + 8\delta_1' + 8\delta_2'+ 3r\delta_1+ 2r\eps_1 + 2r\eps_2}.
    \]
\end{proposition}

\begin{proof}
      Suppose $\abs{X - \Prob\sbra{x = 0}} \leq \varepsilon_H$ and $U_{p_1(A)p_2(B)}$ is a $\rbra{1, 2n+n_\rho+n_\sigma+4, \delta_1'+\delta_2'}$-block-encoding of $\frac{1}{4}p_1(A)p_2(B)$, we have
    \[
        \abs*{ \rbra*{2X-1} - \Re\rbra*{\tr\rbra*{\bra{0}_{2n+n_\rho+n_\sigma+4}U_{p_1\rbra{A}p_2\rbra{B}}\ket{0}_{2n+n_\rho+n_\sigma+4} \rho}} } \leq 2\varepsilon_H ,
    \]
    which gives
    \[
        \abs*{ 4\rbra*{2X-1} - \tr\rbra*{\rho p_1\rbra*{\rho} p_2\rbra*{\sigma}} } \leq 8\varepsilon_H + 4\delta_1' + 4\delta_2'.
    \]
    By \cref{lem:tr-dif-bound,prop:error-operator-norm}, we have
    \[
    \abs*{\tr\rbra*{\rho p_1\rbra*{\rho} p_2\rbra*{\sigma}} 
    - \tr \rbra*{\frac{\delta_1^{1-\alpha}}{4}\rho^{\alpha}\sigma^{1-\alpha}} }\leq r\rbra*{\frac{3}{2}\delta_1+ \eps_1 + \eps_2},
    \]
    where $r = \min\cbra{\rank(\rho), \rank(\sigma)}$.

    Therefore, the result follows from the triangle inequality.
\end{proof}

\begin{theorem}[Query upper bound for estimating quantum Tsallis relative entropy with minimum rank bounded]\label{thm:query-complexity-of-general-affinity-min-rank}
    Let $\alpha \in \interval[open]{0}{1}$ be a constant.
    There is a quantum algorithm that, for any $\eps \in \interval[open]{0}{1}$, given query access to density operators $\rho, \sigma \in \calD(\calH)$ with  $\min\{\rank\rbra{\rho},\rank\rbra{\sigma}\}\le r$, with probability at least $2/3$, estimating $\dTsa{\alpha}{\rho}{\sigma}{}$ to within additive error $\eps$, using 
    \[
         \begin{cases}
             O\rbra*{\frac{r^{\frac{1}{\alpha}+\frac{2}{1-\alpha}-1}}{\varepsilon^{\frac{1}{\alpha}+\frac{2}{1-\alpha}}}}, \text{ if } 0< \alpha \le \frac{1}{2}, \\
             O\rbra*{\frac{r^{\frac{2}{\alpha}+\frac{1}{1-\alpha}-1}}{\varepsilon^{\frac{2}{\alpha}+\frac{1}{1-\alpha}}}},  \text{ if } \frac{1}{2}< \alpha < 1
         \end{cases}
    \]
    queries to $\calO_{\rho}$ and $\calO_{\sigma}$.
\end{theorem}
\begin{proof}
   We notice that
    $\dTsa{\alpha}{\rho}{\sigma}{} = \frac{1}{1-\alpha} \rbra{1-\dAa{\alpha}{\rho}{\sigma}{}}$.
    Therefore, 
    to obtain an estimate of $\dTsa{\alpha}{\rho}{\sigma}{}$ within additive error
    $\varepsilon$, 
    it suffices to estimate $\dAa{\alpha}{\rho}{\sigma}{}$ to within
    $\rbra{1-\alpha}\eps$ error. 
    We now show how to 
 estimate $\dAa{\alpha}{\rho}{\sigma}{}$ to within
    $\eps$ error.

      For any $\alpha\in \interval[open]{0}{1}$, setting 
  setting 
    \[
        \varepsilon_1 = \varepsilon_2= \delta_1 = \frac{\eps^{1/\alpha}}{200^{1/\alpha}r^{1/\alpha}}, \quad \eps_H = \frac{\eps^{1/\alpha}}{8 \cdot 200^{(1-\alpha)/\alpha}r^{(1-\alpha)/\alpha}},
        \quad \delta_1' = \delta_2' = \frac{\eps^{1/\alpha}}{16 \cdot 200^{(1-\alpha)/\alpha}r^{(1-\alpha)/\alpha}}
    \]
    in \cref{prop:query-error-total-min-rank}, 
    we have
    \[
     \abs*{\frac{16}{\delta_1^{1-\alpha}}\rbra{1-2X} - \dAa{\alpha}{\rho}{\sigma}{}{}} \le \varepsilon.
    \]
    
    Now we consider the sample complexity of the algorithm. 
    By our choice of parameters, we have
    \[
        d_1 = O\rbra*{\frac{1}{\delta_1}\log\rbra*{\frac{1}{\varepsilon_1}}} = O\rbra*{\frac{r^{1/\alpha}}{\varepsilon^{1/\alpha}}\log \rbra*{\frac{r}{\varepsilon}}},
    \quad 
        d_2 = O\rbra*{\frac{1}{\eps_2^{1/(1-\alpha)}}} =  O\rbra*{\frac{r^{1/\alpha(1-\alpha)}}{\varepsilon^{1/\alpha(1-\alpha)}}}.
    \]
Since we need to repeat $O(1/\eps_H)$ times, the total number of queries is
\[
O\rbra*{\frac{r^{\rbra*{(1-\alpha)^2+1}/\alpha(1-\alpha)}}{\varepsilon^{(2-\alpha)/\alpha(1-\alpha)}}} = O\rbra*{\frac{r^{\frac{2}{\alpha}+\frac{1}{1-\alpha}-1}}{\varepsilon^{\frac{2}{\alpha}+\frac{1}{1-\alpha}}}}. 
\]
Note that $\dAa{\alpha}{\rho}{\sigma}{} = \dAa{1-\alpha}{\sigma}{\rho}{}$, this yields an algorithm with query complexity
\[
O\rbra*{\frac{r^{\frac{1}{\alpha}+\frac{2}{1-\alpha}-1}}{\varepsilon^{\frac{1}{\alpha}+\frac{2}{1-\alpha}}}}, 
\]
which is no larger for $0< \alpha \le \frac{1}{2}$.
\end{proof}

\subsection{Sample complexity upper bounds}\label{sec:sample_up}

The sample complexity upper bound for estimating the quantum Tsallis relative entropy is stated as follows.

\begin{theorem}[Sample upper bound for estimating quantum Tsallis relative entropy]\label{thm:Tsallis-relative-sample}
    Let $\alpha\in \interval[open]{0}{1}$ be a constant.
    There is a quantum algorithm that, for any $\varepsilon \in \interval[open]{0}{1}$, given sample access to quantum states $\rho, \sigma \in \mathcal{D}(\mathcal{H})$ of rank at most $r$, with probability at least $2/3$, estimates $\dTsa{\alpha}{\rho}{\sigma}{}$ to within additive error $\eps$, using   
    \[
    \begin{cases}
       O\rbra*{\frac{r^{2+3\alpha}}{\varepsilon^{\frac{2}{\alpha}+\frac{3}{1-\alpha} }}\log^{2}\rbra*{\frac{r}{\varepsilon}}}, & \text{ if }
           \alpha \in \interval[open]{0}{1/2}, \\
            O\rbra*{\frac{r^{3.5}}{\varepsilon^{10}}\log^{4}\rbra*{\frac{r}{\varepsilon}}}, & \text{ if }
           \alpha = 1/2, \\
            O\rbra*{\frac{r^{5-3\alpha}}{\varepsilon^{\frac{3}{\alpha}+\frac{2}{1-\alpha} }}\log^{2}\rbra*{\frac{r}{\varepsilon}}}, & \text{ if } \alpha \in \interval[open]{1/2}{1}.
    \end{cases}
    \]
    samples of $\rho$ and $\sigma$.
\end{theorem}

Similarly to the query case, the crucial part of estimating quantum Tsallis relative entropy is the estimation of quantum affinity. The main difference is that here we use the samplizer to simulate the quantum query algorithm in \cref{thm:Tsallis-relative-query} by another quantum algorithm with sample access, albeit at the cost of obtaining only block-encodings
of $\rho/2$ and $\sigma/2$.

We first describe the algorithm as follows and formally state it in \cref{algo:affinity-sample}.

\begin{algorithm}[!htbp]
    \caption{$\mathsf{AffinityEstS}_{\alpha}(\rho, \sigma, \varepsilon_1, \varepsilon_2, \varepsilon_H, \delta, \delta_1, \delta_1', \delta_2')$}
    \label{algo:affinity-sample}
    \begin{algorithmic}[1]
        \Require Identical copies of quantum states $\rho$ and $\sigma$; parameters $\varepsilon_1, \varepsilon_2, \varepsilon_H, \delta, \delta_1, \delta_1', \delta_2' \in \interval[open]{0}{1}$.

        \Ensure An estimate of $\dAa{\alpha}{\rho}{\sigma}{}$. 

        \State Let $p_1\coloneqq p_{1-\alpha, \varepsilon_1, \delta_1, -}$ be the polynomial specified in \cref{lem:poly_approx_neg}, and $p_2\coloneqq p_{0, 1-\alpha, \varepsilon_2, +}$ be the polynomial specified in \cref{lem:poly_approx_pos}. 

        \State Let $U_{A}$ be a unitary operator that is a $\rbra{1, a, 0}$-block-encoding of $A$ and $U_B$ be a unitary operator that is a $\rbra{1, b, 0}$-block-encoding of $B$, where $A,B$ block-encode $\rho/2$, $\sigma/2$, respectively.
        
        \State Let $U_{p_1(A)} \gets \mathsf{EigenTrans} (U_A, p_1/2, \delta_1')$, and $U_{p_2(B)} \gets \mathsf{EigenTrans} (U_B, p_2/2, \delta_2')$ by \cref{lem:qsvt}.

        \State Let $U_{p_1(A)p_2(B)}\gets \mathsf{BEProduct}(U_{p_1(A)}, U_{p_2(B)})$ by \cref{lem:product_BE_matrices}.

        \State Let $U_{\textup{HT}}$ denote the part of quantum circuit of $ \mathsf{HadamardTest}(U_{p_1(A)p_2(B)}, \rho)$ without input and measurement by \cref{lem:hadamard_test}.

        \For{$i=1, 2, \dots, k=\Theta(1/\varepsilon_H^2)$}
            \State $X_i\gets$ the measurement outcome of the first qubit of $\mathsf{Samplize}_{\delta}\ave{U_{\textup{HT}}^{U_A, U_B}}[\rho, \sigma]\rbra{\ketbra{0}{0} \otimes \rho }$ in the computational basis by \Cref{thm:multi-samplizer}.
        \EndFor

        \State $X \gets \frac{1}{k}\sum_i X_i$.
        
        \State \Return $16\delta_1^{\alpha-1}(1-2X)$.
    \end{algorithmic}
\end{algorithm}

Let $U_{A}$ be a unitary operator that is a $\rbra{1, a, 0}$-block-encoding of $A$ and $U_B$ be a unitary operator that is a $\rbra{1, b, 0}$-block-encoding of $B$. 
Here, $A$ and $B$ are supposed to be $\rho/2$ and $\sigma/2$, respectively. 

\textbf{Step 1: Construct a block-encoding of $p_1(A)$ with $p_1(x) \approx \frac{\delta_1^{1-\alpha}}{2}x^{\alpha-1}$.}
Let $\varepsilon_1, \delta_1, \delta_1' \in \interval[open]{0}{1/2}$ be parameters to be determined.
By \cref{lem:poly_approx_neg}, there exists a polynomial $p_1(x)\in \mathbb{R}[x]$ of degree $d_1 = O(\frac{1}{\delta_1} \log (\frac{1}{\varepsilon_1}))$ satisfying
\[
    \abs*{p_1\rbra{x} - \frac{\delta_1^{1-\alpha}}{2} x^{\alpha-1}} \leq \varepsilon_1 \textup{ for } x \in \sbra{-1, -\delta_1} \cup \sbra{\delta_1, 1},
\]
and
\[
    \abs*{p_1\rbra{x}} \leq 1 \textup{ for } x \in \sbra{-1, 1}.
\]

By \cref{lem:qsvt}, with $p \coloneqq \frac{1}{2}p_1$, $\alpha \coloneqq 1$, $a \coloneqq a$ and $\eps \coloneqq 0$, we can implement a quantum circuit $U_{p_1\rbra{A}}$ that is a $\rbra{1,a+2,\delta_1'}$-block-encoding of $\frac{1}{2}p_1(A)$, by using $O\rbra{d_1}=O\rbra{\frac{1}{\delta_1}\log\rbra{\frac{1}{\eps_1}}}$ queries to $U_A$, and the circuit description of $U_{p_1\rbra{A}}$ can be computed in classical time $\poly\rbra{ d_1, \log\rbra{\frac{1}{\delta'_1}}}$.

\textbf{Step 2: Construct a block-encoding of $p_2\rbra{B}$ where $p_2(x)\approx \frac{1}{2} x^{1-\alpha}$.} 
Let $\varepsilon_2, \delta_2' \in \rbra{0, 1/2}$ be parameters to be determined. 
Let $p_2 \in \mathbb{R}\sbra{x}$ be a polynomial of degree $d_2 = O\rbra{\rbra{\frac{1}{\varepsilon_2}}^{\frac{1}{1-\alpha}}}$ given by \cref{lem:poly_approx_pos} such that
\[
    \abs*{p_2\rbra{x} - \frac{1}{2}x^{1-\alpha}} \leq \varepsilon_2 \textup{ for } x \in \sbra{0, 1},
\]
and
\[
    \abs*{p_2\rbra{x}} \leq 1 \textup{ for } x \in \sbra{-1, 1}.
\]
By \cref{lem:poly_approx_pos}, with $p \coloneqq \frac{1}{2}p_2$, $\alpha \coloneqq 1$, $a \coloneqq b$ and $\eps \coloneqq 0$, we can implement a quantum circuit $U_{p_2\rbra{B}}$ that is a $O(1,b+2,\delta_2')$-block-encoding of $\frac{1}{2}p_2(B)$, by using $O\rbra{d_2} = O\rbra{\rbra{\frac{1}{\varepsilon_2}}^{\frac{1}{1-\alpha}}}$ queries to $U_B$, and the circuit description of $U_{p_2\rbra{B}}$ can be computed in classical time $\poly\rbra{d_2, \log\rbra{\frac{1}{\delta_2'}}}$.

\textbf{Step 3: Construct a block-encoding of $p_1\rbra{A} p_2\rbra{B}$.} 
By~\cref{lem:product_BE_matrices},
we can implement a quantum circuit $U_{p_1\rbra{A}p_2\rbra{B}}$  
that is a $\rbra{1,a+b+4,\delta_1'+\delta_2'}$-block-encoding of $\frac{1}{4}p_1\rbra{A}p_2\rbra{B}$.

\textbf{Step 4: Estimate $\tr\rbra{\rho p_1\rbra{A} p_2\rbra{B}}$.}
By \cref{lem:hadamard_test}, we can implement a quantum circuit (family) $\mathcal{C}$ using one query to  $U_{p_1\rbra{A}p_2\rbra{B}}$ and  a sample of $\rho$ that outputs $x \in \cbra{0, 1}$ such that
\[
    \Prob\sbra{x = 0} = \frac{1 + \Re\rbra{\tr\rbra{\bra{0}_{a+b+4} U_{p_1\rbra{A}p_2\rbra{B}} \ket{0}_{a+b+4} \rho}}}{2}.
\]

\textbf{Step 5: Estimate $\tr\rbra{\rho p_1\rbra{\rho} p_2\rbra{\sigma}}$.}
Let $\delta > 0$, to be determined. 
By \cref{thm:multi-samplizer}, we consider  $\mathsf{Samplize}_{\delta}\ave*{\mathcal{C}^{U_A, U_B}}$. 
Let $\widetilde x \in \cbra{0, 1}$ be the measurement outcome of the first qubit of
\[
    \mathsf{Samplize}_{\delta}\ave*{\mathcal{C}^{U_A, U_B}}\sbra{\rho, \sigma} \rbra*{ \ketbra{0}{0} \otimes \ketbra{0}{0}^{\otimes \rbra{a+b+4}} \otimes \rho }
\]
in the computational basis. 
Then, by the definition of samplizer and the property of diamond norm, we have $
    \abs*{ \Prob\sbra{x = 0} - \Prob\sbra{\widetilde x = 0} } \leq \delta$.
Let  $\varepsilon_H \in \rbra{0, 1}$  be a precision parameter to be  determined and $k = \Theta\rbra{1/\varepsilon_H^2}$. 
Let $X_1, X_2, \dots, X_k \in \cbra{0, 1}$ be $k$ identical and independent samples of $\widetilde x$. 
Let
\[
    X = \frac{1}{k}\sum_{i=1}^k X_i.
\]

\textbf{Step 6: Return $16\delta_1^{\alpha-1}(1-2X)$ as an estimate of $\dAa{\alpha}{\rho}{\sigma}{}$.}

\vspace{10pt}

We now analyze the error and determine all the parameters in the algorithm as follows.

\begin{proposition}\label{prop:sample-error-single}
   Let $\alpha\in \interval[open]{0}{1}$ be a constant. 
   For any density operator $\rho\in \mathcal{D}(\mathcal{H})$, positive
   real numbers $\varepsilon_1, \delta_1 \in \interval[open]{0}{1}$,
   we have
    \[
        \Abs*{\rho p_1\rbra*{\frac{\rho}{2}} - \delta_1^{1-\alpha}\rbra*{\frac{\rho}{2}}^{\alpha} } \le 4\delta_1+ \varepsilon_1,
    \]
    where $p_1\coloneqq p_{1-\alpha, \varepsilon_1, \delta_1, -}$ is the polynomial specified in \cref{lem:poly_approx_neg}.
\end{proposition}

\begin{proof}
    Let $\lambda_1, \lambda_2, \dots, \lambda_k$ denote the non-zero eigenvalues of $\rho$.
    For any $j\in [k]$, if $\lambda_j\ge 2\delta_1$, by our choice of $p_1$, we have
    \[
        \abs*{p_1\rbra*{\frac{\lambda_j}{2}}-\frac{\delta_1^{1-\alpha}}{2}\rbra*{\frac{\lambda_j}{2}}^{\alpha-1}}\le \varepsilon_1.
    \]
    Note that $0 \le \lambda_j\le 1$, we conclude
    \[
        \abs*{\lambda_j p_1\rbra*{\frac{\lambda_j}{2}}-\delta_1^{1-\alpha}\rbra*{\frac{\lambda_j}{2}}^{\alpha}}\le \varepsilon_1.
    \]

    Now consider the case when $0 \le \lambda_j \le 2\delta_1$.
    In this case, we have
       \[
    \abs*{p_1\rbra*{\frac{\lambda_j}{2}}-\frac{\delta_1^{1-\alpha}}{2}\rbra*{\frac{\lambda_j}{2}}^{\alpha-1}}\le  \abs*{p_1\rbra*{\frac{\lambda_j}{2}}} + \abs*{\frac{\delta_1^{1-\alpha}}{2}\rbra*{\frac{\lambda_j}{2}}^{\alpha-1}} \le 2,
    \]
    and multiplying both sides by $\lambda_j$ gives the $4\delta_1$ upper bound.

    Combining both cases, we obtain the upper bound 
    $4\delta_1+ \varepsilon_1$ as we desired.
\end{proof}

\begin{proposition}\label{prop:sample-error-first}
    Let $\alpha\in \interval[open]{0}{1}$ be a constant. 
    For any density operators $\rho, \sigma\in \mathcal{D}(\mathcal{H})$, positive real numbers $\varepsilon_1, \delta_1, \varepsilon_2 \in \interval[open]{0}{1}$, we have
    \[
    \abs*{\tr\rbra*{\rho p_1\rbra*{\frac{\rho}{2}} p_2\rbra*{\frac{\sigma}{2}}} 
    - \tr\rbra*{\rho  \frac{\delta_1^{1-\alpha}}{2}\rbra*{\frac{\rho}{2}}^{\alpha-1}p_2\rbra*{\frac{\sigma}{2}}} } 
    \le \rbra*{r\varepsilon_2+ 2^{\alpha-2} r^{\alpha}} \rbra*{4\delta_1+ \varepsilon_1},
    \]
    where $r=\max\{\rank(\rho), \rank(\sigma)\}$, $p_1\coloneqq p_{1-\alpha, \varepsilon_1, \delta_1, -}$ is the polynomial specified in \cref{lem:poly_approx_neg}, and $p_2\coloneqq p_{0, 1-\alpha, \varepsilon_2 , +}$ is the polynomial specified in \cref{lem:poly_approx_pos}.
\end{proposition}

\begin{proof}
    By our choice of $p_2$, we know
    \[
        \Abs*{p_2\rbra*{\frac{\sigma}{2}}-\frac{1}{2}\rbra*{\frac{\sigma}{2}}^{1-\alpha}}_1 \le r\varepsilon_2.
    \]
    Let $\lambda_1, \lambda_2, \dots, \lambda_j$ denote the non-zero eigenvalues of $\sigma$ with $j\le r$.
    We have $\sum_i \lambda_i = 1$.
    By the power mean inequality, for $1-\alpha\le 1$, we have
    \[
    \rbra*{\frac{\sum_i \lambda_i^{1-\alpha}}{j}}^{\frac{1}{1-\alpha}} \le 
    \frac{\sum_i \lambda_i}{j} = \frac{1}{j},
    \]
    which gives $\sum_i \lambda_i^{1-\alpha} \le j^{\alpha} \le r^{\alpha}$. 
    This gives
    \[
        \Abs*{\frac{1}{2}\rbra*{\frac{\sigma}{2}}^{1-\alpha}}_1 \le 2^{\alpha-2} r^{\alpha}. 
    \]
    Combining the above, by the triangle inequality, we can get 
    \[
        \Abs*{p_2\rbra*{\frac{\sigma}{2}}}_1 \le  \Abs*{p_2\rbra*{\frac{\sigma}{2}}-\frac{1}{2}\rbra*{\frac{\sigma}{2}}^{1-\alpha}}_1 +   \Abs*{\frac{1}{2}\rbra*{\frac{\sigma}{2}}^{1-\alpha}}_1 \le r\varepsilon_2+ 2^{\alpha-2} r^{\alpha}.
    \]

    Now we have
    \[
    \begin{aligned}
    \abs*{\tr\rbra*{\rho p_1\rbra*{\frac{\rho}{2}} p_2\rbra*{\frac{\sigma}{2}}} 
    - \tr\rbra*{\rho  \frac{\delta_1^{1-\alpha}}{2}\rbra*{\frac{\rho}{2}}^{\alpha-1}  p_2\rbra*{\frac{\sigma}{2}}} }
    & = \tr\rbra*{p_2\rbra*{\frac{\sigma}{2}}
    \rbra*{\rho p_1\rbra*{\frac{\rho}{2}}-
    \rho \frac{\delta_1^{1-\alpha}}{2}\rbra*{\frac{\rho}{2}}^{\alpha-1} }}\\
    & \le \Abs*{p_2\rbra*{\frac{\sigma}{2}}}_1  \Abs*{\rho p_1\rbra*{\frac{\rho}{2}} - \delta_1^{1-\alpha}\rbra*{\frac{\rho}{2}}^{\alpha} } \\
    & \le \rbra*{r\varepsilon_2+ 2^{\alpha-2} r^{\alpha}} \rbra*{4\delta_1+ \varepsilon_1},
    \end{aligned}
    \]
    where in the second line we use the matrix H\"{o}lder inequality (\cref{fact:holder}) and the fourth line is obtained by applying \cref{prop:sample-error-single}.
\end{proof}

\begin{proposition}\label{prop:sample-error-second}
   Let $\alpha\in \interval[open]{0}{1}$ be a constant. 
   For any density operators $\rho, \sigma\in \mathcal{D}(\mathcal{H})$, 
   positive real numbers $\varepsilon_1, \delta_1, \varepsilon_2 \in 
   \interval[open]{0}{1}$,
   we have
    \[
    \abs*{  \tr\rbra*{ \delta_1^{1-\alpha}\rbra*{\frac{\rho}{2}}^{\alpha}p_2\rbra*{\frac{\sigma}{2}}}- \tr \rbra*{\frac{\delta_1^{1-\alpha}}{4}\rho^{\alpha}\sigma^{1-\alpha}} } 
    \le \frac{\delta_1^{1-\alpha}}{2^{\alpha}} r^{1-\alpha} \varepsilon_2,
    \]
    where $r=\max\{\rank(\rho), \rank(\sigma)\}$, $p_1\coloneqq p_{1-\alpha, \varepsilon_1, \delta_1, -}$ is the polynomial specified in \cref{lem:poly_approx_neg}, and $p_2\coloneqq p_{0, 1-\alpha, \varepsilon_2 , +}$ is the polynomial specified in \cref{lem:poly_approx_pos}.
\end{proposition}

\begin{proof}
    This follows a similar reasoning as in \cref{prop:sample-error-first}.
    First, 
    by the power mean inequality, we have
    \[
        \Abs*{\rho^{\alpha}}_1 \le r^{1-\alpha}.
    \]
    By our choice of $p_2$, we have
    \[
        \Abs*{p_2\rbra*{\frac{\sigma}{2}} -  \frac{1}{2}\rbra*{\frac{\sigma}{2}}^{1-\alpha}}\le \varepsilon_2.
    \]
    Therefore, we deduce
    \[
    \begin{aligned}
       \abs*{  \tr\rbra*{ \delta_1^{1-\alpha}\rbra*{\frac{\rho}{2}}^{\alpha}p_2\rbra*{\frac{\sigma}{2}}}- \tr \rbra*{\frac{\delta_1^{1-\alpha}}{4}\rho^{\alpha}\sigma^{1-\alpha}} }
    & = \frac{\delta_1^{1-\alpha}}{2^{\alpha}} \tr \rbra*{\rho^{\alpha} \rbra*{p_2\rbra*{\frac{\sigma}{2}}-\frac{1}{2}\rbra*{\frac{\sigma}{2}}^{1-\alpha}}} \\
    &\le  \frac{\delta_1^{1-\alpha}}{2^{\alpha}} \Abs*{\rho^{\alpha}}_1  
    \Abs*{p_2\rbra*{\frac{\sigma}{2}}-\frac{1}{2}\rbra*{\frac{\sigma}{2}}^{1-\alpha}} \\
      &\le  \frac{\delta_1^{1-\alpha}}{2^{\alpha}} r^{1-\alpha} \varepsilon_2,
    \end{aligned}
    \]
    where the second line is obtained by the matrix H\"{o}lder inequality (\cref{fact:holder}).
\end{proof}

\begin{proposition}\label{prop:sample-error-total}
    Let $X$,  $\varepsilon_H$, $\delta$, $\delta_1$,
    $\delta_1'$, $\delta_2'$ be the parameters in \cref{algo:affinity-sample}.
    If $\abs{X - \Prob\sbra{\widetilde{x} = 1}} \leq \varepsilon_H$, then
    \begin{align*}
        \abs*{\frac{16}{\delta_1^{1-\alpha}}\rbra{1-2X} - \dAa{\alpha}{\rho}{\sigma}{}{}} 
        \le & \frac{16}{\delta_1^{1-\alpha}} \rbra*{2\rbra*{\varepsilon_H + \delta} + \delta_1' + \delta_2'} + \\
        & \quad \rbra*{r\varepsilon_2+ 2^{\alpha-2} r^{\alpha}} \rbra*{16\delta_1^{\alpha}+ \frac{4\varepsilon_1}{\delta_1^{1-\alpha}}} + 2^{2-\alpha}r^{1-\alpha} \varepsilon_2.
    \end{align*}
\end{proposition}

\begin{proof}
Let $\widetilde x \in \cbra{0, 1}$ be the measurement outcome of 
\[
\mathsf{Samplize}_{\delta}\ave*{\mathcal{C}^{U_A, U_B}}\sbra{\rho, \sigma} \rbra*{ \ketbra{0}{0} \otimes \ketbra{0}{0}^{\otimes \rbra{a+b+4}} \otimes \rho }
\]
In the computational basis. 
Then, by the definition of samplizer and the property of diamond norm,
we have
\[
    \abs*{ \Prob\sbra{x = 0} - \Prob\sbra{\widetilde x = 0} } \leq \delta,
\]
where
\[
    \Prob\sbra{x = 0} = \frac{1+\Re\rbra*{\tr\rbra*{\bra{0}_{a+b+4}U_{p_1\rbra{A}p_2\rbra{B}}\ket{0}_{a+b+4} \rho}}}{2},
\]
and 
\[
\abs*{\tr\rbra*{\bra{0}_{a+b+4}U_{p_1\rbra{A}p_2\rbra{B}}\ket{0}_{a+b+4} \rho}
-\frac{1}{4} \tr \rbra*{\rho p_1\rbra*{\frac{\rho}{2}} p_2\rbra*{\frac{\sigma}{2}}}}
\le \delta_1'+\delta_2'.
\]
By Hoeffding's inequality~\cite{Hoe63}, we have
\[
    \Prob\sbra*{ \abs*{X - \Prob\sbra{\widetilde x = 1}} \leq \varepsilon_H } \geq \frac{3}{4},
\]
for $k = \Theta(1/\varepsilon_H^2)$.

    By our assumption, we have $\abs{X - \Prob\sbra{\widetilde{x} = 1}} \leq \varepsilon_H$.
    Since $U_{p_1(A)p_2(B)}$ is a $\rbra{1, a+b+4, \delta_1'+\delta_2'}$-block-encoding of $\frac{1}{4}p_1(A)p_2(B)$, we have
    \[
        \abs*{ \rbra{1-2X} - \Re\rbra*{\tr\rbra*{\bra{0}_{a+b+4}U_{p_1\rbra{\rho}p_2\rbra{\sigma}}\ket{0}_{a+b+4} \rho}} } \leq 2\rbra*{\varepsilon_H + \delta},
    \]
    which gives
    \[
        \abs*{ 4\rbra{1-2X} - \tr\rbra*{\rho p_1\rbra*{\frac{\rho}{2}} p_2\rbra*{\frac{\sigma}{2}}} } \leq 8\rbra*{\varepsilon_H + \delta} + 4\rbra*{\delta_1' + \delta_2'}.
    \]
    By \cref{prop:sample-error-first,prop:sample-error-second}, we have
    \[
    \abs*{\tr\rbra*{\rho p_1\rbra*{\frac{\rho}{2}} p_2\rbra*{\frac{\sigma}{2}}} 
    - \tr \rbra*{\frac{\delta_1^{1-\alpha}}{4}\rho^{\alpha}\sigma^{1-\alpha}} }\le \rbra*{r\varepsilon_2+ 2^{\alpha-2} r^{\alpha}} \rbra*{4\delta_1+ \varepsilon_1} + \frac{\delta_1^{1-\alpha}}{2^{\alpha}} r^{1-\alpha} \varepsilon_2.
    \]
    Therefore, the result follows from the triangle inequality.
\end{proof}

\begin{theorem}[Sample upper bound for estimating quantum affinity]\label{thm:sample-complexity-of-general-affinity}
    Let $\alpha \in \interval[open]{0}{1}$ be a constant.
    There is a quantum algorithm $\mathsf{AffinityEstS}_{\alpha}(\rho,\sigma, r, \varepsilon)$ that, for any $\varepsilon \in \interval[open]{0}{1}$, given sample access to quantum states $\rho, \sigma \in \mathcal{D}(\mathcal{H})$ of rank at most $r$, with probability at least $2/3$, estimating $\dAa{\alpha}{\rho}{\sigma}{}$ to within additive error $\varepsilon$, using 
    \[
    \begin{cases}
           O\rbra*{\frac{r^{2+3\alpha}}{\varepsilon^{\frac{2}{\alpha}+\frac{3}{1-\alpha} }}\log^{2}\rbra*{\frac{r}{\varepsilon}}}, & \text{ if }
           \alpha \in \interval[open]{0}{1/2}, \\
            O\rbra*{\frac{r^{3.5}}{\varepsilon^{10}}\log^{4}\rbra*{\frac{r}{\varepsilon}}}, & \text{ if }
           \alpha = 1/2, \\
            O\rbra*{\frac{r^{5-3\alpha}}{\varepsilon^{\frac{3}{\alpha}+\frac{2}{1-\alpha} }}\log^{2}\rbra*{\frac{r}{\varepsilon}}}, & \text{ if } \alpha \in \interval[open]{1/2}{1}.
    \end{cases}
    \]
    samples of $\rho$ and $\sigma$. 
\end{theorem}

\begin{proof}
    For any $\alpha\in \interval[open]{0}{1}$, setting 
    \[
        \eps_1 =\delta_1 =\frac{\eps^{\frac{1}{\alpha}}}{40^{\frac{1}{\alpha}}r}, \quad
        \eps_2 = \frac{r^{\alpha-1}}{8}\eps, \quad
        \delta = \eps_H = \frac{\eps^{\frac{1}{\alpha}}}{256 \cdot 40^{\frac{1}{\alpha}}r^{1-\alpha}}, \quad
        \delta_1' = \delta_2' = \frac{\eps^{\frac{1}{\alpha}}}{128 \cdot 40^{\frac{1}{\alpha}}r^{1-\alpha}}    
    \]
    in \cref{prop:sample-error-total},
    we have
    \[
     \abs*{\frac{16}{\delta_1^{1-\alpha}}\rbra{1-2X} - \dAa{\alpha}{\rho}{\sigma}{}{}} \le \varepsilon.
    \]

    Now we consider the sample complexity of the algorithm. 
    By our choice of parameters, we have
    \[
        d_1 = O\rbra*{\frac{1}{\delta_1}\log\rbra*{\frac{1}{\varepsilon_1}}} = O\rbra*{\frac{r}{\varepsilon^{1/\alpha}}\log \rbra*{\frac{r}{\varepsilon}}},
    \quad 
        d_2 = O\rbra*{\frac{1}{\eps_2^{1/(1-\alpha)}}} =  O\rbra*{\frac{r}{\varepsilon^{1/(1-\alpha)}}}.
    \]
    We then discuss the complexity by case.

    \textbf{Case 1: $\alpha \in \interval[open left]{0}{1/2}$.}
    In this case, we have $d_2 = O(d_1)$.
    Then, the samplizer uses
    \[
        O\rbra*{\frac{\rbra{d_1+d_2}^2}{\delta} \log^2\rbra*{\frac{d_1+d_2}{\delta}}} = O\rbra*{\frac{r^{3-\alpha}}{\varepsilon^{3/\alpha}}\log^{4}\rbra*{\frac{r}{\varepsilon}}}
    \]
    samples. Since we need to repeat $O(1/\varepsilon_H^2)$ times, the total sample complexity is 
    \[
        O\rbra*{\frac{r^{5-3\alpha}}{\varepsilon^{5/\alpha}}\log^{4}\rbra*{\frac{r}{\varepsilon}}}.
    \]

    \textbf{Case 2: $\alpha \in \interval[open]{1/2}{1}$.}
    In this case, we have $d_1 = O(d_2)$.
    Then, the samplizer uses
    \[
        O\rbra*{\frac{\rbra{d_1+d_2}^2}{\delta} \log^2\rbra*{\frac{d_1+d_2}{\delta}}} = O\rbra*{\frac{r^{3-\alpha}}{\varepsilon^{2/(1-\alpha)+ 1/\alpha}}\log^{2}\rbra*{\frac{r}{\varepsilon}}}
    \]
    samples. Since we need to repeat $O(1/\varepsilon_H^2)$ times,
    the total sample complexity is 
    \[
        O\rbra*{\frac{r^{5-3\alpha}}{\varepsilon^{2/(1-\alpha)+ 3/\alpha}}\log^{2}\rbra*{\frac{r}{\varepsilon}}}.
    \]

    Now, note that $\dAa{\alpha}{\rho}{\sigma}{} = \dAa{1-\alpha}{\sigma}{\rho}{}$. Therefore, for $\alpha \in \interval[open]{0}{1/2}$, we also have an algorithm with sample complexity
    \[
        O\rbra*{\frac{r^{2+3\alpha}}{\varepsilon^{2/\alpha+ 3/(1-\alpha)}}\log^{2}\rbra*{\frac{r}{\varepsilon}}}.
    \]
    Similarly, for $\alpha \in \interval[open]{1/2}{1}$, we also have an algorithm with sample complexity
    \[
        O\rbra*{\frac{r^{2+3\alpha}}{\varepsilon^{5/(1-\alpha)}}\log^{4}\rbra*{\frac{r}{\varepsilon}}}.
    \]

    Combining the above discussions, the sample complexity of the algorithm is
    \[
    \begin{cases}
           O\rbra*{\dfrac{r^{2+3\alpha}}{\varepsilon^{2/\alpha+ 3/(1-\alpha)}}\log^{2}\rbra*{\dfrac{r}{\varepsilon}}}, & \text{ if }
           \alpha \in \interval[open]{0}{1/2}, \\
           O\rbra*{\dfrac{r^{3.5}}{\varepsilon^{10}}\log^{4}\rbra*{\dfrac{r}{\varepsilon}}}, & \text{ if }
           \alpha = 1/2, \\
            O\rbra*{\dfrac{r^{5-3\alpha}}{\varepsilon^{2/(1-\alpha)+ 3/\alpha}}\log^{2}\rbra*{\dfrac{r}{\varepsilon}}}, & \text{ if } \alpha \in \interval[open]{1/2}{1}. 
    \end{cases} 
    \]
    These yield the proof.
\end{proof}

\cref{algo:affinity-sample} can be applied to estimating the Tsallis relative entropy of quantum states.

\begin{proof}[Proof of \cref{thm:Tsallis-relative-sample}]
    We notice that $\dTsa{\alpha}{\rho}{\sigma}{} = \frac{1}{1-\alpha} (1-\dAa{\alpha}{\rho}{\sigma}{}) $.
    Therefore, to obtain an estimate of $\dTsa{\alpha}{\rho}{\sigma}{}$ within additive error $\varepsilon$, it suffices to estimate $\dAa{\alpha}{\rho}{\sigma}{}$ to within additive error $(1-\alpha)\varepsilon$. 
    The claim follows from using the algorithm $\mathsf{AffinityEstS}_{\alpha}(\rho , \sigma, r, (1-\alpha)\varepsilon)$ and applying \cref{thm:sample-complexity-of-general-affinity}.
\end{proof}

\subsection{Sample complexity with weaker rank conditions}

In this part, we consider the sample complexity of estimating the quantum Tsallis relative entropy when we can bound the rank of 
only one density operator.

\subsubsection{Given the rank of \texorpdfstring{$\rho$}{}}

We consider the case when the rank of $\rho$ is bounded.

Note that \Cref{prop:sample-error-single,prop:sample-error-second} still holds in our case. For readability, we rewrite \Cref{prop:sample-error-second} as follows.
\begin{proposition}[Restatement of \cref{prop:sample-error-second}]\label{prop:sample-error-second-rank-rho}
   Let $\alpha\in \interval[open]{0}{1}$ be a constant. 
   For any density operators $\rho, \sigma\in \mathcal{D}(\mathcal{H})$, 
   positive real numbers $\varepsilon_1, \delta_1, \varepsilon_2 \in 
   \interval[open]{0}{1}$,
   we have
    \[
    \abs*{  \tr\rbra*{ \delta_1^{1-\alpha}\rbra*{\frac{\rho}{2}}^{\alpha}p_2\rbra*{\frac{\sigma}{2}}}- \tr \rbra*{\frac{\delta_1^{1-\alpha}}{4}\rho^{\alpha}\sigma^{1-\alpha}} } 
    \le \frac{\delta_1^{1-\alpha}}{2^{\alpha}} r_{\rho}^{1-\alpha} \varepsilon_2,
    \]
    where $r_{\rho}\ge\rank(\rho)$, $p_1\coloneqq p_{1-\alpha, \varepsilon_1, \delta_1, -}$ is the polynomial specified in \cref{lem:poly_approx_neg}, and $p_2\coloneqq p_{0, 1-\alpha, \varepsilon_2 , +}$ is the polynomial specified in \cref{lem:poly_approx_pos}.
\end{proposition}

\begin{proposition}\label{prop:sample-error-first-rank-rho}
    Let $\alpha\in \interval[open]{0}{1}$ be a constant. 
    For any density operators $\rho, \sigma\in \mathcal{D}(\mathcal{H})$, positive real numbers $\varepsilon_1, \delta_1, \varepsilon_2 \in \interval[open]{0}{1}$, we have
    \[
    \abs*{\tr\rbra*{\rho p_1\rbra*{\frac{\rho}{2}} p_2\rbra*{\frac{\sigma}{2}}} 
    - \tr\rbra*{\rho  \frac{\delta_1^{1-\alpha}}{2}\rbra*{\frac{\rho}{2}}^{\alpha-1}p_2\rbra*{\frac{\sigma}{2}}} } 
    \le \rbra*{\varepsilon_2+ 2^{\alpha-2} } \rbra*{4\delta_1+ \varepsilon_1}r_{\rho},
    \]
    where $r_{\rho}\ge \rank(\rho)$, $p_1\coloneqq p_{1-\alpha, \varepsilon_1, \delta_1, -}$ is the polynomial specified in \cref{lem:poly_approx_neg}, and $p_2\coloneqq p_{0, 1-\alpha, \varepsilon_2 , +}$ is the polynomial specified in \cref{lem:poly_approx_pos}.
\end{proposition}

\begin{proof}
    By our choice of $p_2$, we know
    \[
        \Abs*{p_2\rbra*{\frac{\sigma}{2}}-\frac{1}{2}\rbra*{\frac{\sigma}{2}}^{1-\alpha}} \le \varepsilon_2.
    \]
    Also, we have
    \[
        \Abs*{\frac{1}{2}\rbra*{\frac{\sigma}{2}}^{1-\alpha}} \le 2^{\alpha-2}. 
    \]
    Combining the above, by the triangle inequality, we can get 
    \[
        \Abs*{p_2\rbra*{\frac{\sigma}{2}}} 
        \le  \Abs*{p_2\rbra*{\frac{\sigma}{2}} - \frac{1}{2} \rbra*{\frac{\sigma}{2}}^{1-\alpha}} +  \Abs*{\frac{1}{2}\rbra*{\frac{\sigma}{2}}^{1-\alpha}}
        \le\varepsilon_2+ 2^{\alpha-2}.
    \]

    Now we have
    \[
    \begin{aligned}
    \abs*{\tr\rbra*{\rho p_1\rbra*{\frac{\rho}{2}} p_2\rbra*{\frac{\sigma}{2}}} 
    - \tr\rbra*{\rho  \frac{\delta_1^{1-\alpha}}{2}\rbra*{\frac{\rho}{2}}^{\alpha-1}  p_2\rbra*{\frac{\sigma}{2}}} }
    & = \tr\rbra*{p_2\rbra*{\frac{\sigma}{2}}
    \rbra*{\rho p_1\rbra*{\frac{\rho}{2}}-
    \rho \frac{\delta_1^{1-\alpha}}{2}\rbra*{\frac{\rho}{2}}^{\alpha-1} }}\\
    & \le \Abs*{p_2\rbra*{\frac{\sigma}{2}}}  \Abs*{\rho p_1\rbra*{\frac{\rho}{2}} - \delta_1^{1-\alpha}\rbra*{\frac{\rho}{2}}^{\alpha} }_1 \\
    & \le \rbra*{\varepsilon_2+ 2^{\alpha-2}} \rbra*{4\delta_1+ \varepsilon_1}r_{\rho},
    \end{aligned}
    \]
    where the second line is obtained by the matrix H\"{o}lder inequality (\cref{fact:holder}), and the third line is obtained by applying \cref{prop:sample-error-single}.
\end{proof}

\begin{proposition}\label{prop:sample-error-total-rank-rho}
    Let $X$,  $\varepsilon_H$, $\delta$, $\delta_1$,
    $\delta_1'$, $\delta_2'$ be the parameters in \cref{algo:affinity-sample}.
    If $\abs{X - \Prob\sbra{\widetilde{x} = 1}} \leq \varepsilon_H$, then
    \begin{align*}
        \abs*{\frac{16}{\delta_1^{1-\alpha}}\rbra{1-2X} - \dAa{\alpha}{\rho}{\sigma}{}{}} 
        \le & \frac{16}{\delta_1^{1-\alpha}} \rbra*{2\rbra*{\varepsilon_H + \delta} + \delta_1' + \delta_2'} + \\
        & \quad \rbra*{r_{\rho}\varepsilon_2+ 2^{\alpha-2} r_{\rho}} \rbra*{16\delta_1^{\alpha}+ \frac{4\varepsilon_1}{\delta_1^{1-\alpha}}} + 2^{2-\alpha}r_{\rho}^{1-\alpha} \varepsilon_2.
    \end{align*}
\end{proposition}

\begin{proof}
    In the same reasoning as in~\Cref{prop:sample-error-total}, we have
    \[
        \abs*{ 4\rbra{1-2X} - \tr\rbra*{\rho p_1\rbra*{\frac{\rho}{2}} p_2\rbra*{\frac{\sigma}{2}}} } \leq 8\rbra*{\varepsilon_H + \delta} + 4\rbra*{\delta_1' + \delta_2'}.
    \]

    By \cref{prop:sample-error-first-rank-rho,prop:sample-error-second-rank-rho}, we have
    \[
    \abs*{\tr\rbra*{\rho p_1\rbra*{\frac{\rho}{2}} p_2\rbra*{\frac{\sigma}{2}}} 
    - \tr \rbra*{\frac{\delta_1^{1-\alpha}}{4}\rho^{\alpha}\sigma^{1-\alpha}} }\le \rbra*{r_{\rho}\varepsilon_2+ 2^{\alpha-2} r_{\rho}} \rbra*{4\delta_1+ \varepsilon_1} + \frac{\delta_1^{1-\alpha}}{2^{\alpha}} r_{\rho}^{1-\alpha} \varepsilon_2.
    \]
    Therefore, the result follows from the triangle inequality.
\end{proof}

\begin{theorem}[Sample upper bound for estimating quantum Tsallis relative entropy with rank bounds on $\rho$] \label{thm:sample-complexity-of-general-affinity-rank-rho}
    Let $\alpha \in \interval[open]{0}{1}$ be a constant.
    There is a quantum algorithm that, for any $\varepsilon \in \interval[open]{0}{1}$, given sample access to quantum states $\rho, \sigma \in \mathcal{D}(\mathcal{H})$ with the rank of $\rho$ upper bounded by $r$, with probability at least $2/3$, estimating $\dTsa{\alpha}{\rho}{\sigma}{}$ to within additive error $\varepsilon$, using 
    \[
    O\rbra*{\frac{r^{\frac{5}{\alpha}-3}}{\varepsilon^{\frac{5}{\alpha}}}\log^{4}\rbra*{\frac{r}{\varepsilon}} + \frac{r^{ \frac{3}{\alpha}-1}}{\varepsilon^{ \frac{3}{\alpha} + \frac{2}{1-\alpha}}}\log^{2}\rbra*{\frac{r}{\varepsilon}}}
    \]
    samples of $\rho$ and $\sigma$. 
\end{theorem}

\begin{proof}
   We notice that
    $\dTsa{\alpha}{\rho}{\sigma}{} = \frac{1}{1-\alpha} \rbra{1-\dAa{\alpha}{\rho}{\sigma}{}}$.
    Therefore, 
    to obtain an estimate of $\dTsa{\alpha}{\rho}{\sigma}{}$ within additive error
    $\varepsilon$, 
    it suffices to estimate $\dAa{\alpha}{\rho}{\sigma}{}$ to within
    $\rbra{1-\alpha}\eps$ error. 
    We now show how to 
 estimate $\dAa{\alpha}{\rho}{\sigma}{}$ to within
    $\eps$ error.

    For any $\alpha\in \interval[open]{0}{1}$, setting 
    \[
        \eps_1 =\delta_1 =\frac{\eps^{\frac{1}{\alpha}}}{(40r)^{\frac{1}{\alpha}}}, \quad
        \eps_2 = \frac{\eps }{32 r^{1-\alpha}}, \quad
        \delta = \eps_H = \frac{\eps^{\frac{1}{\alpha}}}{256 \cdot 40^{\frac{1}{\alpha}}r^{(1-\alpha)/\alpha}}, \quad
        \delta_1' = \delta_2' = \frac{\eps^{\frac{1}{\alpha}}}{128 \cdot 40^{\frac{1}{\alpha}}r^{(1-\alpha)/\alpha}}    
    \]
    in \cref{prop:sample-error-total-rank-rho}, 
    we have
    \[
     \abs*{\frac{16}{\delta_1^{1-\alpha}}\rbra{1-2X} - \dAa{\alpha}{\rho}{\sigma}{}{}} \le \varepsilon.
    \]

    Now we consider the sample complexity of the algorithm. 
    By our choice of parameters, we have
    \[
        d_1 = O\rbra*{\frac{1}{\delta_1}\log\rbra*{\frac{1}{\varepsilon_1}}} = O\rbra*{\frac{r^{1/\alpha}}{\varepsilon^{1/\alpha}}\log \rbra*{\frac{r}{\varepsilon}}},
    \quad 
        d_2 = O\rbra*{\frac{1}{\eps_2^{1/(1-\alpha)}}} =  O\rbra*{\frac{r}{\varepsilon^{1/(1-\alpha)}}}.
    \]
    Then, the samplizer uses
    \[
        O\rbra*{\frac{\rbra{d_1+d_2}^2}{\delta} \log^2\rbra*{\frac{d_1+d_2}{\delta}}} = O\rbra*{\frac{r^{3/\alpha-1}}{\varepsilon^{3/\alpha}}\log^{4}\rbra*{\frac{r}{\varepsilon}} + \frac{r^{1+ 1/\alpha}}{\varepsilon^{2/(1-\alpha)+ 1/\alpha}}\log^{2}\rbra*{\frac{r}{\varepsilon}} }
    \]
    samples.
    Since we need to repeat $O(1/\varepsilon_H^2)$ times, the total sample complexity is 
    \[
        O\rbra*{\frac{r^{5/\alpha-3}}{\varepsilon^{5/\alpha}}\log^{4}\rbra*{\frac{r}{\varepsilon}} + \frac{r^{ 3/\alpha-1}}{\varepsilon^{2/(1-\alpha)+ 3/\alpha}}\log^{2}\rbra*{\frac{r}{\varepsilon}}}.
    \]
\end{proof}

\subsubsection{Given the rank of \texorpdfstring{$\sigma$}{}}

We consider the case when the rank of $\sigma$ is bounded.

Note that \Cref{prop:sample-error-single,prop:sample-error-first} still holds in our case. For readability, we rewrite \Cref{prop:sample-error-first} as follows.

\begin{proposition}[Restatement of \cref{prop:sample-error-first}]\label{prop:sample-error-first-rank-sigma}
    Let $\alpha\in \interval[open]{0}{1}$ be a constant. 
    For any density operators $\rho, \sigma\in \mathcal{D}(\mathcal{H})$, positive real numbers $\varepsilon_1, \delta_1, \varepsilon_2 \in \interval[open]{0}{1}$, we have
    \[
    \abs*{\tr\rbra*{\rho p_1\rbra*{\frac{\rho}{2}} p_2\rbra*{\frac{\sigma}{2}}} 
    - \tr\rbra*{\rho  \frac{\delta_1^{1-\alpha}}{2}\rbra*{\frac{\rho}{2}}^{\alpha-1}p_2\rbra*{\frac{\sigma}{2}}} } 
    \le \rbra*{r_{\sigma}\varepsilon_2+ 2^{\alpha-2} r_{\sigma}^{\alpha}} \rbra*{4\delta_1+ \varepsilon_1},
    \]
    where $r_{\sigma}\ge \rank(\sigma)$, $p_1\coloneqq p_{1-\alpha, \varepsilon_1, \delta_1, -}$ is the polynomial specified in \cref{lem:poly_approx_neg}, and $p_2\coloneqq p_{0, 1-\alpha, \varepsilon_2 , +}$ is the polynomial specified in \cref{lem:poly_approx_pos}.
\end{proposition}

If the rank of $\rho$ can not be bounded, we still have the following slightly weaker bound.

\begin{proposition}\label{prop:sample-error-second-rank-sigma}
   Let $\alpha\in \interval[open]{0}{1}$ be a constant. 
   For any density operators $\rho, \sigma\in \mathcal{D}(\mathcal{H})$, 
   positive real numbers $\varepsilon_1, \delta_1, \varepsilon_2 \in 
   \interval[open]{0}{1}$,
   we have
    \[
    \abs*{  \tr\rbra*{ \delta_1^{1-\alpha}\rbra*{\frac{\rho}{2}}^{\alpha}p_2\rbra*{\frac{\sigma}{2}}}- \tr \rbra*{\frac{\delta_1^{1-\alpha}}{4}\rho^{\alpha}\sigma^{1-\alpha}} } 
    \le \frac{\delta_1^{1-\alpha}}{2^{\alpha}} r_{\sigma} \varepsilon_2,
    \]
    where $r_{\sigma}\ge  \rank(\sigma)$, $p_1\coloneqq p_{1-\alpha, \varepsilon_1, \delta_1, -}$ is the polynomial specified in \cref{lem:poly_approx_neg}, and $p_2\coloneqq p_{0, 1-\alpha, \varepsilon_2 , +}$ is the polynomial specified in \cref{lem:poly_approx_pos}.
\end{proposition}

\begin{proof}
    First, it is direct to see that
    \[
        \Abs*{\rho^{\alpha}} \le 1.
    \]
    By our choice of $p_2$, we have
    \[
        \Abs*{p_2\rbra*{\frac{\sigma}{2}} -  \frac{1}{2}\rbra*{\frac{\sigma}{2}}^{1-\alpha}}_1 \le r_{\sigma}\varepsilon_2.
    \]
    Therefore, we deduce
    \[
    \begin{aligned}
       \abs*{  \tr\rbra*{ \delta_1^{1-\alpha}\rbra*{\frac{\rho}{2}}^{\alpha}p_2\rbra*{\frac{\sigma}{2}}}- \tr \rbra*{\frac{\delta_1^{1-\alpha}}{4}\rho^{\alpha}\sigma^{1-\alpha}} }
    & \le \frac{\delta_1^{1-\alpha}}{2^{\alpha}} \tr \rbra*{\rho^{\alpha} \rbra*{p_2\rbra*{\frac{\sigma}{2}}-\frac{1}{2}\rbra*{\frac{\sigma}{2}}^{1-\alpha}}} \\
    &\le  \frac{\delta_1^{1-\alpha}}{2^{\alpha}} \Abs*{\rho^{\alpha}}_1  
    \Abs*{p_2\rbra*{\frac{\sigma}{2}}-\frac{1}{2}\rbra*{\frac{\sigma}{2}}^{1-\alpha}} \\
      &\le  \frac{\delta_1^{1-\alpha}}{2^{\alpha}} r_{\sigma} \varepsilon_2,
    \end{aligned}
    \]
    where the second line is obtained by the matrix H\"{o}lder inequality (\cref{fact:holder}).
\end{proof}

\begin{proposition}\label{prop:sample-error-total-rank-sigma}
    Let $X$,  $\varepsilon_H$, $\delta$, $\delta_1$,
    $\delta_1'$, $\delta_2'$ be the parameters in \cref{algo:affinity-sample}.
    If $\abs{X - \Prob\sbra{\widetilde{x} = 1}} \leq \varepsilon_H$, then
    \begin{align*}
        \abs*{\frac{16}{\delta_1^{1-\alpha}}\rbra{1-2X} - \dAa{\alpha}{\rho}{\sigma}{}{}} 
        \le {} & \frac{16}{\delta_1^{1-\alpha}} \rbra*{2\rbra*{\varepsilon_H + \delta} + \delta_1' + \delta_2'} + \\
        & \quad \rbra*{r_{\sigma}\varepsilon_2+ 2^{\alpha-2} r_{\sigma}^{\alpha}} \rbra*{16\delta_1^{\alpha}+ \frac{4\varepsilon_1}{\delta_1^{1-\alpha}}} + 2^{2-\alpha}r_{\sigma} \varepsilon_2,
    \end{align*}
     where $r_{\sigma}\ge  \rank(\sigma)$.
\end{proposition}

\begin{proof}
    In the same reasoning as in~\Cref{prop:sample-error-total}, we have
    \[
        \abs*{ 4\rbra{1-2X} - \tr\rbra*{\rho p_1\rbra*{\frac{\rho}{2}} p_2\rbra*{\frac{\sigma}{2}}} } \leq 8\rbra*{\varepsilon_H + \delta} + 4\rbra*{\delta_1' + \delta_2'}.
    \]

    By \cref{prop:sample-error-first-rank-sigma,prop:sample-error-second-rank-sigma}, we have
    \[
    \abs*{\tr\rbra*{\rho p_1\rbra*{\frac{\rho}{2}} p_2\rbra*{\frac{\sigma}{2}}} 
    - \tr \rbra*{\frac{\delta_1^{1-\alpha}}{4}\rho^{\alpha}\sigma^{1-\alpha}} }
    \le \rbra*{r_{\sigma}\varepsilon_2+ 2^{\alpha-2} r_{\sigma}^{\alpha}} \rbra*{4\delta_1+ \varepsilon_1} + \frac{\delta_1^{1-\alpha}}{2^{\alpha}} r_{\sigma} \varepsilon_2.
    \]
    Therefore, the result follows from the triangle inequality.
\end{proof}

\begin{theorem}[Sample upper bound for estimating quantum Tsallis relative with rank bounds on $\sigma$]\label{thm:sample-complexity-of-general-affinity-rank-sigma}
    Let $\alpha \in \interval[open]{0}{1}$ be a constant.
    There is a quantum algorithm that, for any $\varepsilon \in \interval[open]{0}{1}$, given sample access to quantum states $\rho, \sigma \in \mathcal{D}(\mathcal{H})$ with the rank of $\sigma$ upper bounded by $r$, with probability at least $2/3$, estimating $\dTsa{\alpha}{\rho}{\sigma}{}$ to within additive error $\varepsilon$, using 
    \[
 O\rbra*{\frac{r^{5-3\alpha}}{\varepsilon^{\frac{5}{\alpha}}}\log^{4}\rbra*{\frac{r}{\varepsilon}} + \frac{r^{\frac{2}{1-\alpha}+ 3(1-\alpha)}}{\varepsilon^{\frac{3}{\alpha} +\frac{2}{1-\alpha}  }}\log^{2}\rbra*{\frac{r}{\varepsilon}}}
    \]
    samples of $\rho$ and $\sigma$. 
\end{theorem}

\begin{proof}
   We notice that
    $\dTsa{\alpha}{\rho}{\sigma}{} = \frac{1}{1-\alpha} \rbra{1-\dAa{\alpha}{\rho}{\sigma}{}}$.
    Therefore, 
    to obtain an estimate of $\dTsa{\alpha}{\rho}{\sigma}{}$ within additive error
    $\varepsilon$, 
    it suffices to estimate $\dAa{\alpha}{\rho}{\sigma}{}$ to within
    $\rbra{1-\alpha}\eps$ error. 
    We now show how to 
 estimate $\dAa{\alpha}{\rho}{\sigma}{}$ to within
    $\eps$ error.

    For any $\alpha\in \interval[open]{0}{1}$, setting 
    \[
        \eps_1 =\delta_1 =\frac{\eps^{\frac{1}{\alpha}}}{40^{\frac{1}{\alpha}}r}, \quad
        \eps_2 = \frac{\eps}{32r}, \quad
        \delta = \eps_H = \frac{\eps^{\frac{1}{\alpha}}}{256 \cdot 40^{\frac{1}{\alpha}}r^{1-\alpha}}, \quad
        \delta_1' = \delta_2' = \frac{\eps^{\frac{1}{\alpha}}}{128 \cdot 40^{\frac{1}{\alpha}}r^{1-\alpha}}    
    \]
    in \cref{prop:sample-error-total-rank-sigma}, 
    we have
    \[
     \abs*{\frac{16}{\delta_1^{1-\alpha}}\rbra{1-2X} - \dAa{\alpha}{\rho}{\sigma}{}{}} \le \varepsilon.
    \]

    Now we consider the sample complexity of the algorithm. 
    By our choice of parameters, we have
    \[
        d_1 = O\rbra*{\frac{1}{\delta_1}\log\rbra*{\frac{1}{\varepsilon_1}}} = O\rbra*{\frac{r}{\varepsilon^{1/\alpha}}\log \rbra*{\frac{r}{\varepsilon}}},
    \quad 
        d_2 = O\rbra*{\frac{1}{\eps_2^{1/(1-\alpha)}}} =  O\rbra*{\frac{r^{1/(1-\alpha)}}{\varepsilon^{1/(1-\alpha)}}}.
    \]

    Then, the samplizer uses
    \[
        O\rbra*{\frac{\rbra{d_1+d_2}^2}{\delta} \log^2\rbra*{\frac{d_1+d_2}{\delta}}} = O\rbra*{\frac{r^{3-\alpha}}{\varepsilon^{3/\alpha}}\log^{4}\rbra*{\frac{r}{\varepsilon}} + \frac{r^{2/(1-\alpha)+ (1-\alpha)}}{\varepsilon^{2/(1-\alpha)+ 1/\alpha}}\log^{2}\rbra*{\frac{r}{\varepsilon}} }
    \]
    samples. Since we need to repeat $O(1/\varepsilon_H^2)$ times, the total sample complexity is 
    \[
        O\rbra*{\frac{r^{5-3\alpha}}{\varepsilon^{5/\alpha}}\log^{4}\rbra*{\frac{r}{\varepsilon}} + \frac{r^{2/(1-\alpha)+ 3(1-\alpha)}}{\varepsilon^{2/(1-\alpha)+ 3/\alpha}}\log^{2}\rbra*{\frac{r}{\varepsilon}}}.
    \]
\end{proof}

\subsubsection{Given the smaller rank of \texorpdfstring{$\rho$}{} and \texorpdfstring{$\sigma$}{}}

In the following, we consider the case when the minimum of rank of $\rho$ and rank of $\sigma$ has an upper bound $r$.
Note that~\Cref{prop:sample-error-single} still holds in this case.

\begin{proposition}\label{prop:error-operator-norm-sample}
      Let $\alpha \in \interval[open]{0}{1}$ be a constant.
    For any density operators $\rho, \sigma \in \calD\rbra{\calH}$, positive real numbers $\eps_1, \delta_1, \eps_2 \in \interval[open]{0}{1}$, we have
\begin{equation*}
 \Abs*{\rho p_1\rbra*{\frac{\rho}{2}} p_2\rbra*{\frac{\sigma}{2}}
    - \frac{\delta_1^{1-\alpha}}{4}\rho^{\alpha}\sigma^{1-\alpha} }\le 4\delta_1+ \eps_1 + \eps_2,
\end{equation*}
    where $p_1\coloneqq p_{1-\alpha, \varepsilon_1, \delta_1, -}$ is the polynomial specified in \cref{lem:poly_approx_neg}, and $p_2\coloneqq p_{0, 1-\alpha, \varepsilon_2 , +}$ is the polynomial specified in \cref{lem:poly_approx_pos}.
\end{proposition}

\begin{proof}
    By~\Cref{prop:sample-error-single}, we have
    \[
      \Abs*{\rho p_1\rbra*{\frac{\rho}{2}} - \delta_1^{1-\alpha}\rbra*{\frac{\rho}{2}}^{\alpha} } \le 4\delta_1+ \varepsilon_1,
    \]
     By our choice of $p_2$, we have
    \[
       \Abs*{p_2\rbra*{\frac{\sigma}{2}}-\frac{1}{2}\rbra*{\frac{\sigma}{2}}^{1-\alpha}} \le \varepsilon_2.
    \]
   Note that 
    $ \Abs*{p_2\rbra*{\frac{\sigma}{2}}} \leq \varepsilon_2 + \frac{1}{2^{2-\alpha}} \le 1$.
   Therefore, we have
   \[
   \begin{aligned}
         &\Abs*{\rho p_1\rbra*{\frac{\rho}{2}} p_2\rbra*{\frac{\sigma}{2}} - \frac{\delta_1^{1-\alpha}}{4}\rho^{\alpha}\sigma^{1-\alpha} } \\
         & \le 
 \Abs*{\rho p_1\rbra*{\frac{\rho}{2}} p_2\rbra*{\frac{\sigma}{2}}  - \frac{1}{2}\rho p_1\rbra*{\frac{\rho}{2}}\rbra*{ \frac{\sigma}{2}}^{1-\alpha} } + \Abs*{\frac{1}{2}\rho p_1\rbra*{\frac{\rho}{2}}\rbra*{ \frac{\sigma}{2}}^{1-\alpha}  - \frac{\delta_1^{1-\alpha}}{4}\rho^{\alpha}\sigma^{1-\alpha} }\\
 &\le 
 \Abs*{\rho p_1\rbra*{\frac{\rho}{2}}}\Abs*{p_2\rbra*{\frac{\sigma}{2}}-\frac{1}{2}\rbra*{\frac{\sigma}{2}}^{1-\alpha} } + \Abs*{\frac{1}{2^{2-\alpha}}\sigma^{1-\alpha}}\Abs*{\rho p_1\rbra*{\frac{\rho}{2}} - \delta_1^{1-\alpha}\rbra*{\frac{\rho}{2}}^{\alpha}} \\
 &\le 4\delta_1+ \eps_1 + \eps_2.
   \end{aligned}
   \]
\end{proof}

\begin{proposition}\label{prop:sample-error-total-min-rank}
    Let $X$,  $\varepsilon_H$, $\delta$, $\delta_1$,
    $\delta_1'$, $\delta_2'$ be the parameters in \cref{algo:affinity-sample}.
    If $\abs{X - \Prob\sbra{\widetilde{x} = 1}} \leq \varepsilon_H$, then
    \begin{align*}
        \abs*{\frac{16}{\delta_1^{1-\alpha}}\rbra{1-2X} - \dAa{\alpha}{\rho}{\sigma}{}{}} 
        \le {} & \frac{4}{\delta_1^{1-\alpha}} \rbra*{8\varepsilon_H + 8\delta + 4\delta_1' + 4\delta_2'+4r\delta_1+ r\eps_1 + r\eps_2},
    \end{align*}
     where $r\ge \min\{\rank(\rho), \rank(\sigma)\}$.
\end{proposition}

\begin{proof}
    In the same reasoning as in~\Cref{prop:sample-error-total}, we have
    \[
        \abs*{ 4\rbra{1-2X} - \tr\rbra*{\rho p_1\rbra*{\frac{\rho}{2}} p_2\rbra*{\frac{\sigma}{2}}} } \leq 8\rbra*{\varepsilon_H + \delta} + 4\rbra*{\delta_1' + \delta_2'}.
    \]

    By \cref{lem:tr-dif-bound,prop:error-operator-norm-sample}, we have
    \[
    \abs*{\tr\rbra*{\rho p_1\rbra*{\frac{\rho}{2}} p_2\rbra*{\frac{\sigma}{2}}} 
    - \tr \rbra*{\frac{\delta_1^{1-\alpha}}{4}\rho^{\alpha}\sigma^{1-\alpha}} }
    \le 4r\delta_1+ r\eps_1 + r\eps_2.
    \]
    Therefore, the result follows from the triangle inequality.
\end{proof}

\begin{theorem}[Sample upper bound for estimating quantum affinity with minimum rank bounded]\label{thm:sample-complexity-of-general-affinity-min-rank}
   Let $\alpha \in \interval[open]{0}{1}$ be a constant.
    There is a quantum algorithm that, for any $\varepsilon \in \interval[open]{0}{1}$, given sample access to quantum states $\rho, \sigma \in \mathcal{D}(\mathcal{H})$ with $r\ge \min\{\rank(\rho), \rank(\sigma)\}$, with probability at least $2/3$, estimating $\dTsa{\alpha}{\rho}{\sigma}{}$ to within additive error $\varepsilon$, using 
    \[
   \begin{cases}
               O\rbra*{\frac{r^{\frac{2}{\alpha}+\frac{5}{1-\alpha}-3}}{\varepsilon^{\frac{2}{\alpha}+\frac{5}{1-\alpha}}}\log^{2}\rbra*{\frac{r}{\varepsilon}} }, \text{ if } 0< \alpha \le \frac{1}{2}, \\
                O\rbra*{\frac{r^{\frac{5}{\alpha}+\frac{2}{1-\alpha}-3}}{\varepsilon^{\frac{5}{\alpha}+\frac{2}{1-\alpha}}}\log^{2}\rbra*{\frac{r}{\varepsilon}} },  \text{ if } \frac{1}{2}< \alpha < 1
         \end{cases}
    \]
    samples of $\rho$ and $\sigma$. 
\end{theorem}
\begin{proof}
   We notice that
    $\dTsa{\alpha}{\rho}{\sigma}{} = \frac{1}{1-\alpha} \rbra{1-\dAa{\alpha}{\rho}{\sigma}{}}$.
    Therefore, 
    to obtain an estimate of $\dTsa{\alpha}{\rho}{\sigma}{}$ within additive error
    $\varepsilon$, 
    it suffices to estimate $\dAa{\alpha}{\rho}{\sigma}{}$ to within
    $\rbra{1-\alpha}\eps$ error. 
    We now show how to 
 estimate $\dAa{\alpha}{\rho}{\sigma}{}$ to within
    $\eps$ error.

      For any $\alpha\in \interval[open]{0}{1}$, setting 
  setting 
    \[
        \varepsilon_1 = \varepsilon_2= \delta_1 = \frac{\eps^{1/\alpha}}{200^{1/\alpha}r^{1/\alpha}}, \quad \delta =\eps_H = \frac{\eps^{1/\alpha}}{8 \cdot 200^{(1-\alpha)/\alpha}r^{(1-\alpha)/\alpha}},
        \quad \delta_1' = \delta_2' = \frac{\eps^{1/\alpha}}{16 \cdot 200^{(1-\alpha)/\alpha}r^{(1-\alpha)/\alpha}}
    \]
    in \cref{prop:sample-error-total-min-rank}, 
    we have
    \[
     \abs*{\frac{16}{\delta_1^{1-\alpha}}\rbra{1-2X} - \dAa{\alpha}{\rho}{\sigma}{}{}} \le \varepsilon.
    \]

    Now we consider the sample complexity of the algorithm. 
    By our choice of parameters, we have
    \[
        d_1 = O\rbra*{\frac{1}{\delta_1}\log\rbra*{\frac{1}{\varepsilon_1}}} = O\rbra*{\frac{r^{1/\alpha}}{\varepsilon^{1/\alpha}}\log \rbra*{\frac{r}{\varepsilon}}},
    \quad 
        d_2 = O\rbra*{\frac{1}{\eps_2^{1/(1-\alpha)}}} =  O\rbra*{\frac{r^{1/\alpha(1-\alpha)}}{\varepsilon^{1/\alpha(1-\alpha)}}}.
    \]
    Note that $d_1= O\rbra{d_2}$.
    Then, the samplizer uses
    \[
        O\rbra*{\frac{\rbra{d_1+d_2}^2}{\delta} \log^2\rbra*{\frac{d_1+d_2}{\delta}}} = O\rbra*{\frac{r^{2/\alpha(1-\alpha) + (1-\alpha)/\alpha}}{\varepsilon^{2/\alpha(1-\alpha)+1/\alpha}}\log^{2}\rbra*{\frac{r}{\varepsilon}} } =  O\rbra*{\frac{r^{\frac{3}{\alpha}+\frac{2}{1-\alpha}-1}}{\varepsilon^{\frac{3}{\alpha}+\frac{2}{1-\alpha}}}\log^{2}\rbra*{\frac{r}{\varepsilon}} }
    \]
    samples. Since we need to repeat $O(1/\varepsilon_H^2)$ times,
    the total sample complexity is 
    \[
       O\rbra*{\frac{r^{\frac{5}{\alpha}+\frac{2}{1-\alpha}-3}}{\varepsilon^{\frac{5}{\alpha}+\frac{2}{1-\alpha}}}\log^{2}\rbra*{\frac{r}{\varepsilon}} }.
    \]
    Note that $\dAa{\alpha}{\rho}{\sigma}{} = \dAa{1-\alpha}{\sigma}{\rho}{}$, this yields an algorithm with sample complexity
    \[
       O\rbra*{\frac{r^{\frac{2}{\alpha}+\frac{5}{1-\alpha}-3}}{\varepsilon^{\frac{2}{\alpha}+\frac{5}{1-\alpha}}}\log^{2}\rbra*{\frac{r}{\varepsilon}} }.
    \]
    which is no larger if $0< \alpha \le \frac{1}{2}$.
\end{proof}
   
\subsection{Application: Tolerant quantum state certification in Hellinger distance} \label{sec:app}

As an application, our algorithm can be used to estimate the Hellinger distance between quantum states, and thus is useful in the tolerant quantum state certification with respect to the Hellinger distance.

\begin{corollary}[Tolerant quantum state certification in Hellinger distance]\label{thm:tolerant-state-certification-Hellinger-distance}
    For
    any real numbers $0 \le  \varepsilon_1 < \varepsilon_2$, given two unknown quantum states $\rho$ and $\sigma$:
    \begin{itemize}
        \item If $r = \max\cbra{\rank\rbra{\rho}, \rank\rbra{\sigma}}$ is known, then we can distinguishes the case $\dH{\rho}{\sigma}{}{} \le \varepsilon_1$ from the case $\dH{\rho}{\sigma}{}{} \ge \varepsilon_2$ using
            \begin{itemize}
               \item  
               $O\rbra*{\frac{r^{3.5}}{\rbra*{\varepsilon_2-\varepsilon_1}^{10}}\log^{4}\rbra*{\frac{r}{{\varepsilon_2-\varepsilon_1}}}}$ samples of $\rho$ and $\sigma$ 
                (by \cref{thm:Tsallis-relative-sample}).

                 \item  
                 $O\rbra*{\frac{r^{1.5}}{\rbra*{\varepsilon_2-\varepsilon_1}^{4}}\log\rbra*{\frac{r}{{\varepsilon_2-\varepsilon_1}}}}$ queries to the state-preparation circuits of $\rho$ and $\sigma$ 
                 (by~\cref{thm:Tsallis-relative-query}).

            \end{itemize} 
      
        \item If $r = \rank\rbra{\sigma}$ or $r=\rank\rbra{\sigma}$ is known, then we can distinguishes the case $\dH{\rho}{\sigma}{}{} \le \varepsilon_1$ from the case $\dH{\rho}{\sigma}{}{} \ge \varepsilon_2$ using
            \begin{itemize}
                \item $O\rbra*{ \frac{r^{5.5}}{\rbra*{\varepsilon_2-\varepsilon_1}^{10 }}\log^{2}\rbra*{\frac{r}{{\varepsilon_2-\varepsilon_1}}}}$ samples of $\rho$ and $\sigma$ (by \cref{thm:sample-complexity-of-general-affinity-rank-sigma});
                \item $O\rbra*{ \frac{r^{2.5}}{\rbra*{\varepsilon_2-\varepsilon_1}^{4}}}$ queries to the state-preparation circuits of $\rho$ and $\sigma$ (by \cref{thm:query-complexity-of-general-affinity-rank-sigma}).
            \end{itemize} 

        \item If $r = \min\cbra{\rank\rbra{\rho}, \rank\rbra{\sigma}}$ is known, then we can distinguishes the case $\dH{\rho}{\sigma}{}{} \le \varepsilon_1$ from the case $\dH{\rho}{\sigma}{}{} \ge \varepsilon_2$ using 
            \begin{itemize}
                  \item $O\rbra*{\frac{r^{11}}{\rbra*{\varepsilon_2-\varepsilon_1}^{14}}\log^{2}\rbra*{\frac{r}{{\varepsilon_2-\varepsilon_1}}} }$ samples of $\rho$ and $\sigma$  (by \cref{thm:sample-complexity-of-general-affinity-min-rank});
                   \item $O\rbra*{\frac{r^{5}}{\rbra*{\varepsilon_2-\varepsilon_1}^{6}}}$  queries to the state-preparation circuits of $\rho$ and $\sigma$ (by \cref{thm:query-complexity-of-general-affinity-min-rank}).
            \end{itemize} 
    \end{itemize}
\end{corollary}

\section{Lower Bounds} \label{sec:lb}

In this section, we investigate the query and sample complexity lower bounds for estimating the quantum Tsallis relative entropy.
The lower bounds obtained in this section are summarized in the following theorem. 

\begin{theorem}[\Cref{thm:lower-query-affinity,thm:sample-lower-bound-affinity,thm:sample-lower-bound-Hellinger-distance} combined]
Let $\alpha \in \rbra{0, 1}$ be a constant. Given two unknown quantum states $\rho$ and $\sigma$ of rank at most $r$, for any sufficiently small $\varepsilon > 0$, 
    \begin{itemize}
        \item Estimating $\dTsa{\alpha}{\rho}{\sigma}{}$ or $\dH{\rho}{\sigma}{}$ to with additive error $\varepsilon$ requires query complexity $\Omega\rbra{r^{1/3}+1/\varepsilon}$. 
        \item Estimating $\dTsa{\alpha}{\rho}{\sigma}{}$ to with additive error $\varepsilon$ requires sample complexity $\Omega\rbra{r/\varepsilon+1/\varepsilon^2}$. 
        \item Estimating $\dH{\rho}{\sigma}{}$ to with additive error $\varepsilon$ requires sample complexity $\Omega\rbra{r/\varepsilon^2}$. 
    \end{itemize}
\end{theorem}

\subsection{Query complexity lower bound}\label{sec:query_lb}

To show a query complexity lower bound, we need the following result in \cite{CFMdW10}, which was recently used to show the quantum query lower bounds for estimating the Tsallis entropy \cite{LW25a} and fidelity \cite{UNWT25}. 

\begin{theorem}[Adapted from {\cite[Theorem 4.1 in the full version]{CFMdW10}}] \label{thm:uniform_test_lb}
    Let 
    \[
        \rho = \sum_{i=0}^{d-1} p_i \ketbra{i}{i}
    \]
    be a diagonal mixed quantum state with $p = \rbra{p_0, p_1, \dots, p_{d-1}}$ forming a discrete probability distribution. 
    Given purified quantum query access to $\rho$, for any $\varepsilon \in (0, 1/2]$, determining whether the distribution $p$ is uniform or $\varepsilon$-far from being uniform in the total variation distance requires query complexity $\Omega\rbra{d^{1/3}}$. 
\end{theorem}

We also need a lower bound for estimating the fidelity between two pure quantum states in the precision $\varepsilon$, which was shown in \cite{BBC+01,NW99}. 
Here, we use the version in \cite{Wan24}. 

\begin{theorem}[Adapted from {\cite[Theorems V.2 and V.3]{Wan24}}] \label{thm:lb-sqr-fid}
    Given purified quantum query access to two unknown pure quantum states $\ket{\varphi}$ and $\ket{\psi}$, for $\varepsilon \in \rbra{0, 1/2}$, any quantum query algorithm that estimates $\mathrm{F}^2\rbra{\ketbra{\varphi}{\varphi},\ketbra{\psi}{\psi}} = \abs{\braket{\varphi}{\psi}}^2$ or $\dtr{\ketbra{\varphi}{\varphi}}{\ketbra{\psi}{\psi}}{} = \sqrt{1 - \abs{\braket{\varphi}{\psi}}^2}$ to within additive error $\varepsilon$ requires query complexity $\Omega\rbra{1/\varepsilon}$. 
\end{theorem}

\begin{theorem}[Query lower bound for estimating quantum Tsallis relative entropy and quantum Hellinger distance] \label{thm:lower-query-affinity}
    Let $\alpha \in \rbra{0, 1}$ be a constant. 
    Given purified quantum query access to two unknown quantum states $\rho$ and $\sigma$ of rank $r$, 
    \begin{itemize}
        \item For $\varepsilon \in \rbra{0, \min\cbra{\frac{1-\alpha}{2}, \frac{\alpha}{4}}}$, any quantum query algorithm that estimates $\dTsa{\alpha}{\rho}{\sigma}{}$ to within additive error $\varepsilon$ requires query complexity $\Omega\rbra{r^{1/3} + 1/\varepsilon}$.
        \item For $\varepsilon \in \rbra{0, \frac{\sqrt{2}}{4}}$, any quantum query algorithm that estimates $\dH{\rho}{\sigma}{}$ to within additive error $\varepsilon$ requires query complexity $\Omega\rbra{r^{1/3} + 1/\varepsilon}$.
    \end{itemize} 
\end{theorem}

\begin{proof}
    Let
    \[
        \rho_{\mathsf{m}} = \sum_{i=0}^{r-1} \frac{1}{r} \ketbra{i}{i}, \qquad \rho = \sum_{i=0}^{r-1} p_i \ketbra{i}{i}.
    \]
    Let $\mu$ be the uniform distribution over $r$ elements. 

    By~\Cref{lem:Tsa_vs_tr}, noting that $ \dtv{p}{\mu} = \dtr{\rho}{\rho_\mathsf{m}}{}$, we have
      \[
        \begin{aligned}
            \dtv{p}{\mu} = 0 & \quad \Longrightarrow \quad \dTsa{\alpha}{\rho}{\rho_{\mathsf{m}}}{} = 0, \\
            \dtv{p}{\mu} \geq  \sqrt{\varepsilon/\alpha} & \quad \Longrightarrow \quad \dTsa{\alpha}{\rho}{\rho_{\mathsf{m}}}{} \ge 2\varepsilon. 
        \end{aligned}
    \]
    Therefore, any quantum algorithm that estimates $\dTsa{\alpha}{\rho}{\sigma}{}$ to within additive error $\varepsilon$ can be used to distinguish whether $p$ is uniform or $\sqrt{\varepsilon/\alpha}$-far from being uniform in the total variation distance. By \cref{thm:uniform_test_lb}, for $\varepsilon\in \interval[open]{0}{\alpha/4}$, 
    it requires query complexity $\Omega\rbra{r^{1/3}}$.
    Therefore, any quantum algorithm that estimates $\dTsa{\alpha}{\rho}{\sigma}{}$ to within additive error $\varepsilon$ requires query complexity $\Omega\rbra{r^{1/3}}$. 
    On the other hand, for any $\varepsilon\in \interval[open]{0}{(1-\alpha)/2}$, note that $\dAa{\alpha}{\rho}{\sigma}{} = \mathrm{F}^2\rbra{\rho, \sigma}$ when both $\rho$ and $\sigma$ are pure. 
    By \cref{thm:lb-sqr-fid}, for $\varepsilon \in \interval[open]{0}{1/2}$, any quantum query algorithm that estimates $\dAa{\alpha}{\rho}{\sigma}{}$ to within additive error $\varepsilon$ requires query complexity $\Omega\rbra{1/\varepsilon}$.
    Note that $\dTsa{\alpha}{\rho}{\sigma}{} =\frac{1}{1-\alpha}(1-\dAa{\alpha}{\rho}{\sigma}{})$.
    Combining both cases yields the proof.

    For the special case when $\alpha = 1/2$, by \cref{fact:dh_vs_dtr}, we have that
    \[
        \begin{aligned}
            \dtv{p}{\mu} = 0 & \quad \Longrightarrow \quad \dH{\rho}{\rho_\mathsf{m}}{} = 0, \\
            \dtv{p}{\mu} \geq  2\sqrt{2}\varepsilon & \quad \Longrightarrow \quad \dH{\rho}{\rho_\mathsf{m}}{} \ge 2\varepsilon. 
        \end{aligned}
    \]
    Therefore, any quantum algorithm that estimates $\dH{\rho}{\sigma}{}$ to within additive error $\varepsilon$ can be used to distinguish whether $p$ is uniform or $\sqrt{2}\varepsilon$-far from being uniform in the total variation distance. 
    By \cref{thm:uniform_test_lb}, for $\varepsilon\in \interval[open]{0}{\sqrt{2}/4}$, 
    it requires query complexity $\Omega\rbra{r^{1/3}}$.
    Therefore, any quantum algorithm that estimates $\dH{\rho}{\sigma}{}$ to within additive error $\varepsilon$ requires query complexity $\Omega\rbra{r^{1/3}}$. 
    On the other hand, when both $\rho$ and $\sigma$ are pure, $\dH{\rho}{\sigma}{} = \dtr{\rho}{\sigma}{}$. 
    By \cref{thm:lb-sqr-fid}, given purified quantum query access to two pure states $\rho$ and $\sigma$, for $\varepsilon \in \interval[open]{0}{1/2}$, estimating $\dtr{\rho}{\sigma}{} = \dH{\rho}{\sigma}{}$ to within additive error $\varepsilon$ requires query complexity $\Omega\rbra{1/\varepsilon}$.
    Combining both cases yields the proof.
\end{proof}

\subsection{Sample complexity lower bound}\label{sec:sample_lb}

We first recall a sample complexity lower bound for quantum state certification~\cite{OW21, BOW19}.
\begin{theorem}[{\cite[Corollary 4.3]{OW21}}]\label{thm:lower-bound-state-certification}
    Suppose $d$ is an even integer and $\eps$ is a positive real with $\eps \in \interval[open left]{0}{1/2}$.
    Let $\sigma = I/d$, and $\mathcal{C}_{\eps}$ denote the set of density operators with $d/2$ eigenvalues being $(1-2\eps)/d$ and $d/2$ eigenvalues being $(1+2\eps)/d$.
    Then, any measurement strategy that can distinguish the case $\rho = \sigma$ from the case $\rho\in \mathcal{C}_{\eps}$ with probability at least $1/3$ must use at least $0.15d/\eps^2$ samples.
\end{theorem}

We also need an $\Omega(1/\eps^2)$ lower bound for inner product estimation given in~\cite{ALL22}.

\begin{theorem}[{\cite[Lemma 13 in the full version]{ALL22}}]\label{thm:lower-bound-inner-product}
    Suppose $\varepsilon \in \interval{0}{1/2}$.
    Denote $\ket{\phi_0} = \sqrt{\frac{1}{2}-\varepsilon}\ket{0}
    +\sqrt{\frac{1}{2}+\varepsilon}\ket{1}$,
    and $\ket{\phi_1} = \sqrt{\frac{1}{2}+\varepsilon}\ket{0}
    +\sqrt{\frac{1}{2}-\varepsilon}\ket{1}$.
    Let $\rho$ be a
    density operator on a $d$-dimensional space,
    and $\sigma = \ket{0}\bra{0}$ be a
    density operator on a $d$-dimensional space.
    If there is an algorithm that,
    on input $\rho^{\otimes k}\otimes \sigma^{\otimes k}$,
    successfully distinguishes the case $\rho =\ket{\phi_0}\bra{\phi_0}$
    from $\rho = \ketbra{\phi_1}{\phi_1}$,
    with probability at least $2/3$,
    then $k = \Omega(1/\eps^2)$.
\end{theorem}

Given the above theorems, we can show the following sample complexity lower bounds for computing the quantum affinity and the quantum Hellinger distance.

\begin{theorem}[Sample lower bound for estimating quantum Tsallis relative entropy]\label{thm:sample-lower-bound-affinity}
    Let $\alpha \in \rbra{0, 1}$ be a constant. 
    Let $r$ be an integer and
    $\varepsilon\in \interval[open]{0}{\alpha(1-\alpha)/4}$ be a positive real number.
    Given a known quantum state $\sigma$
    of rank $r$ 
    and copies of an unknown quantum state $\rho$
    which is promised to have rank at most $r$, for any constant $\alpha \in \rbra{0, 1}$, estimating $\dTsa{\alpha}{\rho}{\sigma}{}$ to within additive error $\varepsilon$ 
    requires $\Omega(r/\varepsilon+ 1/\varepsilon^2)$ samples of $\rho$.
\end{theorem}
\begin{proof}
    Without loss of generality, we assume that $r$ is even.
    In the following, we show a reduction from the quantum state
    certification problem to our affinity estimation problem.
    Given any instance of the quantum state certification problem,
    with $\sigma = I/r$ and
    $\mathcal{C}_{\varepsilon'}$ being a set of density operators
    on an $r$-dimensional Hilbert space
    and $\varepsilon' \in \interval[open]{0}{1/2}$.
    We can regard $\rho$ and $\sigma$
    as density operators 
    on a $d$-dimensional Hilbert space
    for any $d\ge r$.
    If $\rho\in \mathcal{C}_{\varepsilon'}$,
    by direct computation,
    one have
    \[
    \dAa{\alpha}{\rho}{\sigma}{} = 
    \tr(\rho^{\alpha}\sigma^{1-\alpha}) = \frac{\rbra{1+2\varepsilon'}^{\alpha}+\rbra{1-2\varepsilon'}^{\alpha}}{2} \le 1 -2\alpha(1-\alpha)\varepsilon'^2,
    \]
    If $\rho = \sigma$, the affinity is $1$.

    Then, 
    using the quantum algorithm
    for estimating the affinity
    with precision $\varepsilon = \alpha(1-\alpha)\varepsilon'^2$, where
    $\varepsilon\in \interval[open]{0}{\alpha(1-\alpha)/4}$,
    we could distinguish either $\rho = \sigma$
    if the estimate value is more than $1- \alpha(1-\alpha)\varepsilon'^2$,
    and 
    $\rho\in \mathcal{C}_{\varepsilon'}$ otherwise.
    Therefore,
    by \cref{thm:lower-bound-state-certification},
    the number of samples 
    should be at least $\Omega(r/\varepsilon'^2) = 
    \Omega(r/\varepsilon)$.

    On the other hand, we also give a reduction from the inner product estimation problem to our affinity estimation problem.
    Given any instance of the inner product estimation problem with $\rho $ either being $\ketbra{\phi_0}{\phi_0}$ or $\ketbra{\phi_1}{\phi_1}$, and $\sigma = \ketbra{0}{0}$.
    By direct computation, we know if $\rho = \ketbra{\phi_0}{\phi_0}$, then
    \[
    \tr(\rho^{\alpha}\sigma^{1-\alpha}) = 
    \abs{\braket{\phi_0}{0}}^2 = \frac{1}{2}-\varepsilon,
    \]
    and similarly if $\rho = \ketbra{\phi_1}{\phi_1}$, then
    \[
    \tr(\rho^{\alpha}\sigma^{1-\alpha}) = 
    \abs{\braket{\phi_1}{0}}^2 = \frac{1}{2}+\varepsilon.
    \]
    Therefore, estimating $\dAa{\alpha}{\rho}{\sigma}{}$ within additive error $\varepsilon$ suffices to distinguish 
    the case $\rho = \ketbra{\phi_0}{\phi_0}$ from 
    $\rho = \ketbra{\phi_1}{\phi_1}$.
    By \Cref{thm:lower-bound-inner-product},
    we know that this must require $\Omega(1/\varepsilon^2)$ copies of $\rho$ and $\sigma$.
    Note that $\dTsa{\alpha}{\rho}{\sigma}{} =\frac{1}{1-\alpha}(1-\dAa{\alpha}{\rho}{\sigma}{})$, we require $\varepsilon\in \interval[open]{0}{(1-\alpha)/2}$.
    Combining both cases yields the proof.
\end{proof}

\begin{theorem}[Sample lower bound for estimating quantum Hellinger distance]\label{thm:sample-lower-bound-Hellinger-distance}
    Let $r$ be an integer and
    $\varepsilon\in \interval[open]{0}{\sqrt{2}/12}$ be a positive real number.
    Given a known quantum state $\sigma$
    of rank $r$ 
    and copies of an unknown quantum state $\rho$
    which is promised to have rank at most $r$, estimating $\dH{\rho}{\sigma}{}$ to within additive error $\varepsilon$ 
    requires $\Omega(r/\varepsilon^2)$ samples of $\rho$.
\end{theorem}

\begin{proof}
    We give a reduction from the quantum state
    certification problem to our Hellinger 
    distance estimation problem.
    Given any instance of the quantum 
    state certification problem,
    with $\sigma = I/r$ and
    $\mathcal{C}_{\varepsilon'}$ being a set of density operators
    on an $r$-dimensional Hilbert space
    and $\varepsilon' \in \interval[open]{0}{1/2}$.
    Without loss of generality,
    we can regard $\sigma$ and $\rho$
    as density operators 
    on $d$-dimensional Hilbert space
    for $d\ge r$.
    If $\rho\in \mathcal{C}_{\varepsilon'}$,
    by direct computation,
    one have
    \[
        \tr(\sqrt{\rho}\sqrt{\sigma}) = \frac{\sqrt{1+2\varepsilon'}+\sqrt{1-2\varepsilon'}}{2} \le 1 -\frac{\varepsilon'^2}{2},
    \]
    meaning that
    \[
    \dH{\rho}{\sigma}{}{}
    = \sqrt{1-\tr(\sqrt{\rho}\sqrt{\sigma})}
    \ge \frac{\varepsilon'}{\sqrt{2}}.
    \]
    
  Then, by applying the quantum algorithm for estimating the Hellinger distance with precision $\varepsilon = \sqrt{2}\varepsilon'/6$, we can distinguish between the cases $\rho = \sigma$ and $\rho \in \mathcal{C}_{\varepsilon'}$, depending on whether the estimate is below $\varepsilon$ or not.
        Therefore,
    by \cref{thm:lower-bound-state-certification}
    the number of samples 
    should be at least $\Omega(r/\varepsilon'^2) = 
    \Omega(r/\varepsilon^2)$.
\end{proof}

\section{Computational Hardness}

In this section, we show the $\QSZK$-completeness of estimating the quantum Tsallis relative entropy and quantum Hellinger distance between general quantum states in~\cref{sec:full-rank_complexity}, and the $\BQP$-completeness of estimating the Quantum Tsallis relative entropy and quantum Hellinger distance between low-rank quantum states in~\cref{sec:low-rank_complexity}.

We first introduce a generalization of the $\QSD$ problem from~\cite{Wat02}, where the trace distance is replaced by quantum $\alpha$-Tsallis relative entropy.

\begin{definition}[Quantum state distinguishability problem with respect to the quantum $\alpha$-Tsallis relative entropy, $\TsaQSD_{\alpha}$]
    Let $\alpha \in \rbra{0, 1}$ be a constant. 
    Let $Q_\rho$ and $Q_\sigma$ be two quantum circuits with $m\rbra{n}$-qubit input and $n$-qubit output, where $m\rbra{n}$ is a polynomial in $n$. Let $\rho$ and $\sigma$ be $n$-qubit quantum states obtained by performing $Q_\rho$ and $Q_\sigma$ on input state $\ket{0}^{\otimes m\rbra{n}}$. 
     Let $a\rbra{n}$ and $b\rbra{n}$ be efficiently computable functions such that $0 \leq b\rbra{n} < a\rbra{n} \leq 1$. The problem $\TsaQSD_{\alpha}\sbra{a, b}$ is to decide whether:
    \begin{itemize}
        \item \textup{(Yes)} $\dTsa{\alpha}{\rho}{\sigma}{} \geq a\rbra{n}$, or
        \item \textup{(No)} $\dTsa{\alpha}{\rho}{\sigma}{} \leq b\rbra{n}$.
    \end{itemize}   
\end{definition}

Since the quantum Hellinger distance is a special case of the quantum $\alpha$-Tsallis relative entropy when $\alpha$ is set to $1/2$, we also define the quantum state distinguishability problem in terms of the quantum Hellinger distance.

\begin{definition}[Quantum state distinguishability problem with respect to the quantum Hellinger distance, $\HellQSD$]
     Let $Q_\rho$ and $Q_\sigma$ be two quantum circuits with $m\rbra{n}$-qubit input and $n$-qubit output, where $m\rbra{n}$ is a polynomial in $n$. Let $\rho$ and $\sigma$ be $n$-qubit quantum states obtained by performing $Q_\rho$ and $Q_\sigma$ on input state $\ket{0}^{\otimes m\rbra{n}}$.
     Let $a\rbra{n}$ and $b\rbra{n}$ be efficiently computable functions such that $0 \leq b\rbra{n} < a\rbra{n} \leq 1$. The problem $\HellQSD\sbra{a, b}$ is to decide whether:
    \begin{itemize}
        \item \textup{(Yes)} $\dH{\rho}{\sigma}{} \geq a\rbra{n}$, or
        \item \textup{(No)} $\dH{\rho}{\sigma}{} \leq b\rbra{n}$.
    \end{itemize}   
\end{definition}

When restricted to low-rank quantum states, we also define the quantum state distinguishability problem for them.

\begin{definition}[Low-rank quantum state distinguishability problem with respect to the quantum $\alpha$-Tsallis relative entropy and the quantum Hellinger distance, $\TsaLowQSD_{\alpha}$ and $\HellLowQSD$]
    Let $\alpha \in \rbra{0, 1}$ be a constant. 
    Let $Q_\rho$ and $Q_\sigma$ be two quantum circuits with $m\rbra{n}$-qubit input and $n$-qubit output, where $m\rbra{n}$ is a polynomial in $n$. Let $\rho$ and $\sigma$ be $n$-qubit quantum states of rank at most $r\rbra{n}$, obtained by performing $Q_\rho$ and $Q_\sigma$ on input state $\ket{0}^{\otimes m\rbra{n}}$, where $r\rbra{n}$ is a polynomial in $n$.
    Let $a\rbra{n}$ and $b\rbra{n}$ be efficiently computable functions such that $0 \leq b\rbra{n} < a\rbra{n} \leq 1$. 
     
    \begin{enumerate}
    \item The problem $\TsaLowQSD_{\alpha}\sbra{a, b}$ is to decide whether:
    \begin{itemize}
    \item \textup{(Yes)} $\dTsa{\alpha}{\rho}{\sigma}{} \geq a\rbra{n}$, or
    \item \textup{(No)} $\dTsa{\alpha}{\rho}{\sigma}{} \leq b\rbra{n}$.
    \end{itemize} 
    \item The problem $\HellLowQSD\sbra{a, b}$ is to decide whether:
    \begin{itemize}
    \item \textup{(Yes)} $\dH{\rho}{\sigma}{} \geq a\rbra{n}$, or
    \item \textup{(No)} $\dH{\rho}{\sigma}{} \leq b\rbra{n}$.
    \end{itemize}   
    \end{enumerate}
\end{definition}

Our theorem is stated as follows.

\begin{theorem}[\Cref{thm:Tsa_QSZK_hard,thm:Tsa_QSZK_containment,thm:Tsa_pure_BQP_hard,lemma:bqp-contain-TsaLowQSD,} combined] \label{thm:hardness-combined}
Let $\alpha \in \rbra{0, 1}$ be a constant. 
Let $a\rbra{n}$ and $b\rbra{n}$ be efficiently computable functions such that $0 \leq b\rbra{n} < a\rbra{n} \leq 1$. 
\begin{enumerate}
    \item For every constant $\tau \in \interval[open]{0}{1/2}$, $\TsaQSD_{\alpha}\sbra{a, b}$ is $\QSZK$-complete if $\rbra{1-\alpha}^2a\rbra{n}^2-\sqrt{\frac{b\rbra{n}}{2\alpha}} \geq 1/O\rbra{\log n}$, and $a\rbra{n} \leq 2\alpha\rbra{1-2^{-n^\tau}}$ and $b\rbra{n} \geq \frac{2^{-n^\tau}}{1-\alpha}$ for sufficiently large $n$. \label{thm-hard-item:1}
    \item For every constant $\tau \in \interval[open]{0}{1/2}$, $\HellQSD\sbra{a, b}$ is $\QSZK$-complete if $a\rbra{n}^4-\sqrt{2}b(n) \geq 1/O\rbra{\log n}$, and $a\rbra{n} \leq \sqrt{\frac{1-2^{-n^\tau}}{2}}$ and $b\rbra{n} \geq  2^{-\frac{n^\tau}{2}}$ for sufficiently large $n$. \label{thm-hard-item:2}
    \item $\TsaLowQSD_{\alpha}\sbra{a, b}$ is $\BQP$-complete if $a\rbra{n} - b\rbra{n} \geq \frac{1}{\poly\rbra{n}}$, and $a(n) \leq \frac{\rbra{1-2^{-n-1}}^{2}}{1-\alpha}$ and $b(n) \geq \frac{2^{-2n-2}}{1-\alpha}$  for sufficiently large $n$. \label{thm-hard-item:3}
    \item $\HellLowQSD\sbra{a, b}$ is $\BQP$-complete if $a\rbra{n} - b\rbra{n} \geq \frac{1}{\poly\rbra{n}}$, and $a(n) \leq 1-2^{-n-1}$ and $b(n) \geq 2^{-n-1}$ for sufficiently large $n$. \label{thm-hard-item:4}
\end{enumerate}
\end{theorem}
\begin{proof}
    \cref{thm-hard-item:1} combines \cref{thm:Tsa_QSZK_hard,thm:Tsa_QSZK_containment}. 
    \cref{thm-hard-item:3} combines \cref{thm:Tsa_pure_BQP_hard,lemma:bqp-contain-TsaLowQSD}.
    \cref{thm-hard-item:2,thm-hard-item:4} are respectively the special cases of \cref{thm-hard-item:1,thm-hard-item:3} when $\alpha = 1/2$ using the fact that $2\dH{\rho}{\sigma}{2}{} = \dTsa{1/2}{\rho}{\sigma}{}$. 
\end{proof}

\subsection{Estimating of quantum Tsallis relative entropy in general}\label{sec:full-rank_complexity}

Now we prove the $\QSZK$-hardness and $\QSZK$-containment of $\TsaQSD_\alpha$.
To prove the $\QSZK$-hardness, we reduce from the $\QSD$ problem by \cref{lem:qszk-hard}.

\begin{lemma}[$\QSZK$-hardness of $\TsaQSD_\alpha$]\label{thm:Tsa_QSZK_hard}
    Let $\alpha \in \interval[open]{0}{1}$ be a constant. 
    Let $a\rbra{n}$ and $b\rbra{n}$ be efficiently computable functions such that $0 \leq b\rbra{n} < a\rbra{n} \leq 1$. For every constant $\tau \in \interval[open]{0}{1/2}$, $\TsaQSD_\alpha\sbra{a,b}$ is $\QSZK$-hard if $a\rbra{n} \leq 2\alpha\rbra{1-2^{-n^\tau}}$ and $b\rbra{n} \geq  \frac{2^{-n^\tau}}{1-\alpha}$ for sufficiently large $n$.
\end{lemma}

\begin{proof}
    By~\cref{lem:qszk-hard}, as $\QSD\sbra{1-2^{-n^\tau},2^{-n^\tau}}$ is $\QSZK$-hard for any $\tau \in \interval[open]{0}{1/2}$ and any $n \in \N$, we reduce $\QSD\sbra{1-2^{-n^\tau},2^{-n^\tau}}$ to $\TsaQSD_\alpha$. We have the following implications.
\[
\begin{aligned}
    \dtr{\rho}{\sigma}{} &\geq 1-2^{-n^\tau} &\implies \dTsa{\alpha}{\rho}{\sigma}{} &\geq 2\alpha \rbra{1-2^{-n^\tau}} &\eqqcolon a'\rbra{n},\\
    \dtr{\rho}{\sigma}{} &\leq 2^{-n^\tau} &\implies \dTsa{\alpha}{\rho}{\sigma}{} &\leq  \frac{2^{-n^\tau}}{1-\alpha} &\eqqcolon b'\rbra{n}.
\end{aligned}
\]
The gap between 
\[
a'\rbra{n}-b'\rbra{n}=2\alpha \rbra{1-2^{-n^\tau}} - \frac{2^{-n^\tau}}{1-\alpha} = \frac{\rbra{2\alpha-2\alpha^2}-\rbra{2\alpha-2\alpha^2+1}2^{-n^\tau}}{1-\alpha} \eqqcolon g\rbra{n}
\]
Obviously $g\rbra{n}$ is an increasing function. To obtain $g\rbra{n}>0$, it suffices to choose
\[
    n \geq \ceil*{\rbra*{\log\rbra*{\frac{2\alpha-2\alpha^2+1}{2\alpha-2\alpha^2}}}^{1/\tau}}.
\]
Therefore, we have $a'\rbra{n} \geq b'\rbra{n}$ for sufficiently large $n$.
\end{proof}

Now we show the regime where $\TsaQSD_\alpha$ is contained in $\QSZK$.

\begin{lemma}[$\QSZK$-containment of $\TsaQSD_\alpha$]\label{thm:Tsa_QSZK_containment}
    Let $\alpha \in \rbra{0, 1}$ be a constant. 
    Let $a\rbra{n}$ and $b\rbra{n}$ be efficiently computable functions such that $0 \leq b\rbra{n} < a\rbra{n} \leq 1$. If $\rbra{1-\alpha}^2a\rbra{n}^2-\sqrt{\frac{b\rbra{n}}{2\alpha}} \geq 1/O\rbra{\log n}$, then $\TsaQSD_\alpha\sbra{a,b}$ is in $\QSZK$.
\end{lemma}

\begin{proof}
To prove $\QSZK$-containment, we reduce $\TsaQSD_\alpha$ to $\QSD$ by \cref{lem:qszk-hard}. Specifically, by \cref{lem:Tsa_vs_tr}, we have the following implications.
\[
\begin{aligned}
    \dTsa{\alpha}{\rho}{\sigma}{} &\geq a\rbra{n} &\implies \dtr{\rho}{\sigma}{} &\geq \rbra{1-\alpha}a\rbra{n},\\    
    \dTsa{\alpha}{\rho}{\sigma}{} &\leq b\rbra{n} &\implies \dtr{\rho}{\sigma}{} &\leq \sqrt{\frac{b\rbra{n}}{2\alpha}}.
\end{aligned}
\]
Thus, $\TsaQSD_\alpha\sbra{a,b}$ can be reduced to $\QSD\sbra{\rbra{1-\alpha}a, \sqrt{\frac{b}{2\alpha}}}$. 
To make $\QSD\sbra{\rbra{1-\alpha}a, \sqrt{\frac{b}{2\alpha}}}$ in $\QSZK$, 
it is sufficient to have $\rbra{1-\alpha}^2a\rbra{n}^2-\sqrt{\frac{b\rbra{n}}{2\alpha}} \geq \frac{1}{O\rbra{\log n}}$. 
Therefore, $\TsaQSD_\alpha\sbra{a,b}$ is in $\QSZK$.
\end{proof}

\subsection{Low-rank estimating of quantum Tsallis relative entropy}\label{sec:low-rank_complexity}

To prove the $\BQP$-hardness of $\TsaLowQSD_\alpha$ and $\HellLowQSD$, we introduce the following quantum state distinguishability problems restricted to pure states, which are the generalizations of $\PureQSD$, where the trace distance is replaced by either quantum Tsallis relative entropy or quantum Hellinger distance.

\begin{definition}[$\TsaPureQSD_{\alpha}$ and $\HellPureQSD$]
    Let $\alpha \in \rbra{0, 1}$ be a constant. 
    Let $Q_\phi$ and $Q_\psi$ be two quantum circuits with $m\rbra{n}$-qubit input and $n$-qubit ouput, where $m\rbra{n}$ is a polynomial in $n$. 
    Let $\ket{\phi}$ and $\ket{\psi}$ be $n$-qubit quantum states obtained by performing $Q_\phi$ and $Q_\psi$ on input state $\ket{0}^{\otimes m\rbra{n}}$.
    Let $a\rbra{n}$ and $b\rbra{n}$ be efficiently computable functions such that $0 \leq b\rbra{n} < a\rbra{n} \leq 1$. 
    \begin{enumerate}
        \item The problem $\TsaPureQSD_{\alpha}\sbra{a,b}$ is to decide whether:
    \begin{itemize}
        \item \textup{(Yes)} $\dTsa{\alpha}{\ketbra{\phi}{\phi}}{\ketbra{\psi}{\psi}}{} \geq a\rbra{n}$;
        \item \textup{(No)} $\dTsa{\alpha}{\ketbra{\phi}{\phi}}{\ketbra{\psi}{\psi}}{} \leq b\rbra{n}$.
    \end{itemize}     
        \item The problem $\HellPureQSD\sbra{a,b}$ is to decide whether:
    \begin{itemize}
        \item \textup{(Yes)} $\dH{\ketbra{\phi}{\phi}}{\ketbra{\psi}{\psi}}{} \geq a\rbra{n}$;
        \item \textup{(No)} $\dH{\ketbra{\phi}{\phi}}{\ketbra{\psi}{\psi}}{} \leq b\rbra{n}$.
    \end{itemize}     
    \end{enumerate}
\end{definition}

Now we prove the $\BQP$-hardness of $\TsaPureQSD_\alpha$ and $\HellPureQSD$.

\begin{lemma}[$\BQP$-hardness of $\TsaPureQSD_{\alpha}$]\label{thm:Tsa_pure_BQP_hard}
    Let $\alpha \in \rbra{0, 1}$ be a constant. 
    Let $a\rbra{n}$ and $b\rbra{n}$ be efficiently computable functions such that $a(n) \leq \frac{\rbra{1-2^{-n-1}}^{2}}{1-\alpha}$ and $b(n) \geq  \frac{2^{-2n-2}}{1-\alpha}$ for sufficiently large $n$.
    Then, $\TsaPureQSD_{\alpha}[a,b]$ is $\BQP$-hard.
\end{lemma}

\begin{proof}
    Given pure states $\ket{\phi}$ and $\ket{\psi}$, 
    we have
    \[
    \dTsa{\alpha}{\ketbra{\psi}{\psi}}{\ketbra{\phi}{\phi}}{}
    = \frac{\dtr{\ketbra{\psi}{\psi}}{\ketbra{\phi}{\phi}}{2}  }{1-\alpha}.
    \]
    By~\cref{lem:BQP-hard}, we know that $\PureQSD\sbra{1-2^{-n-1},2^{-n-1}}$ is $\BQP$-hard. Therefore, we obtain
    $\TsaPureQSD_{\alpha}[a,b]$ is $\BQP$-hard
    if $a(n) \leq \frac{\rbra{1-2^{-n-1}}^{2}}{1-\alpha}$, $b(n) \geq  \frac{2^{-2n-2}}{1-\alpha}$. 
\end{proof}

Now we can prove the $\BQP$-completeness of $\TsaLowQSD_{\alpha}$. 

\begin{lemma}[$\BQP$-containment of $\TsaLowQSD_{\alpha}$] \label{lemma:bqp-contain-TsaLowQSD}
    Let $\alpha \in \rbra{0, 1}$ be a constant. 
    Let $a\rbra{n}$ and $b\rbra{n}$ be efficiently computable functions such that $0 \leq b\rbra{n} < a\rbra{n} \leq 1$ and $a\rbra{n} - b\rbra{n} \geq \frac{1}{\poly\rbra{n}}$. 
    Then, $\TsaLowQSD_{\alpha}[a,b]$ is in $\BQP$.
\end{lemma}
\begin{proof}
    Let $\varepsilon = \rbra{a\rbra{n}-b\rbra{n}}/4$. 
    Let $x$ be an estimate of $\dTsa{\alpha}{\rho}{\sigma}{}$ within additive error $\varepsilon$ obtained by the algorithm specified in \cref{thm:Tsallis-relative-query}. 
    Then, with probability at least $2/3$, $\abs{x - \dTsa{\alpha}{\rho}{\sigma}{}} \leq \varepsilon$. 
    It can be seen that $x$ can be obtained in quantum time $\widetilde{O}\rbra{r^{1.5}/\varepsilon^4} = \poly\rbra{n}$. 
    To decide whether $\dTsa{\alpha}{\rho}{\sigma}{} \geq a\rbra{n}$ or $\dTsa{\alpha}{\rho}{\sigma}{} \leq b\rbra{n}$ is as follows. 
    \begin{itemize}
        \item If $x > \rbra{a\rbra{n}+b\rbra{n}}/2$, then return the case of $\dTsa{\alpha}{\rho}{\sigma}{} \geq a\rbra{n}$. 
        \item Otherwise, return the case of $\dTsa{\alpha}{\rho}{\sigma}{} \leq b\rbra{n}$.
    \end{itemize}
    It can be seen that this polynomial-time quantum algorithm solves $\TsaLowQSD_{\alpha}[a,b]$ and thus it is in $\BQP$. 
\end{proof}

% \section*{Acknowledgment}
% The work of M.G.\ was supported by the National Key R\&D Program of China under Grant No.\ 2023YFA1009403.
% The work of J.B.\ and Q.W.\ was supported by the Engineering and Physical Sciences Research Council under Grant EP/X026167/1.

\addcontentsline{toc}{section}{References}

\bibliographystyle{alphaurl}
\bibliography{main}

\appendix

\newpage

\section{Quantum Multi-Samplizer}
\label{sec:multi-samplizer}

To provide an implementation of the quantum multi-samplizer, we need the following lemma. 

\begin{lemma}[{\cite[Lemma 2.21]{WZ25}}] \label{lemma:BE-states}
    For every $\delta \in \rbra{0, 1}$, we can approximately implement (the controlled version of) a unitary operator $U$ and its inverse $U^\dag$ in diamond norm distance $\delta$ using $O\rbra{\frac{1}{\delta}\log^2\rbra{\frac{1}{\delta}}}$ samples of an $n$-qubit quantum state $\rho$ and $O\rbra{\frac{n}{\delta}\log^2\rbra{\frac{1}{\delta}}}$ two-qubit gates such that $U$ is a $\rbra{2, 4, 0}$-block-encoding of $\rho$.
\end{lemma}

We prove \cref{thm:multi-samplizer} as follows. 

\begin{proof}[Proof of \cref{thm:multi-samplizer}]
    The construction generalizes those in \cite{WZ25,WZ24b,WZ24a}.
    Suppose that
    \[
        \mathcal{A}^{U_1, U_2, \dots, U_k} = G_Q V_Q \dots G_2 V_2 G_1 V_1 G_0,
    \]
    where each of $V_1, V_2, \dots, V_Q$ is either (controlled-)$U_j$ or (controlled-)$U_j^\dag$ for some $1 \leq j \leq k$, and each of $G_0, G_1, \dots, G_Q$ is a unitary operator independent of $U_1, U_2, \dots, U_k$. 

    Let $\varepsilon = \delta/Q$. 
    By \cref{lemma:BE-states}, for each $1 \leq j \leq k$, we can approximately implement (the controlled version) of a unitary $U_{\rho_j}$ and its inverse in diamond norm distance $\varepsilon$, using $O\rbra{\frac{1}{\varepsilon}\log^2\rbra{\frac{1}{\varepsilon}}}$ samples of $\rho_j$ and $O\rbra{\frac{n}{\varepsilon}\log^2\rbra{\frac{1}{\varepsilon}}}$ two-qubit gates, such that $U_{\rho_j}$ is a $\rbra{2,4,0}$-block-encoding of $\rho_j$. 
    Therefore, for $1 \leq q \leq Q$, if $V_q$ is (controlled-)$U_j$ or (controlled-)$U_j^\dag$ for some $1 \leq j \leq k$, we can implement a quantum channel $\mathcal{E}_q$ such that $\Abs{\mathcal{E}_q - V_q \rbra{\cdot} V_q^\dag}_\diamond \leq \varepsilon$, using $O\rbra{\frac{1}{\varepsilon}\log^2\rbra{\frac{1}{\varepsilon}}}$ samples of $\rho_j$ and $O\rbra{\frac{n}{\varepsilon}\log^2\rbra{\frac{1}{\varepsilon}}}$ two-qubit gates.
    Then, consider the quantum channel:
    \[
        \mathcal{F} = \mathcal{G}_Q \circ \mathcal{E}_Q \circ \dots \circ \mathcal{G}_2 \circ \mathcal{E}_2 \circ \mathcal{G}_1 \circ \mathcal{E}_1 \circ \mathcal{G}_0,
    \]
    where $\mathcal{G}_q \colon \sigma \mapsto G_q \sigma G_q^\dag$ for each $0 \leq q \leq Q$. 
    Then, it can be verified that
    \[
        \Abs*{\mathcal{F} - \mathcal{A}^{U_1, U_2, \dots, U_k}}_\diamond \leq \sum_{q=1}^Q \Abs*{\mathcal{E}_q - V_q\rbra{\cdot}V_q^\dag}_\diamond \leq Q\varepsilon = \delta. 
    \]
    Therefore, the construction of $\mathcal{F}$ is a $k$-samplizer. 
    
    Moreover, if there are $Q_j$ queries to $U_j$ among $V_1, V_2, \dots, V_Q$, then the implementation of $\mathcal{F}$ uses
    \[
        Q_j \cdot O\rbra*{\frac{1}{\varepsilon}\log^2\rbra*{\frac{1}{\varepsilon}}} = O\rbra*{\frac{Q_jQ}{\delta}\log^2\rbra*{\frac{Q}{\delta}}}
    \]
    samples of $\rho_j$ and 
    \[
        Q_j \cdot O\rbra*{\frac{n}{\varepsilon}\log^2\rbra*{\frac{1}{\varepsilon}}} = O\rbra*{\frac{Q_jQn}{\delta}\log^2\rbra*{\frac{Q}{\delta}}}
    \]
    additional two-qubit gates for each $j$.
    In summary, $\mathcal{F}$ can be implemented using 
    \[
        \sum_{j=1}^k O\rbra*{\frac{Q_jQn}{\delta}\log^2\rbra*{\frac{Q}{\delta}}} = O\rbra*{\frac{Q^2n}{\delta}\log^2\rbra*{\frac{Q}{\delta}}}
    \]
    additional one- and two-qubit gates.
\end{proof}

\end{document}